\documentclass[11pt]{article}
\usepackage{amsmath,amsfonts}
\usepackage{amsthm,amssymb}
\usepackage{natbib}
\usepackage{booktabs}
\usepackage{graphicx}
\usepackage{subcaption}

\usepackage{epstopdf}
\newtheorem{theorem}{Theorem}[section]

\newtheorem{corollary}[theorem]{Corollary}
\newtheorem{lemma}[theorem]{Lemma}

\newtheorem{claim}{Claim}

\newtheorem{remark}{Remark}

\usepackage{algpseudocode}
\usepackage{algorithm}

\usepackage{enumitem}
\usepackage{bm}
\usepackage{hyperref}
\usepackage{filecontents}
\usepackage{amsmath} 
\usepackage{xcolor}

\begin{document}



\title{\textit{Nonparametric Estimation of Isotropic Covariance Function}}



{
	\title{\bf Nonparametric Estimation of Isotropic Covariance Function%
		\thanks{Accepted at the Journal of Nonparametric Statistics.}}
	
	\author{Yiming Wang and Sujit K. Ghosh\thanks{
			Yiming Wang is a graduate student (E-mail: wangyiming859@gmail.com);
			Sujit K. Ghosh is a Professor (E-mail: sujit\_ghosh@ncsu.edu),
			Department of Statistics, North Carolina State University, Raleigh, NC 27695.
		}
		\hspace{.2cm}\\
		Department of Statistics, North Carolina State University, Raleigh, NC 27695}
	
	\date{}
	\maketitle
}


\begin{abstract}
	A nonparametric model using a sequence of Bernstein polynomials is constructed to approximate arbitrary isotropic covariance functions valid in $\mathbb{R}^\infty$ and related approximation properties are investigated using the popular $L_{\infty}$ norm and $L_2$ norms. A computationally efficient sieve maximum likelihood (sML) estimation is then developed to nonparametrically estimate the unknown isotropic covaraince function valid in $\mathbb{R}^\infty$. Consistency of the proposed sieve ML estimator is established under increasing domain regime. The proposed methodology is compared numerically with couple of existing nonparametric as well as with commonly used parametric methods. Numerical results based on simulated data show that our approach outperforms the parametric methods in reducing bias due to model misspecification and also the nonparametric methods in terms of having significantly lower values of expected $L_{\infty}$ and $L_2$ norms. Application to precipitation data is illustrated to showcase a real case study.  Additional technical details and numerical illustrations are also made available.
\end{abstract}

\noindent%
\vfill


\section{Introduction}
\label{sec:introduction}
In geostatistics, observations are taken over a set of spatially varying locations and hence observations that are located in adjacent regions are expected to influence each other. Observations at nearer regions are thus more correlated with each other than those located far apart. Such spatially varying correlated data are often measured by covariance or correlation functions between two observations measured at two locations. Let $Y(\bm{s})$ denote the observation measured at location $\bm{s}\in{\cal S}\subseteq\mathbb{R}^d$, where ${\cal S}$ denotes the entire spatial region of interest. In this paper we are mostly interested in regions when $d=2$ or $d=3$, but the methodology proposed in this paper is valid for any dimension $d\geq 1$. Assuming that $E[Y^2(\bm{s})]<\infty$, the covariance function is defined as $\tilde{C}(\bm{s},\bm{s}^\prime)=Cov(Y(\bm{s}), Y(\bm{s}^\prime))$. Using the fact that for any finite collection of spatial points $\{\bm{s}_1,\ldots,\bm{s}_n\}$ in ${\cal S}$ and any arbitrary fixed vector $\mathbf{a}^T=(a_1,\ldots,a_n)\in\mathbb{R}^n$ the $Var(\mathbf{a}^T\mathbf{Y})\geq 0$ where $\mathbf{Y}^T=(Y(\bm{s}_1),\ldots,Y(\bm{s}_n))$,  any covariance function has to be positive definite satisfying
\begin{equation*}
\sum_{i=1}^n\sum_{j=1}^n a_i a_j \tilde{C}(\bm{s}_i,\bm{s}_j)\geq 0.
\end{equation*}
The goal is often to estimate the underlying covariance function based on observing a finite collection of measurements $\{Y(\bm{s}): \bm{s}\in \{\bm{s}_1,\ldots,\bm{s}_n\}\}$. However, such a task can be very difficult without any additional simplifying assumptions on the structure of the covariance function. In practice, it is often the case that covariance function depends only on the distance between two locations and hence we can then assume that $\tilde{C}(\bm{s}, \bm{s}^\prime)=C(|| \bm{s}-\bm{s}^\prime||_2)$, where $|| \bm{s}-\bm{s}^\prime ||_2$ is the Euclidean distance between locations $\bm{s}$ and $\bm{s}^\prime$.  In such case, our goal is to estimate the {\em isotropic covariance function} $C(h)$ for $h\geq 0$. Notice that this implies that $C(0)=Var[Y(\bm{s})]$ is constant for any $\bm{s}\in{\cal S}$ and often then the focus is the estimation of correlation function $R(h)=C(h)/C(0)$ and constant variance $\sigma^2=C(0)$. Various classes of isotropic parametric functions for $C(h)$ are available in literature (\citep{Cressie1993, Stein1999, Wang2022}).

One reasonable assumption is stationary covariance in the sense that covariance only depends on the difference of coordinates. i.e., $\tilde{C}(\bm{s}, \bm{s}^\prime)=C(\bm{s}-\bm{s}^\prime)$, for any two locations $\bm{s}$ and $\bm{s}^\prime$. Then resorting to the spectral representation stated in Bochner's theorem, one can characterize the stationary covariance functions in $\mathbb{R}^d$ as
\begin{equation}
C(\bm{s}-\bm{s}^\prime)=\int_{\mathbb{R}^d} exp(i\bm{r}^T (\bm{s}-\bm{s}^\prime))dG(\bm{r}),\label{eq:stationary cov}
\end{equation}
where the positive finite measure $G(\bm{r})$ is the stationary spectral distribution function.
Specially, in practice, it is often the case that covariance function is stationary and isotropic, that is, the covariance only depends on the distance between two locations and hence we can then assume that $\tilde{C}(\bm{s}, \bm{s}^\prime)=C(||\bm{s}-\bm{s}^\prime||_2)$, where $||\bm{s}-\bm{s}^\prime||_2$ is the Euclidean distance.  In such case, our goal is to estimate the {\em stationary and isotropic covariance function} $C(h)$ for $h\geq 0$. Notice that this implies that $C(0)=Var[Y(\bm{s})]$ is constant for any $\bm{s}\in{\cal S}$ and often then the focus is to conduct estimation of correlation function $R(h)=C(h)/C(0)$ and constant variance $\sigma^2=C(0)$. Various classes of isotropic parametric functions for $C(h)$ are available in literature (\citep{Cressie1993, Stein1999}).

For isotropic covariance function defined on $\mathbb{R}^d$, under some mild regularity conditions, it has been established based on (\ref{eq:stationary cov}) that any such function has the useful representation (\citep{Abramowitz1965}, p355): 
\begin{equation}
C(h) = \int_{0}^{\infty}\Omega_d(rh) dF_d(r), \label{eq: Bochner isotropic eq}
\end{equation}
where $\Omega_d(t)=(\frac{2}{t})^{\frac{d-2}{2}}\Gamma(\frac{d}{2})J_{\frac{d-2}{2}}(t)$, $\Gamma(\cdot)$ is the Gamma function, $J_v(\cdot)$ is the Bessel function of the first kind of order $v$ and $F_d(r)$ is a bounded nondecreasing positive measure on $\mathbb{R}$ 
satisfying $F_d(r)=\int_{r\leq || \bm{r} ||_2}dG(\bm{r})$ for $\bm{r}\in\mathbb{R}^d$. 
Moreover, using the fact that $\Omega_d(t)\rightarrow e^{-t^2}$ as $d\rightarrow\infty$, \citep{Schoenberg1938} shows that continuous isotropic covariance functions in $\mathbb{R}^\infty$ can be represented by
\begin{equation}
C(h) = \int_0^\infty exp(-r^2h^2) dF(r) \label{eq: Bochner isotropic Inf eq},
\end{equation}
where $F(r)$ is a bounded non-decreasing positive measure for $r\geq 0$. One can call $F_d(r)$ and $F(r)$, the isotropic spectral measure of $\mathbb{R}^d$ and $\mathbb{R}^{\infty}$, respectively, which we call simply the spectral measure for short hereafter.
Note that continuous isotropic covariance function in $\mathbb{R}^\infty$ with the form (\ref{eq: Bochner isotropic Inf eq}) can also be written as (\ref{eq: Bochner isotropic eq}) since positive definiteness in higher dimensions lead to positive definiteness in lower dimension (\citep{Schoenberg1938}). Thus, in order to estimate the isotropic covariance function $C(h)$, it is equivalent to estimate the spectral density $F(r)$ and our primary focus is to conduct nonparametric estimation using a collection of observed data points  $\{Y(\bm{s}): \bm{s}\in \{\bm{s}_1,\ldots,\bm{s}_n\}\}$ or identically distributed independent copies of spatial data,(e.g., annual precipitations at a selected set of spatial locations) and the primary goal is to explore the spatial correlations only.
There has been extensive work using non-parametric methods for covariance estimation based on the spectral representation.  
In the time series literature, approaches based on smoothing periodograms have been proposed and are more easily implemented for regularly gridded data. \citep{Choudhuri2004} proposed a nonparametric Bayesian method by constructing a prior for spectral density in equation (\ref{eq:stationary cov}) for $\mathbb{R}$ through one-dimension Bernstein polynomials and obtaining pseudoposterior distribution with aid of Whittle likelihood.
As an extension, a similar approach for spectral density estimation in $\mathbb{R}^d\;(d\geq 1)$ is developed by \citep{Zheng2010}. 
Based on irregularly gridded and continuous-indexed data, 
\citep{Reich2012} further generalized this idea by opting for Dirichlet process priors and Dirichlet process mixture priors to model the spectral density in (\ref{eq:stationary cov}) for $\mathbb{R}^2$.
\citep{Im2007} considered a semiparametric model, that is a linear combination of B-splines up to a cutoff frequency and a truncated algebraic trail,  for spectral densities in (\ref{eq: Bochner isotropic eq}) and derived the corresponding isotropic covariance functions in $\mathbb{R}^2$ analytically using Hankel transform.
Relatively more recently, \citep{Huang2011} formulated a nonparametric method for variogram estimation in $\mathbb{R}^d$ via equation (\ref{eq: Bochner isotropic eq}) using a general spline methodology with application in regularly-gridded data.
\citep{Choi2013} modelled the spectral densities for $\mathbb{R}^\infty$ by way of equation (\ref{eq: Bochner isotropic Inf eq}) using B-spline leading to an approximation of completely monotone functions computed through recursion formulas and devised a weighted least square estimator through an empirical covariagram estimate with its efficacy corroborated through simulation studies based on regularly-gridded data. 


Typically the aforementioned methods suffer from cumbersome calculation of approximated covariance functions or limited applicability when dealing with irregular observations.
Furthermore, the theoretical framework for some of the above mentioned models is not explicitly explored. Therefore, it is of interest to develop simpler nonparametric models with pratical flexibility that comes with some large sample theoretical properties.

As a result, we construct a sieve of isotropic covariance functions valid in $\mathbb{R}^\infty$ spanned by a novel class of basis functions $A_{k,m}(h)=\prod_{j=k}^m (1+\frac{h^2}{j})^{-1}$ for $1\leq k\leq m$ and $k,m\in\mathbf{N}^+$. 
We show that linear combinations of these basis functions are guaranteed to produce isotropic covariance functions valid for any dimension $d\geq 1$ through the representation (\ref{eq: Bochner isotropic Inf eq}) under mild conditions. Owing to the property of linear gaussian covariance models (\citep{Anderson1970, Piotr2017}) the coefficients of the linear combinations of the basis functions can be efficiently estimated by maximizing the sieve likelihood function (for any fixed $m>1$) and assuming that the observations are generated from a Gaussian process. 
Admittedly our idea of approximating isotropic covariance functions in $\mathbb{R}^\infty$ by employing a transformation of (\ref{eq: Bochner isotropic Inf eq}) is similar in spirit to (\citep{Choi2013}), our model, based on a novel class of basis functions, gives rise to an explicit form of covariance functions in the sense that neither recursion formula nor complicated analytic expression are required. Additionally, through the use of sieve likelihood, we take fully into account  all the correlations between observations and our methodology is well applicable to irregularly-spaced data.  Furthermore, we establish the consistency of the sieve estimator as $m\rightarrow\infty$ with certain rate of increasing sample (domain) size $n\rightarrow\infty$.

The rest of the paper is organized as follows. In Section \ref{sec:approximation method}, we describe our model and approximation theory 
using suitable norms as well as numerical illustrations of the approximation properties for finite m. Section \ref{sec: estimation method} presents the estimation method based on observed data and related asymptotic theory for large sample consistency.
In Section \ref{sec:simulation} a simulation study is conducted to compare our method with parametric methods together with couple of alternative nonparametric approaches proposed by \citep{Huang2011} and \citep{Choi2013} using several global metrics, followed by a real data example. 
Section \ref{sec:conclusion} summarizes the paper and discusses possible future work. Proofs of the theoretical results along with additional figures and tables are provided in the Appendix. Related R codes are accessible in a publicly shared repository at \url{github.com/whimwang/}.

\section{Models for Isotropic Covariance Functions}
\label{sec:approximation method}
Assume $F(r)$ in (\ref{eq: Bochner isotropic Inf eq}) permits a continuous density $f(r)$ with respect to the Lebesgue measure which we denote by $f(r)dr=dF(r)$. Then the corresponding isotropic covariance function $C$ has the representation
\begin{equation}
C(h) = \int_0^\infty exp(-r^2h^2) f(r)dr \label{eq: Bochner isotropic Inf eq new main}.	
\end{equation}

\noindent First, notice that a `good' approximation of the spectrum $f$ leads to a `good' approximation of the covariance $C$, which can be readily implied from the following simple observation. Let $C_1$ and $C_2$ be two covariance functions corresponding to the two spectrum $f_1$ and $f_2$, respectively linked by the equation (\ref{eq: Bochner isotropic Inf eq new main}). Then assuming 
both $f_1$ and $f_2$ are finitely integrable, one can derive 
\begin{equation*}
||C_1 -C_2||_{\infty}\leq ||f_1-f_2||_1,
\end{equation*}
where $||G||_\alpha :=(\int |G(r)|^\alpha dr)^{{1\over \alpha}}$ and $||G||_\infty := \sup_r |G(r)|$ denote the $L_\alpha$ norm for $\alpha\geq 1$ and $L_\infty$ of the function $G(\cdot)$.
Thus, we may focus on approximating the spectrum using suitable norms. Instead of approximating $f(r)$ on the whole real line, it is sometimes more convenient to transform to an equivalent function defined on a bounded interval, for which extensive approximation theory can be readily adapted.

Letting $s=\exp\{-r^2\}$ in the equation (\ref{eq: Bochner isotropic Inf eq new main}) we can obtain an equivalent expression given by
\begin{equation}
C(h)=\int_0^1 s^{h^2}g(s)ds
\label{eq: g to C},
\end{equation}
where $g(s)$ takes the form as
\begin{equation}
g(s)=\frac{1}{2s\sqrt{-log s}}f(\sqrt{-log s}).
\label{eq: f to g}
\end{equation}
and satisfies $\int_0^1g(s)ds = \int_0^\infty f(r)dr<\infty$. And we can also express $f$ using $g$ equivalently via the inverse relation
\begin{equation}
f(r)=2rexp(-r^2)g(exp(-r^2)).
\label{eq: g to f}
\end{equation}
Due to the bijection relationship between $g(s)$ and $f(r)$, we need only focus on approximation of $g(s)$ which eventually leads to an approximation for $C$ by equations (\ref{eq: g to C})(\ref{eq: f to g}) and (\ref{eq: g to f}).

Bernstein polynomials serve as a remarkable class of functions in approximation theory and comes with well-known simple, efficient recursive algorithms and excellent numerical stability properties (see \citep{Farouki2012} for a comprehensive review),
originally introduced by \citep{Bernstein1912} to facilitate a constructive proof of the  Weierstrass Approximation Theorem (1885) that if $g(x)$ is a continuous function in the interval $[0,1]$, for any $\epsilon>0$, it is always possible to find a polynomial $g_m(x)$ of degree m (depending on $\epsilon$ only), such that $||g-g_m||_{\infty}\equiv\underset{x\in[0,1]}{sup}|g(x)-g_m(x)|\leq \epsilon$.
Thus, assuming that unknown function $g$ is bounded, for any $\epsilon>0$,  $g(s)$ in $[0,1]$  can be uniformly approximated by a sequence of Bernstein polynomials $g_m(s)$ given by
\begin{equation}
g_m(s)=\sum_{k=1}^m w_k {m-1\choose k-1} s^{k-1} (1-s)^{m-k}, w_k\geq 0, k=1,\cdots, m;
\label{eq:form of g_m}
\end{equation}
where if we choose $w_k=g(\frac{k-1}{m-1})$ Bernstein showed that there exists a sufficently large m 
satisfying
\begin{equation}
||g_m-g||_\infty < \epsilon. \label{eq: g_m uniform converge to g}
\end{equation}
Substituting $g_m(s)$ in (\ref{eq: g to f}),  we can obtain approximated spectrum 
\begin{equation*}
f_m(r)=2rexp(-r^2)\sum_{k=1}^m  \{w_k {m-1\choose k-1} exp(-r^2(k-1))(1-exp(-r^2))^{m-k}\}
\end{equation*} 
and it follows from (\ref{eq: g to C}) and  (\ref{eq:form of g_m}) that the corresponding covariance approximating function $C_m(h)$ is given by 
\begin{equation}
C_m(h)= \frac{1}{m}\sum_{k=1}^m w_kA_{k,m}(h), \label{eq: C_m}
\end{equation}
where
\begin{equation*}
A_{k,m}(h)=\frac{Beta(k+h^2,m-k+1)}{Beta(k,m-k+1)}=\prod_{j=k}^m\left(1+\frac{h^2}{j}\right)^{-1}
\end{equation*}
and $Beta(a,b)=\int_0^1 u^{a-1}(1-u)^{b-1} du$ denotes the $Beta$ function for $a, b>0$.
Fig	\ref{fig:Basis functions for m=5 and m=25} presents curves of basis $\{A_{k,m}\}_{k=1}^m$ for $m=5$ and $m=25$.

Recall that the approximation result in (\ref{eq: g_m uniform converge to g}) is based on the assumption that $g(s)$ is bounded on $(0, 1)$ in which case one can define $g(0)=\underset{s\rightarrow0}{\lim\inf}\;g(s)$ and $g(1)=\underset{s\rightarrow 1}{\lim\sup}\;g(s)$. Otherwise,
there exists no such $g_m(s)$ of the form (\ref{eq:form of g_m}) satisfying (\ref{eq: g_m uniform converge to g}) when $g(s)$ explodes to infinity at the endpoints 0 or 1. For example, when true covariance function is Generalized Cauchy in the form $C(h)=(1+h^2)^{-\frac{1}{4}}$ with its spectrum $f(r)=r^{-\frac{1}{2}}exp(-r^2)$, then
$g(s)=\frac{1}{2}(-logs)^{-\frac{3}{4}}$ satisfies $\int_0^1 g(s)ds=\frac{1}{2}\Gamma(\frac{1}{4})<\infty$ and yet $\underset{s\rightarrow1}{\lim}\;g(s)=\infty$.
Nonetheless, our model in (\ref{eq: C_m}) can still accommodate such covariance with nice approximation property.
In the ensuing paragraphs, we investigate approximation theory using several norms.	

\begin{figure}[H]
	\centering
	\begin{minipage}[b]{0.45\textwidth}
		\centering
		\includegraphics[width=\linewidth]{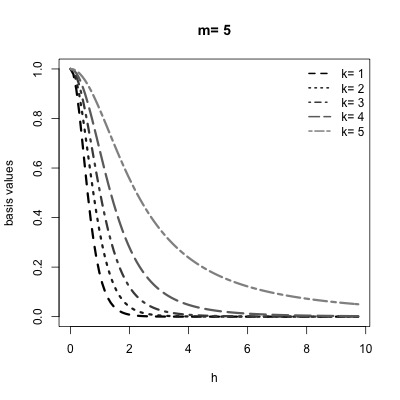}
		
		(a) $m=5$
	\end{minipage}
	\hspace{5pt}
	\begin{minipage}[b]{0.45\textwidth}
		\centering
		\includegraphics[width=\linewidth]{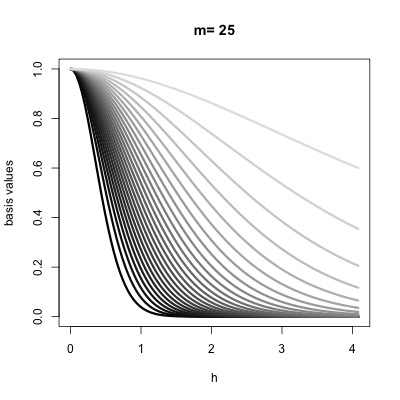}
		
		(b) $m=25$
	\end{minipage}
	\caption{Basis functions. In (b), colors for $A_{1,m},\cdots, A_{m,m}$ change from dark to light.}
	\label{fig:Basis functions for m=5 and m=25}
\end{figure}

\subsection{Approximation Theory}
First, assuming continuous spectrum we provide $L_\infty$ and $L_2$ approximation results based on a sequence of sieves:

\begin{theorem}[Approximation Theorem]
	\label{theorem: covariance approximation} 
	Let $\mathcal{A}$ be the class of continuous isotropic covariance function applicable in $\mathbb{R}^\infty$ with continuous spectrum $f(r)$ defined in (\ref{eq: Bochner isotropic Inf eq new main}). For any $m\in \mathbf{N}^+$, define a class of functions $\mathcal{A}_m\subset\mathcal{A}$ as
	\begin{equation*}
	\mathcal{A}_m = \{C_m(h):=\sum_{k=1}^m w_kA_{k,m}(h), A_{k,m}(h)= 
	\frac{Beta(k+h^2,m-k+1)}{Beta(k,m-k+1)}, w_k\geq 0, k=1,\cdots, m\}.
	\end{equation*}
	\begin{enumerate}
		\item ($L_\infty$ Approximation) For any function $C(h)\in \mathcal{A}$,
		given any $\epsilon>0$, 
		there exists a function $C_m(h) \in \mathcal{A}_m, m\in\mathbf{N}^+$ so that
		$||C-C_m||_\infty<\epsilon.$
		\item 
		($L_2$ Approximation) For any function $C(h)\in \mathcal{A}$ with $||C||_2<\infty$, given any $\epsilon>0$, 
		there exists a function $C_m(h) \in \mathcal{A}_m, m\in\mathbf{N}^+$ so that $||C-C_m||_2<\epsilon.$
	\end{enumerate}
\end{theorem}
\begin{remark}\label{remark:L_q approximation}
	($L_q$ Approximation) In fact, for any functions $C\in\mathcal{A}$ with $||C||_q<\infty$ and $q>\frac{1}{2}$, given any $\epsilon>0$, there exists a function $C_m(h)\in \mathcal{A}_m$, $m\in \mathbf{N}^+$ so that $||C - C_m||_q < \epsilon$.
\end{remark}
\begin{remark}\label{remark:example} Covariance functions $C(h)$ with unbounded $g(s)$ may not belong to $L_q$ $(q\geq 1)$ space. As mentioned in Section \ref{sec:introduction}, admittedly for Generalised Cauchy $C(h)=\frac{1}{(1+h^2)^{\frac{1}{4}}}$ with $g(s)=\frac{1}{2}(-logs)^{-\frac{3}{4}}$ satisfying $||C||_2 = \infty$ and $||C||_1=\infty$,  there always exists a series of functions $\{C_m\in\mathcal{A}_m\}_{m=1}^\infty$ such that $\underset{m\rightarrow\infty}{\lim}{||C-C_m||_{\infty}}=0$.
\end{remark}

Properties of the approximation class of functions $\mathcal{A}_m$ are of great interest.  Corollary \ref{corollary: basis of A_m} indicates  any function in $\mathcal{A}_m$ can be uniquely determined by coefficients. Corollary \ref{corollary: increasing A_m} indicates that such spaces of functions are nested. One can refer to Appendix \ref{sec:verify} for detailed verification.

\begin{corollary}\label{corollary: basis of A_m}
	For any $m\in \mathbf{N}^+, \{A_{k,m}\}_{k=1}^m$ forms the basis of $\mathcal{A}_m$.
\end{corollary}

\begin{corollary}\label{corollary: increasing A_m}
	For any $m\in\mathbf{N}^+$, $\mathcal{A}_m \subseteq \mathcal{A}_{m+1} \subseteq \mathcal{A}$. 
\end{corollary}

\begin{corollary}\label{corollary: compact A_m}
	For any $m\in\mathbf{N}^+$, $\mathcal{A}_m$ is closed under the $L_q$ $(q\geq 1)$ and $L_\infty$ norm. 
\end{corollary}

Thus by Theorem \ref{theorem: covariance approximation} and Corollary \ref{corollary: increasing A_m}-\ref{corollary: compact A_m}, $\{\mathcal{A}_m\}_{m=1}^\infty$ form a  sieve space (nested class of approximation spaces), a powerful structure for constructing consistent estimation. Moreover, all functions in $\mathcal{A}_m$ are themselves valid covariance functions for $m\in\mathbf{N}^+$, a property that is practically very useful. It is also well worth pointing out that any $C_m\in\mathcal{A}_m$ satisfies $C_m(0)=\sum_{k=1}^m w_k$ and thereby we can easily obtain sieve space of isotropic correlation functions:
\begin{equation*}
\mathcal{R}_m = \{R_m(h):=\sum_{k=1}^m w_k A_{k,m}(h), \sum_{k=1}^m w_k=1,w_k\geq 0, k=1,\cdots, m\}.
\end{equation*}

\subsection{Numerical Illustrations}

To demonstrate the accuracy of our model to approximate a given covariance function, we present numerical examples with several popular covariance functions.  
Given the true covariance function $C_0\in\mathcal{A}$, there are a number of ways to assess  the global performance of its approximation errors.  Here we consider $L_2$ norm for ease of computation.

Without loss of generality, we assume $C_0(0)=1$. The approximation functions can be selected from the sieve space $\mathcal{R}_m\;(m\in\mathbb{N^+})$ which inherits the properties in Theorem \ref{theorem: covariance approximation}. For $C_{m}(h;\mathbf{w}_m)=\sum_{k=1}^m w_kA_{k,m}(h)$ with $\mathbf{w}_m=(w_1,\cdots,w_m)^T$,
we employ numerical methods to obtain the optimal weight $\mathbf{w}_m$ that minimizes the approximation error
\begin{equation}
||C_0-C_{m}||_2^2= \int_0^\infty\left[\sum_{k=1}^m w_kA_{k,m}(h)-C_0(h)\right]^2dh, \label{eq:approximation_error_with_w}
\end{equation}
subject to $\mathbf{w}_m\in \mathcal{S}_m:=\{\mathbf{w}:\sum_{k=1}^{m}w_k=1, w_k\geq 0,  k=1,\cdots, m\}$, a m-dimensional simplex.
This problem is akin to so-called positive lasso (\citep{Efron2004}), offering us to efficiently compute a sparse solution. (\ref{eq:approximation_error_with_w}) can be further rewritten as
\begin{align}
||C_0-C_{m}||_2^2 
&=\textbf{w}_m^T \textbf{M} \textbf{w}_m ,\label{eq: quadratic programming problem}
\end{align}
where the symmetric positive definite matrix $\mathbf{M}$ has $(i,j)$ entry $M_{[ij]}=\int_0^\infty (A_{i,m}(h)-C_0(h))\cdot(A_{j,m}(h)-C_0(h))dh
$ obtained numerically by well-established quadrature methods. Accordingly, the above optimization problem can be formulated as a convex quadratic programming problem with a unique solution.
R package {\tt quadprog} is used to compute the optimal weight parameter  for any given $m$ and true covariance function $C_0$. 

Theoretically, 
larger m engenders better approximation but for practical purpose, we set a desired threshold to bound the approximation error to determine the least possible m so that $\mathbf{M}$ remains numerically positive definite. When $C_0$ has finite $L_2$ norm, relative approximation error (RAE) between the true and approximation $\hat{C}_0(h)=C_m(h;\hat{\mathbf{w}}_m)$ defined as
\begin{equation*}
RAE =\frac{||C_0 - \hat{C}_0 ||_2}{|| C_0 ||_2}
\end{equation*}
serves as a sensible choice, accounting for the magnitude of $||C_0||_2$ when assessing the approximation performance of different parametric covariance families. 
In this manner, we select the simplest model with RAE smaller than threshold $\epsilon=0.05$ from a sequence of approximation functions for different values of m.
We implement this procedure for commonly used parametric covariance families.
\begin{enumerate}
	\item Gaussian(1,1) where $C_{Gaussian}(h;\sigma^2,\rho) = \sigma^2 exp(-(\frac{h}{\rho})^2)$; 
	\item Matern(1,1,1) where $C_{Matern}(h;\sigma^2,\rho,\nu) =\sigma^2 \frac{1}{2^{\nu-1} \Gamma(\nu)} (\frac{h}{\rho})^{\nu}  K_{\nu}(\frac{h}{\rho})$; 
	\item Cauchy(1,1) where $C_{Cauchy}(h;\sigma^2,\rho) = \sigma^2 (1+(\frac{h}{\rho})^2)^{-\frac{1}{2}}$. 
\end{enumerate}

Our proposed sieve model can attain the desired approximation for Gaussian and Matern with only $m=4$ and $m=3$, while Cauchy covariance requires $m=461$ (see Fig \ref{fig:RAE vs m and True vs approximation covariance}), indicating that the choice of m varies considerably with the tail decay rate of the true covariance.  Note that Matern(1,1) and Cauchy(1,1,1) belong to $\mathcal{A}$ while
Gaussian(1,1) does not given the fact that $exp(-h^2)=\int_{0}^{\infty}exp(-r^2h^2)f(r)dr$ with $f(r)=\delta(r-a)$ and $\delta$ is Dirac function, yet still included within our scope to test our model. In view of the approximation performance for Gaussian, our model is rather robust to the restriction of continuity on the spectral density $f$.

\begin{figure}[H]
	\centering
	
	\begin{minipage}[b]{0.31\textwidth}
		\centering
		\includegraphics[width=\linewidth]{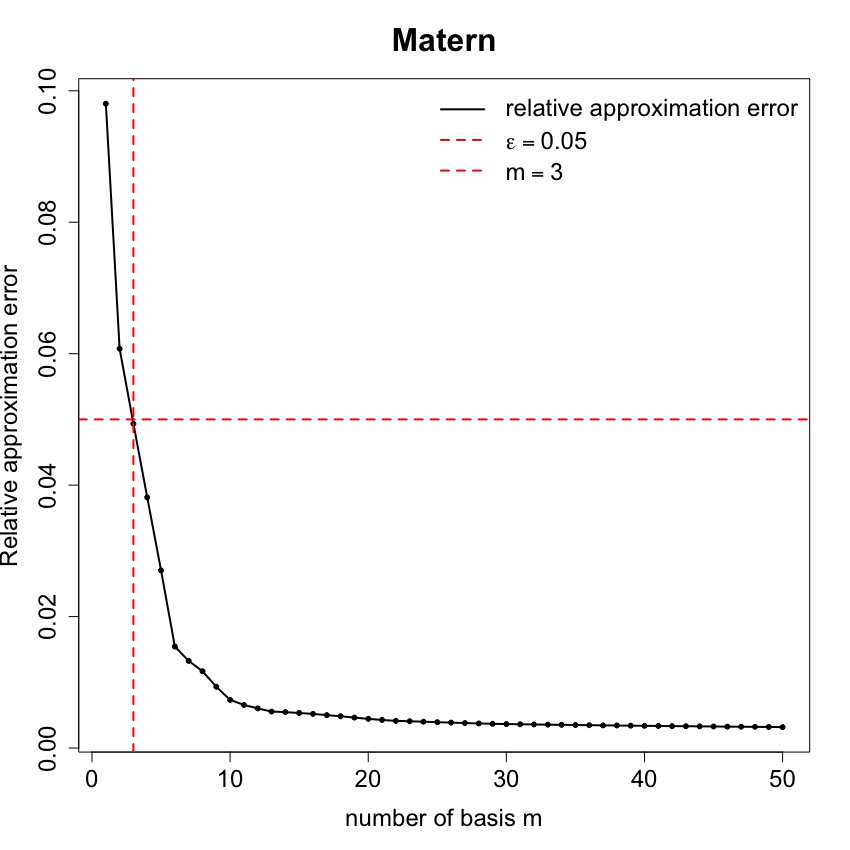}
		
		(a)
	\end{minipage}
	\hspace{1pt}
	\begin{minipage}[b]{0.31\textwidth}
		\centering
		\includegraphics[width=\linewidth]{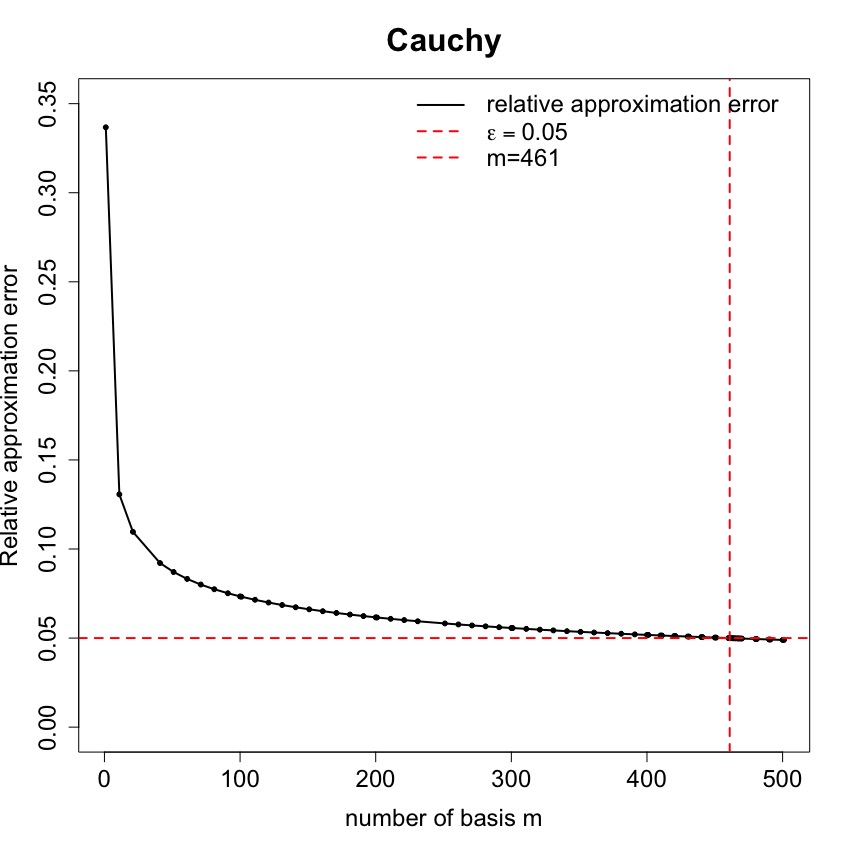}
		
		(b)
	\end{minipage}
	\hspace{1pt}
	\begin{minipage}[b]{0.31\textwidth}
		\centering
		\includegraphics[width=\linewidth]{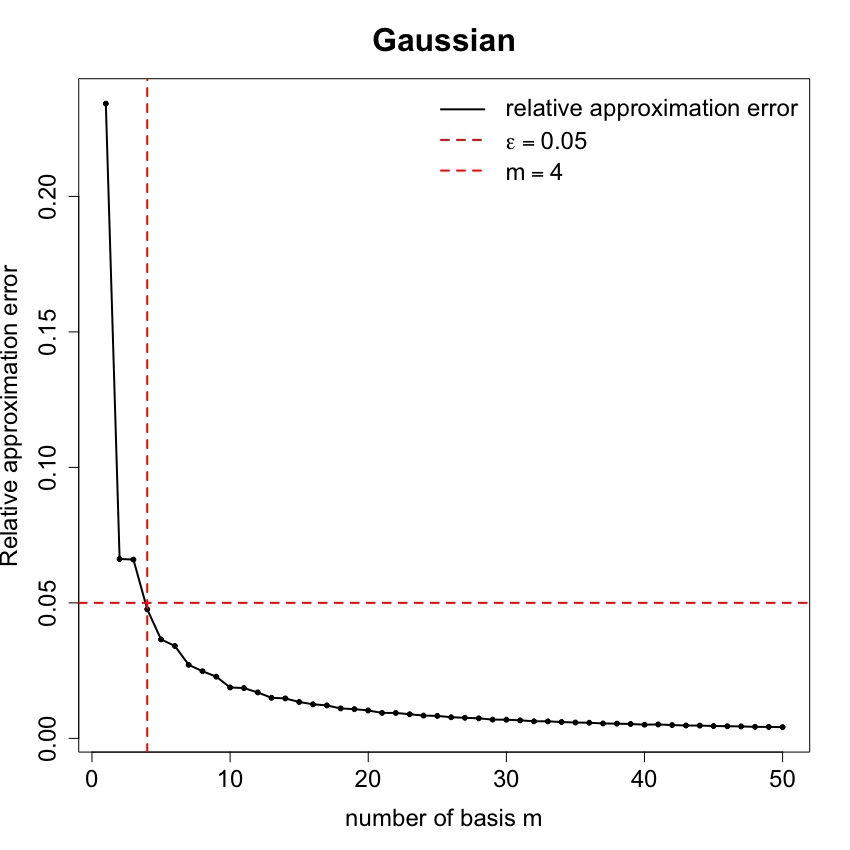}
		
		(c)
	\end{minipage}
	
	\vspace{4pt}
	
	\begin{minipage}[b]{0.31\textwidth}
		\centering
		\includegraphics[width=\linewidth]{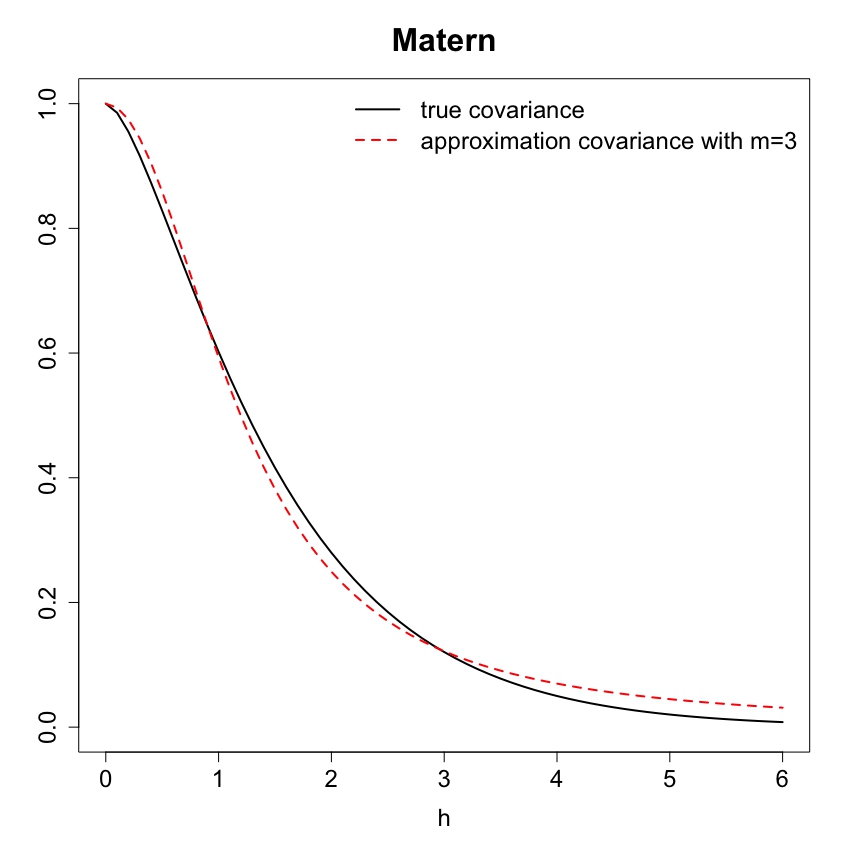}
		
		(d)
	\end{minipage}
	\hspace{1pt}
	\begin{minipage}[b]{0.31\textwidth}
		\centering
		\includegraphics[width=\linewidth]{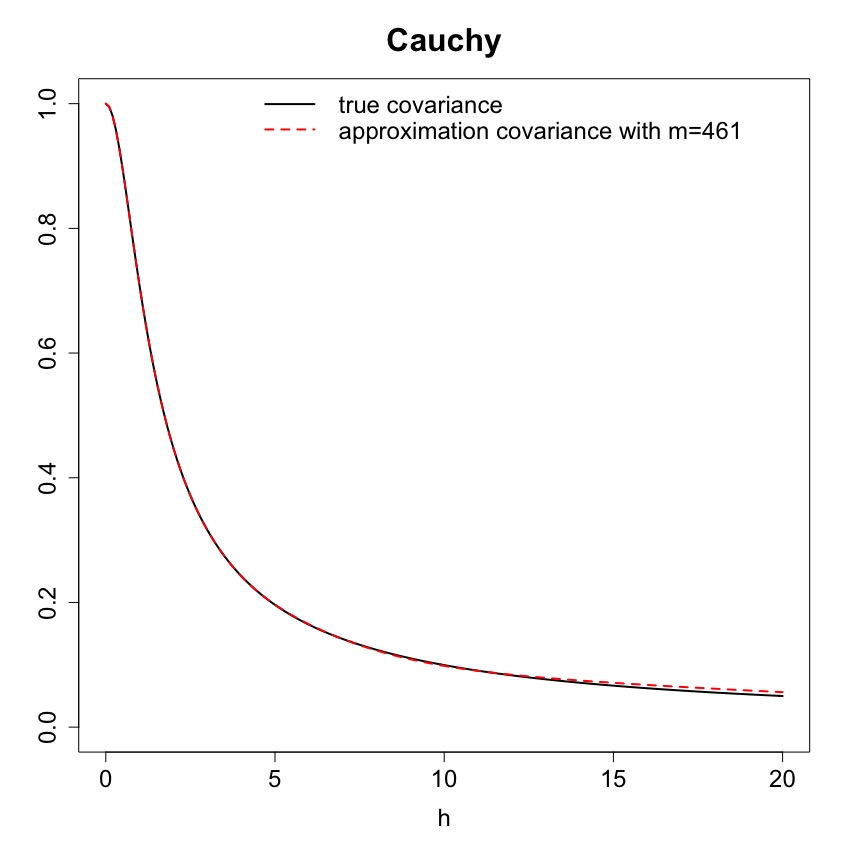}
		
		(e)
	\end{minipage}
	\hspace{1pt}
	\begin{minipage}[b]{0.31\textwidth}
		\centering
		\includegraphics[width=\linewidth]{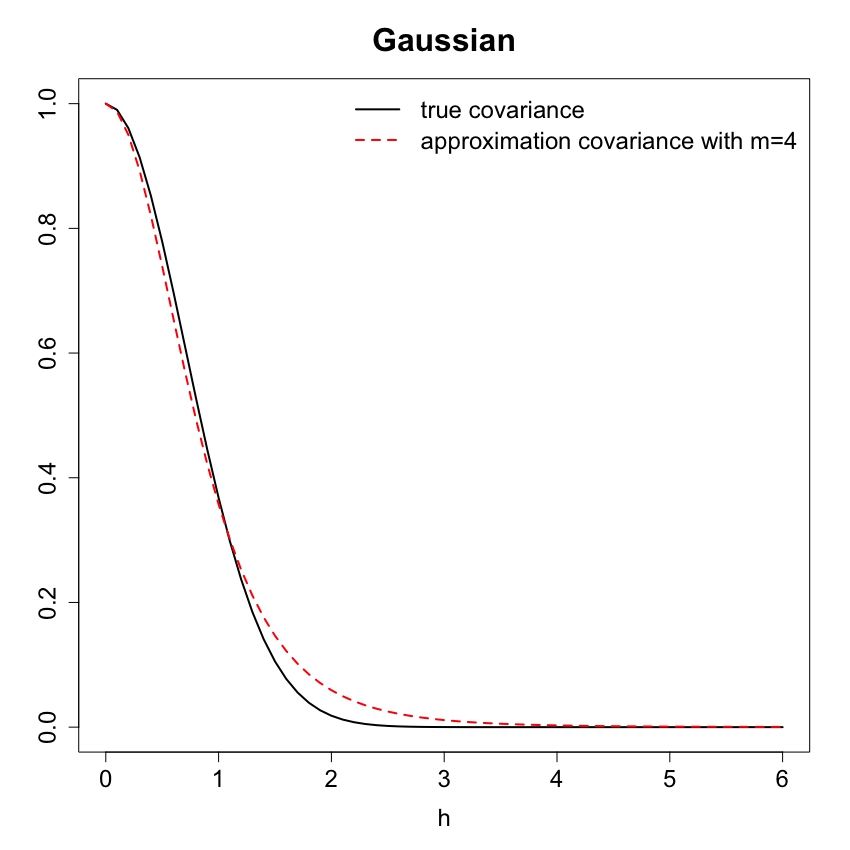}
		
		(f)
	\end{minipage}
	
	\caption{Relative approximation error vs.\ $m$ in (a)--(c), and true vs.\ approximation covariance in (d)--(f). Each row includes results for Matern, Cauchy, and Gaussian.}
		\label{fig:RAE vs m and True vs approximation covariance}
\end{figure}

\section{Sieve Maximum Likelihood Estimation Method}
\label{sec: estimation method}
Consider scenarios where observations are generated from a Gaussian process with an unknown isotropic covariance function valid in $\mathbb{R}^\infty$. We are interested in the statistical estimation of the model (e.g., the weight parameter $\mathbf{w}_m$ and the tuning parameter m) based on a set of random samples for a given sample size n.

Let $\{Y(\bm{s}), \bm{s}\in\cal{S}$ $\subset \mathbb{R}^d\}$ be Gaussian process with isotropic covariance function $C\in \mathcal{A}$ and mean $E(Y(\bm{s}))=\bm{X}^T(\bm{s})\bm{\beta}$, where $\bm{X}(\bm{s})\in \mathbb{R}^p$ is known vector valued predictor measured at location $\bm{s}$.
Unknown parameters $\bm{\beta}$ and $C$ are to be estimated based on a finite number of observations $\bm{Y}_n=(Y_1(\bm{s}_1),\cdots, Y_n(\bm{s}_n))^T$ and $\mathbf{X}_n=(\bm{X}_1(\bm{s}_1), \cdots,\bm{X}_n(\bm{s}_n))^T$  at $S_n = \{\bm{s}_1,\cdots,\bm{s}_n\}\subset \mathcal{S}$ satisfying
\begin{equation*}
\bm{Y}_n(\bm{s})=\mathbf{X}_n(\bm{s})\bm{\beta}+\bm{\epsilon}_n(\bm{s}),\;\;\;\;\bm{\epsilon}_n\sim N_n(\mathbf{0}, \mathbf{\Sigma}_{{C},n}),
\end{equation*}
where $\mathbf{\Sigma}_{C,n}$ is the covariance matrix under $C$
with the $(i,j)$ entry $\mathbf{\Sigma}_{C,n[i,j]} = C(||\bm{s}_i - \bm{s}_j ||_2)$. 
Define the parameter space $\bm{\Theta} = \{(\bm{\beta},C): \bm{\beta}\in \mathcal{B}\subset \mathbb{R}^p, C\in\bm{\Omega}\subset\mathcal{A}\} =\mathcal{B} \times \bm{\Omega}$ and the true parameter $\bm{\theta}_0=(\bm{\beta}_0,C_0)\in\mathcal{B} \times\bm{\Omega}$, where $\mathcal{B}$ and $\bm{\Omega}$ are bounded and convex sets in $\mathbb{R}^p$ and $\mathcal{A}$ respectively. (i.e., $\mathcal{B}=\{\bm{\beta}\in \mathbb{R}^p:||\bm{\beta}||_2\leq U_{\bm{\beta}}<\infty\}$ and $\bm{\Omega} = \{C\in\mathcal{A}:0< L_{C}\leq || C ||_\infty\leq U_{C}<\infty\}$ with some known lower bound $L_{C}$ and upper bounds $U_{\bm{\beta}},U_C$.   
The log-likelihood function is given by 
\begin{equation}
l_{n}(\bm{\beta},C;\bm{Y}_n)=-\frac{1}{2}logdet(\mathbf{\Sigma}_{C,n})-\frac{1}{2}(\bm{Y}_n-\mathbf{X}_n\bm{\beta})^T[\mathbf{\Sigma}_{C,n}]^{-1}(\bm{Y}_n-\mathbf{X}_n\bm{\beta}) -\frac{1}{2} n\log(2\pi), 
\label{eq:deviance}
\end{equation}
where $logdet$ is the natural logarithm of the determinant of a matrix. Notably, the function space $\bm{\Omega}\subset\mathcal{A}$ can be regarded as infinite-dimensional parameter space and it follows from Theorem \ref{theorem: covariance approximation} that $(\bm{\Omega}\cap\mathcal{A})\subseteq\overline{\cup_{m=1}^\infty(\bm{\Omega}\cap\mathcal{A}_m)}$ in the sense that any function $C\in\bm{\Omega}$ can be uniformly approximated by  $\sum_{k=1}^m w_{k,m} A_{k,m}(h)$ for non-negative weight vector $\bm{w}_m=(w_{1,m},\ldots,w_{m,m})^T$ as $m\rightarrow\infty$.

One possibility to estimate the true parameter is to maximize the log-likelihood function (\ref{eq:deviance}) over the entire infinite-dimensional space $\bm{\Theta}$ assuming that such an extremum exists. However, as pointed out by many earlier works (e.g., \citep{Amemiya1985,Newey1994}), such optimization methods are not only intractable in terms of computation but also lead to undesirable large sample properties caused by lack of structured basis. Alternatively, we employ method of sieves to estimate $\bm{\theta}_0$ which is deemed stable with satisfactory asymptotic properties. 

The concept of sieve was first introduced by \citep{Grenander1981}. Since then, the sieve method has been a powerful tool in the area of nonparametric and semiparametric statistical inference.
The intrinsic idea is to approximate an infinite-dimensional space $\bm{\Theta}$ by a series of finite-dimensional parameter space $\bm{\Theta}_{m}$ that depend on the sample size n and to estimate on the finite-dimensional space $\bm{\Theta}_m$ instead of $\bm{\Theta}$.
Mathematically, let $\bar{d}(\cdot,\cdot)$ be a pseudodistance on $\bm{\Theta}$; $L_n=L_n(\bm{\theta};\mathbf{W}_n)$ is a function of the observations $\mathbf{W}_n$ and the parameter $\bm{\theta}$, which is called the {\em empirical criterion function}. $\{\bm{\Theta}_m\}_{m=1}^\infty$ is a sequence of sieve spaces and
typically meets some conditions (\citep{Chen2007},pp 5589), which are: (i) $\bm{\Theta}_m$ are compact under $\bar{d}$; (ii) nondecreasing ($\bm{\Theta}_m\subset\bm{\Theta}_{m+1}\subset\bm{\Theta}, \forall  m\in \mathbf{N}^+$); (iii) For any $\bm{\theta}\in\bm{\Theta}$, there exists an sequence $\{\bm{\theta}_m: \bm{\theta}_m\in \bm{\Theta}_m\}_{m=1}^\infty$ such that $\bar{d}(\bm{\theta}_m,\bm{\theta})\rightarrow0$ as $m\rightarrow\infty$. 
In the language of \citep{Grenander1981}, 
the estimator $\hat{\bm{\theta}}_n$ satisfying 
\begin{align}
L_n(\hat{\bm{\theta}}_n;\mathbf{W}_n) \geq \underset{\bm{\theta}\in\bm{\Theta}_m}{\sup}{L_n(\bm{\theta};\mathbf{W}_n)}-\eta_n \label{eq:empirical criterion function}
\end{align}
is called the {\em sieve} MLE of $\bm{\theta}$ where m increases with n. The rate at which $\eta_n$ tends to 0 does not affect the asymptotic properties of $\hat{\bm{\theta}}_n$

In this manner, we construct sieve spaces $\bm{\Theta}_{m}=\mathcal{B}\times \bm{\Omega}_m,m\in\mathbf{N}^+$ and define $\bar{d}$ as $\bar{d}((\bm{\beta}_1,C_1),(\bm{\beta}_2,C_2)) = || \bm{\beta}_1 - \bm{\beta}_2 ||_2 + ||C_1 - C_2 ||_{\infty}$ for $(\bm{\beta}_i, C_i)\in \bm{\Theta}$, $i=1,2$. Consequently, such sieve spaces fullfill the aforementioned requirements because Corollary \ref{corollary: compact A_m} and \ref{corollary: increasing A_m} lead to condition (i) (ii)
and existence of such sequence meeting condition (iii) is assured by Theorem \ref{theorem: covariance approximation}. We choose $\frac{1}{n}l_n(\bm{\beta},C;\mathbf{Y}_n)$ in (\ref{eq:deviance})
as the empirical criterion function in (\ref{eq:empirical criterion function}) and $\eta_n=0$, then
\begin{align*}
\hat{\bm{\theta}}_n=\arg\underset{\bm{\theta}\in\Theta_{m(n)}}{\sup}\frac{1}{n}l_n(\bm{\theta},\bm{Y}_n),\mbox{ where $m(n)$ depends on $n$}
\end{align*}
is the sieve MLE of $\bm{\theta}_0$. 
\subsection{Asymptotic Theory}
\label{subsec: asymptotic theory}
For independent and repetitive data,  theorems on the asymptotic behaviour of sieve estimators are well-developed using the tool of empirical process (\citep{Shen1994, Van1996, Huang1997} and \citep{Xue2010}), while they can not readily be extended to 
the case when observations $\bm{Y}_n$ is a single realization from Gaussian distribution $N_n(\mathbf{X}_n\bm{\beta},\Sigma_{C_0,n})$.
Nevertheless, we are still able to establish asymptotic results of proposed sieve estimator.
It is worth mentioning that there are two standard approaches to increase the number of observations towards infinity (\citep{Cressie1993}). The first is called increasing-domain asymtotics where the domain expands in spatial extent and the sampling density remains bounded away from infinity while the second is called infill or fixed-domain asymptotics where the domain stays constant and the sampling density increases to infinity. Under fixed-domain asymptotics only a subcomponent of the covariance parameter, called the microergodic parameter,  can be estimated consistently (\citep{Zhang2004}). Given the well-known limitations of the fixed domain asymptotics, we only explore the increasing-domain asymptotics in our work.

Before we jump to the conclusion, some notations are to be clarified. Let $P_{\bm{\theta}}$ be the distribution of $\bm{Y}_n$ given $\mathbf{X}_n$ at $S_n$ under $\bm{\theta}$, $E_0$ be the expectation under $\bm{\theta}_0$. 
Let $\lambda_{max}(\mathbf{A})$ and $\lambda_{min}(\mathbf{A})$ denote, respectively, the largest and smallest eigenvalue of a positive definite matrix $\mathbf{A}$.  Let $\mathbf{\Sigma}_{C,n}$ and $\mathbf{\Sigma}_{0,n}$  be the covariance matrix under $C$ and $C_0$.
Specially $\{\mathbf{A}_{k,m,n}\}_{k=1}^m$ are matrices under covariance basis functions $\{A_{k,m}\}_{k=1}^m$. 
In addition, we refer to the following assumptions on parameter space $\bm{\Theta}$ and the sequence of domains ${S_n}$.

\begin{enumerate}[label=D\arabic*]
	\item \label{assumption: compact and convex}$\bm{\Theta}$ is a bounded and convex set. Specifically, $\mathcal{B}=\{\bm{\beta}\in \mathbb{R}^p:||\bm{\beta}||_2\leq U_{\bm{\beta}}<\infty\}$ and $\Omega = \{C\in\mathcal{A}:0< L_{C}\leq || C ||_\infty\leq U_{C}<\infty\}$ with some known lower bound $L_{C}$ and upper bounds $U_{\bm{\beta}},U_C$.
	\item $\bm{\theta}_0$ is an interior point of $\bm{\Theta}$.
	\item \label{assumption:identification of C and beta} 
	For $\bm{\theta}=(\bm{\beta},C)\in\bm{\Theta}$ and $n\in\mathbb{N}^+$, if $\bm{\theta}\neq \bm{\theta}_0$, then $\mathbf{\Sigma}_{0,n}\neq \mathbf{\Sigma}_{C,n}$ or $\mathbf{X}\bm{\beta}\neq \mathbf{X}\bm{\beta}_0$. 
	
	\item \label{assumption: X^TX} $\underset{n\rightarrow\infty}{\lim\sup}\; \lambda_{max}(\frac{1}{n}\mathbf{X}^T_n \mathbf{X}_n)<\infty$ and $\underset{n\rightarrow\infty}{\lim\inf}\;\lambda_{min}(\frac{1}{n}\mathbf{X}^T_n \mathbf{X}_n)>0$.
	
	\item \label{assumption:max and min eigen for true} 	$\underset{n\rightarrow\infty}{\lim\sup}\;\lambda_{max}(\mathbf{\Sigma}_{0,n}) <\infty$ and 
	$\underset{n\rightarrow\infty}{\lim\inf}\;\lambda_{min}(\mathbf{\Sigma}_{0,n}) >0$.
	
	\item \label{assumption:max and min eigen for A_m,m} For $m\in\mathbb{N}^+$,
	$\underset{n\rightarrow\infty}{\lim\sup}\;\lambda_{max} (\mathbf{A}_{m,m,n}) <\infty$ and	$\underset{n\rightarrow\infty}{\lim\inf}\;\lambda_{min} (\mathbf{A}_{m,m,n}) >0$.
	
	\item \label{assumption: limsup max eigen final}  $\underset{n,m(n)\rightarrow\infty}{\lim\sup}\lambda_{max}(\mathbf{A}_{m(n),m(n),n})<\infty$ for any $m$.
	
	\item \label{assumption: parameter to covariance} 
	$	\underset{\epsilon\rightarrow0}{\lim\sup}\underset{
		\begin{subarray}{c}
		n\in\mathbb{N},C\in\bm{\Omega}\\
		|| C-C_0 ||_\infty\leq \epsilon
		\end{subarray}}{\sup}
	\frac{1}{n}\sum_{i,j}^n(C(|| \bm{s}_i-\bm{s}_j ||_2)-C_0(|| \bm{s}_i-\bm{s}_j ||_2))^2 =0. 
	$
	\item \label{assumption: covariance to parameter} 	For any $\epsilon>0$,
	$	\underset{n\rightarrow\infty}{\lim\inf} \underset{
		\begin{subarray}{c}
		C\in\bm{\Omega}\\
	|| C-C_0 ||_\infty\geq \epsilon
		\end{subarray}}
	{\inf} \frac{1}{n}\sum_{i,j}^n(C(||  \bm{s}_i-\bm{s}_j ||_2)-C_0(|| \bm{s}_i-\bm{s}_j ||_2))^2>0. $
\end{enumerate}

\begin{theorem}[Consistency]\label{theorem:estimation consistency} 
	Assume conditions (\ref{assumption: compact and convex})-(\ref{assumption: covariance to parameter}). Then $\bar{d}(\hat{\bm{\theta}}_{m,n},\bm{\theta}_0)\overset{P_{\bm{\theta}_0}}{\longrightarrow} 0$.
\end{theorem}

Proof of the theorem is presented in the Appendix \ref{sec:verify}. The identification assumption (\ref{assumption:identification of C and beta}) indicates that essentially there should not exist two distinct covariance functions of which the associated covariance matrices are the same for the spatial sampling, which is clearly minimal.
(\ref{assumption: X^TX}) is trival when there exists a positive semi-definite matrix $\mathbf{\Sigma}_{\bm{X}}$ such that $\frac{1}{n}\mathbf{X}_n^T\mathbf{X}_n \rightarrow \mathbf{\Sigma}_{\bm{X}}$. 
For $m\in\mathbb{N}^+$ it can be deduced from (\ref{assumption:max and min eigen for A_m,m}) that
$\underset{n\rightarrow\infty}{\lim\inf}\;\lambda_{min}(\mathbf{\Sigma}_{C,n})>0$  and  $\underset{n\rightarrow\infty}{\lim\sup}\;\lambda_{max}(\frac{\partial \mathbf{\Sigma}_{C,n}}{\partial w_k})<\infty $ for any $C\in\bm{\Omega}_m$.
Such inequalities together with (\ref{assumption:max and min eigen for true}) are commonly used assumptions associated with spatial sampling in the increasing domain asymptotic framework (see \citep{Bachoc2021} and references therein). 
A necessary condition for the eigenvalues bounded below from zero is that
the minimal distance is bounded away from zero.
For more details, the reader can be referred to the proofs of Proposition D.4 in \citep{Bachoc2014} or of Theorem 5 in \citep{Bachoc2016}.
Moreover, when the model incorporates a nugget effect or measurement errors, then
$$\underset{n\rightarrow\infty}{\lim\inf}\;\lambda_{min}(\mathbf{\Sigma}_{0,n})>0\mbox{  and  }
\underset{n\rightarrow\infty}{\lim\inf}\;\lambda_{min}(\mathbf{\Sigma}_{C,n})>0 \mbox{  for } C\in\bm{\Omega}_m \mbox{, }m\in\mathbb{N}^+$$
naturally holds provided that the nugget or error variances are lower-bounded uniformly. It may take a bit of more work to check the upper bound of the eigenvalues and yet are not prohibitively difficult in some contexts (see examples and proofs in \citep{Shaby2012,Bachoc2014,Bachoc2016}).
The assumption (\ref{assumption: limsup max eigen final}) ensures $\lambda_{max}(\mathbf{\Sigma}_{\hat{C}_{m(n)},n})$ is bounded above from infinity, from which the consistency of $\bm{\hat{\beta}}_{m,n}$ can be deduced.
(\ref{assumption: parameter to covariance}) and
(\ref{assumption: covariance to parameter}) describe the relationship between $||C_0-\hat{C}_{m,n} ||_\infty$ and $\frac{1}{n}||\mathbf{\Sigma}_{\hat{C}_{m},n}-\mathbf{\Sigma}_{0,n}  ||_F^2$ where $||\cdot||_F$ is the Frobenius norm. (\ref{assumption: covariance to parameter}) means that for $C\in\bm{\Omega}$ bounded away from 
$C_0$, there is sufficient information in the spatial locations $\{\bm{s}_1,\cdots,\bm{s}_n\}$ to distinguish between the two covariance. It plays a decisive role in translating consistent estimation of the covariance matrix to covariance function through 
$\frac{1}{n}||\mathbf{\Sigma}_{\hat{C}_{m},n}-\mathbf{\Sigma}_{C_0,n}  ||_F^2 \rightarrow 0$.
Intuitively, for increasing domain asymptotics $C(||\bm{s}_i-\bm{s}_j  ||_2)$ can be quite small for many pairs $i,j\in \{1,2,\cdots,n\}$. This is why the normalization factor is $\frac{1}{n}$ rather than $\frac{1}{n^2}$ (see also in \citep{Keshavarz2016,Bachoc2021}).

\subsection{Numerical Implementation}\label{subsec: numerical implementation}
In this section, we discuss the algorithm of covariance estimation along with numerical techniques in details.
Without loss of generality we assume $E(Y)=0$ and only focus on  covariance estimation.

We introduce range parameter $\rho>0$ and marginal variance $\sigma^2>0$ so that $C_0\in \mathcal{A}$ can be written as  $C_0(h)=\sigma^2 R_0(\frac{h}{\rho})$ with $\sigma^2=C_0(0)$, of which the corresponding approximation covariance function for any fixed m is specified as
\begin{equation}
C_m(h)=\sigma^2\sum_{k=1}^m w_k A_{k,m}(\frac{h}{\rho})
\end{equation}
subject to $\bm{w}_m$ in the simplex $\mathcal{S}_m$.
For covariance functions which decay with distance too slowly or too quickly, large m are imperative. For instance, $m=230$ is required for Matern(1,5,1) and Matern(1,0.3,1) to achieve RAE of 0.05 compared to $m=3$ for Matern(1,1,1) computed through (\ref{eq:approximation_error_with_w}).
By means of $\rho$, we are allowed to scale the curve flexibly, and thus attain desired approximation performance efficiently with a relative small m and make possible otherwise infeasible calculation. 

Given a fixed m, parameters $(\bm{w}_m,\rho,\sigma^2)$ are to estimated by maximizing 
\begin{equation*}
\frac{1}{n}l_{n}(\bm{w}_{m},\rho,\sigma^2;\mathbf{Y}_n) = -\frac{1}{2n}logdet(\sigma^2\mathbf{\Sigma}_{\bm{w}_m,\rho,n})-\frac{1}{2n} \bm{Y}_n^T[\sigma^2\bm{\Sigma}_{\bm{w}_m,\rho,n}]^{-1}\bm{Y}_n + \frac{1}{2n}log(2\pi)
\end{equation*}
over $\{\bm{w}_m\in\mathcal{S}_m, \rho>0, \sigma^2>0\}$,
where $\mathbf{\Sigma}_{\bm{w}_m,\rho,n}=\sum_{k=1}^m w_k \mathbf{A}_{k,m,\rho,n}$ with $(i,j)$ element expressed as $\mathbf{A}_{k,m,\rho,n[i,j]}=A_{k,m}(\frac{||\bm{s}_i - \bm{s}_j ||_2}{\rho})$.
Letting $\frac{\partial \frac{1}{n}l_n}{\partial \sigma^2}=0$ gives us the estimation of $\sigma^2$ as
\begin{equation*}
\hat{\sigma}^2=\frac{\textbf{Y}_n^T[\mathbf{\Sigma}_{w_m,\rho,n}]^{-1}\bm{Y}_n}{n} \label{eq: form of v_hat}.
\end{equation*}
Then $(\bm{w}_m,\rho)$ can be computed by maximizing the profile likelihood $\frac{1}{n}l_n(\bm{w}_m,\rho,\hat{\sigma}^2;\mathbf{Y}_n)$
over $\{\bm{w}_m\in\mathcal{S}_m, \rho>0\}$
through iterative ascending optimization algorithm \ref{alg:estimation of w,rho}.  
\begin{algorithm}[H]
	\caption{Estimation of $(\bm{w}_m,\rho)$ given m and $\epsilon=10^{-3}$}
	\small
	\begin{algorithmic}[1]
		\State $i\gets 0$, $\bm{w}^{(0)}_m \gets (\frac{1}{m},\frac{1}{m},\cdots,\frac{1}{m})$
		\While{$i\leq 1$\textbf{ or }$\frac{\frac{1}{n}l_{n}(\bm{w}_m^{(i)},\rho^{(i)}; \bm{Y}_n)-\frac{1}{n}l_{n}(\bm{w}_m^{(i-1)},\rho^{(i-1)}; \mathbf{Y}_n)}{|\frac{1}{n}l_{n}(\bm{w}_m^{(i)},\rho^{(i)};\mathbf{Y}_n)|} > \epsilon$}
		\State $i\gets i+1$
		\State $\rho^{(i)} \gets \underset{\{\rho>0\}}{argmin}\;\frac{1}{n}l_n(\bm{w}_m^{(i-1)},\rho;\bm{Y}_n)$ \Comment{Given $\bm{w}_m^{(i-1)}$, update $\rho$} \label{step: update rho}
		\State $\bm{w}_m^{(i)} \gets \underset{\{\bm{w}_m\in\mathcal{S}_m\}}{argmin}\;\frac{1}{n} l_{n}(\bm{w}_m,\rho^{(i)};\bm{Y}_n)$\Comment{Given $\rho^{(i)}$, update $\bm{w}_m$} \label{step: update w_m}
		\EndWhile
		\State \textbf{return} $(\bm{w}_m^{(i)},\rho^{(i)})$
	\end{algorithmic}	
	\label{alg:estimation of w,rho} 
	
\end{algorithm}
Intrinsically, we need a precise estimation of $\bm{w}_m$ and a relatively coarse estimation of $\rho$. 
In line \ref{step: update rho} of Algorithm \ref{alg:estimation of w,rho}, $\rho$ is obtained through random search in the range of $[h_{min}, h_{max}]$, where $h_{min}$ and $h_{max}$ are minimum and maximum observed distance.
In line \ref{step: update w_m},  sparse solution of $\bm{w}_m$ is calculated through nonlinear optimization with augmented lagrange method via {\tt Rsolnp} package in R.  
Both updates guarantee a larger log-likelihood value than previous. Moreover, the merits of automatic variable selection together with a clear and succint presentation of covariance estimation make it an effective algorithm in practice.

To implement the procedure above, a data-driven m can be chosen from distinct values in $\{1+[n^a]: a=0.1, 0.15, 0.2,\cdots,0.9\}$ (i.e. at most 17 values). We iteratively run Algorithm \ref{alg:estimation of w,rho} with each candidate m until $\frac{1}{n}l_n(\bm{w}_m,\rho,\sigma^2;\mathbf{Y}_n)$ increases by less than 0.1\% and the iteration with maximum log-likelihood value will be retained as the best solution.

\section{Numerical Illustrations using Simulated Data Sets}\label{sec:simulation}
To assess the finite sample accuracy of our methodology, a simulation study was conducted. For comparison purposes, we performed covariance estimation using our proposed likelihood-based nonparametric method as well as the recent popular nonparametric methods \citep{Choi2013} and \citep{Huang2011}, referred to as {\em Choi method} and {\em Huang method} respectively. Their estimator were obtained through restricted weighted least squares as 
\begin{eqnarray}
\hat{\theta}_{Choi}=&\underset{\theta}{argmin}\sum_{i=1}^L w_{i,Choi} (\hat{C}_{E}(h_i)-C_{\theta}(h_i))^2 \hspace{1cm} \label{eq: Choi nonparametric covariance based on two NP}\\
\hat{\theta}_{Huang}=&\underset{\theta}{argmin}\sum_{i=1}^L w_{i,Huang} (\hat{\gamma}_{E}(h_i)-\gamma_{\theta}(h_i))^2 \hspace{1cm}\label{eq: Huang nonparametric variogram based on two NP}
\end{eqnarray}
where $\hat{C}_{E}(\cdot)$ and $\hat{\gamma}_{E}(\cdot)$ are empirical covariance and empircal semivariogram, $\{h_i\}_{i=1}^L$ is a sequence of distance between $L$ distinct pairs of locations, nonparametric models $C_{\theta}(h)$ and $\gamma_{\theta}(h)$ as well as weights $\{w_{i,Choi}\}_{i=1}^L$ and $\{w_{i,Huang}\}_{i=1}^L$ have different specifications covered in \citep{Choi2013} and \citep{Huang2011}. A modified version of Choi method was also implemented in our simulation, that is to minimize (\ref{eq: Choi nonparametric covariance based on two NP}) with $\{w_{i,Choi}\}_{i=1}^L$ replaced by $\{w_{i,Huang}\}_{i=1}^L$, given by
\begin{equation}
\hat{\theta}_{modified-Choi}=\underset{\theta}{argmin}\sum_{i=1}^L w_{i,Huang} (\hat{C}_{E}(h_i)-C_{\theta}(h_i))^2.
\label{eq: modified Choi nonparametric covariance based on two NP}
\end{equation}
In addition, to show the superiority of nonparametric methods over parametric ones in the aspect of reducing risks of model misspecification, we also estimated through a collection of parametric models including Matern, Cauchy, Gaussian and GenCauchy. 
Four types of parametric methods were adopted, namely, likelihood-based, cov-based using weights of Choi, variogram-based parametric methods and cov-based using weights of Huang, corresponding to the four aforementioned nonparametric methods.
Clearly, parameters of likelihood-based parametric methods were estimated by maximizing log-likelihood function. 
As for rest of the three types,  parameters of those parametric models, denoted as $\theta$, were estimated by minimizing (\ref{eq: Choi nonparametric covariance based on two NP})(\ref{eq: Huang nonparametric variogram based on two NP})(\ref{eq: modified Choi nonparametric covariance based on two NP}) in the same spirit of nonparametric approaches.


Regarding irregular locations,
$n=60$  coordinates were selected randomly over domain $[0,20]\times[0,20]$ beforehand. In each simulation,  as suggested in \citep{Im2007}, we generated $r=200$ independent realizations of Gaussian process at these locations totaling 12000 observations.  The covariance matrix corresponding to this type of dataset is block diagonal, which allows us to have a large number of observations so that the parameters can be well estimated while keeping the computational load at a manageable level.   
Note that a technique was employed here to speed up estimation of cov-based and variogram-based methods.
Since each distance is duplicated a number of times, we calculated average of r empirical covariances and semivariograms through 
\begin{eqnarray*}
	\hat{C}_{E}(h_i)=\frac{1}{r}\sum_{k=1}^{r}\hat{C}_{E}(h_i;Y_{1,k},\cdots,Y_{n,k}) 
	\mbox{   and }
	\hat{\gamma}_{E}(h_i)=\frac{1}{r}\sum_{k=1}^{r}\hat{\gamma}_{E}(h_i;Y_{1,k},\cdots,Y_{n,k}) \label{eq:average of empirical variogram},
\end{eqnarray*}
where $\hat{C}_{E}(h_i;Y_{1,k},\cdots,Y_{n,k})$ and $\hat{\gamma}_{E}(h_i;Y_{1,k},\cdots,Y_{n,k})$ are empirical covariance and semivariogram of $k_{th}$ realization $\{Y_{i,k}\}_{i=1}^{n}$, in order that computational burden can be significantly reduced in model fitting.

To mimic measurement-error free version of stochastic process, given coordinates described above, we generated data in $\mathbb{R}^2$ from Gaussian random fields with mean zero and various isotropic covariance functions valid in $\mathbb{R}^\infty$. We chose five covariance functions with equal variance $C_0(0)=1$ including typical parametric covariance families in Setting 1-4 and a linear combination of two matern models in Setting 5. 

\begin{enumerate}
	\item Setting 1: Matern $(1,1.25,1)$ where
	$C_{Matern}(h;\sigma^2,\rho,\nu) =\sigma^2 \frac{1}{2^{\nu-1} \Gamma(\nu)} (\frac{h}{\rho})^{\nu}  K_{\nu}(\frac{h}{\rho})$; 
	\item Setting 2: Cauchy $(1,0.8)$ where
	$C_{Cauchy}(h;\sigma^2,\rho) = \sigma^2 (1+(\frac{h}{\rho})^2)^{-\frac{1}{2}}$; 
	\item Setting 3: Gaussian $(1,3)$ where
	$C_{Gaussian}(h;\sigma^2,\rho) = \sigma^2 exp(-(\frac{h}{\rho})^2)$;
	\item Setting 4: GenCauchy $(1,0.3,0.5,2)$ where 	$C_{GenCauchy}(h;\sigma^2,\rho,\kappa_1,\kappa_2)=\sigma^2(1+(\frac{h}{\rho})^{\kappa_2})^{-\frac{\kappa_1}{\kappa_2}}$;
	\item Setting 5: Linear combination of two matern models  $(1,1.25,\frac{3\sqrt{2}}{4},1,2)$ where
	\begin{equation*}
	C_{LinearMatern}(h;\sigma^2,\rho_1,\rho_2,\nu_1,\nu_2)=
	\frac{1}{2}C_{Matern}(h;\sigma^2,\rho_1,\nu_1) +\frac{1}{2}C_{Matern}(h;\sigma^2,\rho_2,\nu_2).
	\label{eq:linearmatern}
	\end{equation*}
\end{enumerate}

It is worth pointing out that our basis has a different tail behaviour, namely, $1/(1+h^{2(m-k+1)})$,  from the five simulated settings:
Setting 2 and 4 have a much lower decay of $(1+h^2)^{-\frac{1}{2}}$ and $(1+h^2)^{-\frac{1}{4}}$ while Setting 3 has a faster decay of $exp(-h^2)$. Besides Setting 3 is excluded from $\mathcal{A}$ while Setting 1,2,4 and 5 are not. Despite such challenges, we included those covariances here to test our method. 

In order to assess estimation performance, we refer to several quantities below. 
\begin{itemize}
	\item variance: $\hat{C}_0(0)$ serves as a reasonable metric to quantify estimation performance, specifically for parametric models being exactly $\hat{\sigma}^2$.  Considering $C_0(0)=1$ for all five settings,  bias is computed as $Bias(C_0(0))=\hat{C}_0(0)-1$.
	Notably, variogram-based methods directly estimate $\hat{\gamma}_0(\cdot)$ instead of $\hat{C}_0(\cdot)$. Whereas, followed from $C_0(0)=\underset{h\rightarrow\infty}{\lim}\gamma_0(h)$, it appears to be an appropriate way to obtain the marginal variance estimates $\hat{C}_0(0)\approx\hat{\gamma}_0(h_{max})$ 
	as well as correlation and covariance estimates through $\hat{R}_0(h)=1-\hat{\gamma}_0(h)/\hat{C}_0(0)$ and $\hat{C}_0(h)=\hat{C}_0(0)-\hat{\gamma}_0(h)$.
	
	\item  Correlation function: The scaled version of 2-norm $\frac{||\hat{R}_0-R_0||_2}{\sqrt{h_m}}$ suggested by \citep{Im2007} and sup-norm $||\hat{R}_0-R_0||_{\infty}$ are selected to measure the difference between the true correlation  and its estimates with the aid of some upper bound $h_m$ of each setting.
	
	We focus more on strictly decreasing correlation estimation for $h\in[0,h_m]$ than the tail behaviour on $\left[h_m, \infty\right)$. Hence, the upper bound distance $h_m$ for evaluation is determined as the smallest value $h$ that satisfies $R_0(h)\geq \epsilon$, where $\epsilon$ for Setting 1-5 are $0.001,0.05,0.001,0.15,0.001$ and thus $h_m=10.3, 16,7.9,13.3,10.5$ respectively. The reason for larger choices of $\epsilon$ for Setting 3 and 5 is that slower decrease rates of Cauchy and GenCauchy lead to a larger value of minimum correlation evaluated at maximum distance  $20\sqrt{2}$  in coordinate generation region $[0,20]\times [0,20]$.
	Put it plainly, these two metrics are computed through 
	(\ref{eq: 2-norm correlation}) and (\ref{eq: inf-norm correlation})
	\begin{eqnarray}
	\frac{||\hat{R}_0-R_0||_2}{\sqrt{h_m}}&\approx&\sqrt{\frac{1}{K}\sum_{i=1}^K (R_0(h_i)-\hat{R}_0(h_i))^2} ,\label{eq: 2-norm correlation}\\
	||\hat{R}_0-R_0||_{\infty} &\approx &\underset{\{i=1,\cdots,K\}}{max} |R_0(h_i)-\hat{R}_0(h_i)|,\label{eq: inf-norm correlation}
	\end{eqnarray}
	where $\{h_i=\frac{i}{K}h_m\}_{i=1}^K$ are $K=2000$ equispaced points in the range $[0,h_m]$ so that both metrics share the same upper bound of 2 for all five settings.
	
	\item Covariance function and variogram function: Considering that Choi method and modified Choi method aim to minimize 2-norm error of covariance estimates and Huang method works on variogram estimation directly,
	for reference we also computed the scaled version of 2-norm and inf-norm of covariance and semivariance in the same way as (\ref{eq: 2-norm correlation}) and (\ref{eq: inf-norm correlation}). 
\end{itemize}

We ran $N=100$ monte carlo simulations. In each simulation, four types of nonparametric methods along with parametric ones were applied. To implement Choi method and modified Choi method, we set $m=5$ and $p=3$, the same tuning parameters as in \citep{Choi2013}, where p is the order number and m is the knot number, different from the number of basis m specified in our method.
In regard to Huang method, we stuck to the tuning procedures in \citep{Huang2011}. Specifically, we set $v=10$ and $L=200$, where $v$ is the upper bound of spectrum and $L$ is the number of grids. For smoothing parameter selection, $\lambda_k=10^{u_k}$ was chosen, where $\mu_k=6(k-1)/19$, $k=1,\cdots,20$ and generalized cross-validation approach was followed to determine the appropriate value.

To employ our method, with sample size of 12000, the resulting m was selected from \{2,3,5,7,11,17,
27,43,69,110.26,281,449,$\cdots$\}.
As a result, $84,89,87,84$ and $83$ percent of MC simulations stopped respectively for Setting 1-5 when m is no greater than $43,69,110,449$ and $69$.
We did estimation on a Linux machine with dual AMD processors.
On average, Setting 4 took the longest time of 13 mins per simulation while Setting 1 took the shortest  of 0.8 mins per simulation with 16 cores.
Besides, each simulation took around 3.8 mins for Huang method, 1.2 mins for Choi method and modified Choi method with 2 cores.

Numerical results of nonparametric methods and a comparison with parametric methods based on 100 simulations are presented in Table \ref{table:simulation of average for three NP} and 
Table \ref{table:simulation of relative loss efficiency of three nonparametric methods vs best parametric fit}.  Additional results for parametric methods in Table \ref{table:simulation of average for variogram-based parametric methods}-
\ref{table:simulation of average for cov-based parametric methods using modified weights} and 
boxplots in Fig \ref{fig:simulation of boxplot} are in Appendix \ref{sec:additional results of section 4}.
According to row 16-21 and 23-28 of Table \ref{table:simulation of average for three NP}, it is interesting to notice that in contrast to original Choi method, modified Choi method has obviously better performance in terms of norm errors along with competitive results in estimation of $C_0(0)$.
It implicates the weighting scheme devised by \citep{Huang2011} is more feasible than that of \citep{Choi2013} in our simulation setting probably because the latter one outweights $\hat{C}_E$ with low variability and downweights those with high variability, which doesn't work well 
for irregular-gridded data with replicates. Therefore, we no longer consider Choi method together with cov-based parametric methods using weights of Choi and stick to modified Choi method in the rest of papers. 

Generally speaking, our method yields the smallest mean values of scaled 2-norm and sup-norm errors, except when the underlying model is GenCauchy (Setting 4), Huang method has the smallest yet rather similar mean value of scaled $||\hat{\gamma}_0-\gamma_0||_{2}$. And it is true even for Gaussian covariance outside the space $\mathcal{A}$.

It is seen from the row 8-14 of Table \ref{table:simulation of average for three NP}
that Huang method gives substantially large scaled 2-norm and sup-norm errors of correlation estimation that are 229\%,  71\%, 146\%, 168\%, 224\%  and  191\%, 200\%, 150\%, 96\%, 209\% higher in contrast with our method in Setting 1-5; as for covariance estimation, Huang method engenders scaled norm-2 and sup-norm errors 190\%, 76\%, 70\%, 181\%, 168\% and 
125\%, 180\%,  71\%, 135\%, 121\% greater than those of our method for all settings. It seems forgivable for Huang method to
have larger norm errors regarding correlation and covariance estimation for the reason that Huang method directly focuses on variogram estimation and a not large enough value of $h_{max}$ may cause underestimate of $C_0(0)$ computed through $\hat{\gamma}_0(h_{max})$ especially when the true underlying model is GenCauchy shown in Fig \ref{fig: simulation boxplot of GenCauchy C_zero} in Fig \ref{fig:simulation of boxplot}.
Whereas, in respect to variogram estimation, when the true model is Matern, Cauchy, Gaussian and LinearMatern, Huang method have larger scaled 2-norm and inf-norm errors compared to our method by factors of 2.0, 1.4, 2.9, 1.8 and 2.8, 2.5, 2.8, 2.5; when the true model is GenCauchy, our method has slightly larger scaled 2-norm errors (only 7\%) and smaller sup-norm errors (around 17\%). 

Compared to results of modified Choi method in 15-21 row, our method is better by factors of 1.1, 1.2, 6.7, 2.7, 1.6 and  1.3, 1.4, 8.3, 4.2, 1.9 in terms of scaled 2-norm and inf-norm errors of correlation in Setting 1-5. As for 2-norm errors of covariance estimation, the criterion modified Choi method is seeking to minimize,  our method still outperforms especially when the underlying model is Gaussian and GenCauchy. The reason might be that \citep{Choi2013} provides no particular guidelines of tuning parameter selection for modified Choi method (and Choi method), which can be problematic in capturing too quick or too slow trend.

Besides, by examination of bias of $C_0(0)$ in row 1,8,15 and 21 of Table \ref{table:simulation of average for three NP},
although modified Choi method outperforms under Setting 1,2 and 5, our method has quite stable performance in all five settings with even smaller bias for Setting 3. In addition, as shown in Fig \ref{fig:simulation of boxplot} in Appendix \ref{sec:additional results of section 4}, 
our method produces unbiased estimation of $C_0(0)$ and fairly stable and satisfactory correlation estimation.

Table \ref{table:simulation of relative loss efficiency of three nonparametric methods vs best parametric fit}  
provides a comparison of nonparametric methods and best parametric fits, where negative entries implies better performance of corresponding nonparametric methods than all parametric methods.
Under Setting 1-4, likelihood-based estimations under corresponding true models turns out to be the best parametric fits indicating the superiority of maximum likelihood estimator over cov-based and variogram based estimators. 
Even then,  the advantage of using the correct parametric model is modest compared to our nonparametric method because our method can produce at most 5.55, 3.23 and 2.98 times larger norm errors of correlation, covariance and variogram estimates than the best parametric fits in Setting 1-4.  
When no true parametric model can be used for Setting 5, our method outperforms all the four parametric models including frequently used Matern model in terms of norm errors.

\section{A Case Study based on Precipitation Data}\label{sec:real data}

\begin{figure}[H]
	\centering
	\begin{subfigure}{0.48\textwidth}
		\centering
		\includegraphics[width=\linewidth]{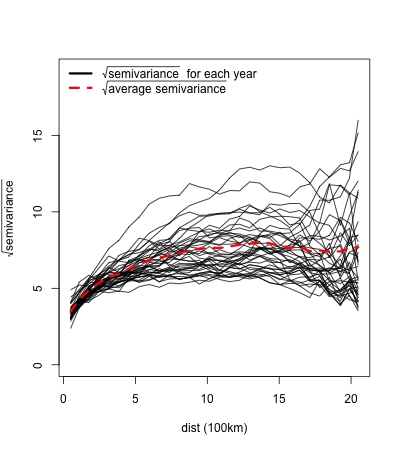}
		\caption{$\sqrt{\text{semivariance}}$ (inch) vs.\ distance (100 km)}
		\label{fig: real data of empirical semivariance}
	\end{subfigure}
	\hspace{5pt}
	\begin{subfigure}{0.48\textwidth}
		\centering
		\includegraphics[width=\linewidth]{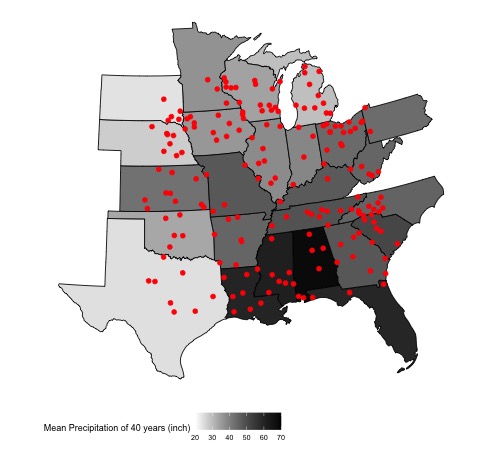}
		\caption{189 stations (dots) in the US together with the mean precipitation surface over forty years}
		\label{fig: real data of heatmap of stations}
	\end{subfigure}
	\caption{(a) $\sqrt{\text{semivariance}}$ (inch) vs.\ distance (100 km), and (b) 189 stations (dots) in the US together with the mean precipitation surface over forty years.}
	\label{fig: real data of empirical and heatmap}
\end{figure}

In this section, we applied the three aforementioned nonparametric and twelve parametric methods to an annual total precipitation dataset obtained from NOAA. 
We chose the study region to be between latitudes 30 and 46 and longitudes -100 and -80. The period of data collection is 1978-2017 in total of $r=40$ years. 
Detailed documentation of this dataset can be found at \url{https://www.ncei.noaa.gov/data/gsoy/doc/GSOYReadme.txt}.  We included 194 stations that had no missing data. To get detrended residuals,  two-way anova model $Y_{ij}=\mu+\alpha_{i}+\beta_{j}+\epsilon_{ij}$ was chosen with adjusted R-squared 0.6788 and significant main effects, where $Y_{ij}$ is the response of Year $j$ at  Station $i$, $\alpha_i$ and $\beta_j$ are station effect and year effect. We consider the residuals as a Gaussian random process.
For each stations, we checked  autocorrelation of residuals, removed 5 stations with significant time dependence and kept the rest $n=189$ stations in total of $rn=7560$ observations.
Thus we modelled the observations from each year as independent realizations of the same process with a different mean. The empirical semivariance in Fig \ref{fig: real data of empirical semivariance} in Fig \ref{fig: real data of empirical and heatmap}  indicates the existence of nugget effect. This is straightforward for our method and Huang method along with likelihood-based and variogram-based parametric methods to include a nugget effect.  
As for modified Choi method and covariance-based parametric models with weights of Huang,  an adaption suggested in \citep{Choi2013} was made to take care of the nugget effect. Concretely, covariance estimation was conducted based on distinct pairs, that is, $\epsilon_{i_1j}\cdot\epsilon_{i_2j}$, with $i_1\neq i_2$, leaving out nugget term. We then subtracted the estimated covariance function at zero from sample variance and got estimation of nugget effect. 
Between these 189 sites shown in Fig \ref{fig: real data of heatmap of stations} in Fig \ref{fig: real data of empirical and heatmap}, the smallest distance is 1.86km, the largest distance is 2120.13km and median distance is 910.16km.

We performed parametric and nonparametric estimation on a 2-core Linux machine with dual AMD processors.
For computation purposes, Huang method and variogram-based parametric methods were manipulated based on a binned version of empirical variogram. 
Following the procedure in Section \ref{sec: estimation method}, our algorithm stopped at m=23 and the model of m=23 was chosen as the best model with maximum log-likelihood value.  Numerical values and plots for candidate m are displayed in Table \ref{table:all m our method Year 40 Station 189} and Fig \ref{fig:all m our method Year 40 Station 189}, which clearly show that our method is quite robust against the choice of m in that when $m\geq 10$, correlation estimations almost overlap and the sum of nugget estimates and $\hat{C}_0(0)$ remain stable around 55.

Table \ref{table:real data results of nonparametric models} and Fig \ref{fig:real data estimated plots of nonparametric methods} present estimated results and plots of three nonparametric methods. 
It is not surprising our method achieved the largest log-likelihood value. In contrast, that of modified Choi method is rather small.
Although it is a bit unfair to compare methods that seek maximizing the likelihood with methods that optimize other criteria, small value of likelihood should give us an idea of how good the estimated function is. Curves of estimated semivariance for our method and Huang method almost overlap when distance is less than $7.5\times 100km$ and are somewhat different from that of modified Choi method.

Table \ref{table:real data results of parametric models} and Fig \ref{fig:real data estimated plots of nonparametric methods and parametric methods} are numerical estimation results and plots of three types of parametric methods. 
As for parametric methods, cov-based parametric methods with weights of Huang have quite different nugget estimates compared to others: those of  Matern and Gaussian are larger than 20 and those of Cauchy and GenCauchy are as small as zero. Moreover, estimated semivarance of cov-based GenCauchy and cov-based Cauchy diverge away from those of the rest parametric methods and have zero estimates of nugget effect, much smaller than others, which coincide with the fact that the former two have smaller likelihood values.  The maximum log-likelihood value among all parametric methods 
is -0.709, smaller than that of our proposed nonparametric method even when the number of basis m is as small as 4. It implies great efficiency of our nonparametric method in maximizing deviance values. For more detailed parametric results, one can refer to Table \ref{table:real data results of parametric models} in Appendix \ref{sec:additional results of section 5}.

\begin{figure}[H]
	\centering
	\begin{minipage}[t]{0.44\textwidth}
		\centering
		\includegraphics[width=\linewidth]{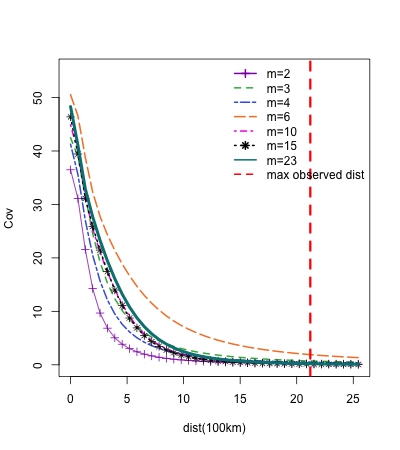}
		\label{fig: real data cov plot of candidate m}
		(a) Estimated covariance
	\end{minipage}
	\hspace{3pt}
	\begin{minipage}[t]{0.44\textwidth}
		\centering
		\includegraphics[width=\linewidth]{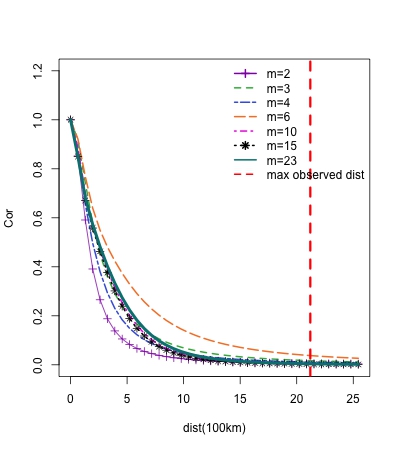}
		\label{fig: real data cor plot of candidate m}
		(b) Correlation for candidate $m$ of our proposed method
	\end{minipage}
	\caption{(a) Estimated covariance and (b) correlation for candidate $m$ of our proposed method. Note: plots for $m=10$ and $m=15$ almost overlap.}
	\label{fig:all m our method Year 40 Station 189}
\end{figure}

\begin{figure}[H]
	\centering
	\begin{minipage}[t]{0.48\textwidth}
		\centering
		\includegraphics[width=\linewidth]{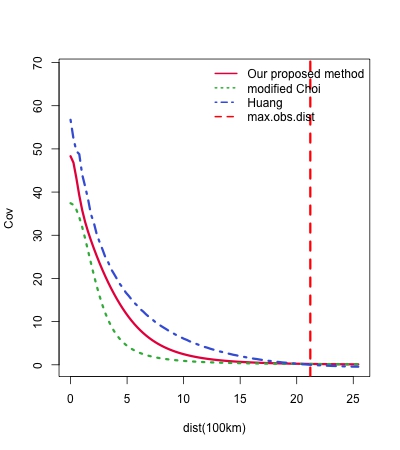}
		
		(a) Estimated covariance
	\end{minipage}
	\hspace{5pt}
	\begin{minipage}[t]{0.48\textwidth}
		\centering
		\includegraphics[width=\linewidth]{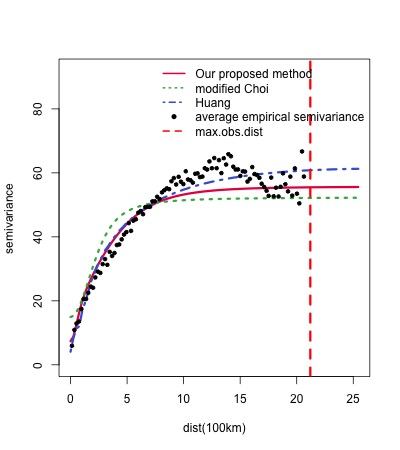}
		
		(b) Semivariance of nonparametric methods based on real data
	\end{minipage}
	\caption{(a) Estimated covariance and (b) semivariance of nonparametric methods based on real data.}
	\label{fig:real data estimated plots of nonparametric methods}
\end{figure}

\begin{figure}[h]
	\centering
	\begin{minipage}[t]{0.48\textwidth}
		\centering
		\includegraphics[width=\linewidth]{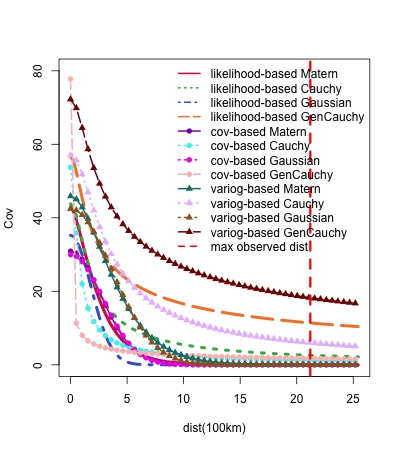}
		\label{fig: real data cov plot of parametric methods}
		(a) Estimated covariance
	\end{minipage}
	\hspace{2pt}
	\begin{minipage}[t]{0.48\textwidth}
		\centering
		\includegraphics[width=\linewidth]{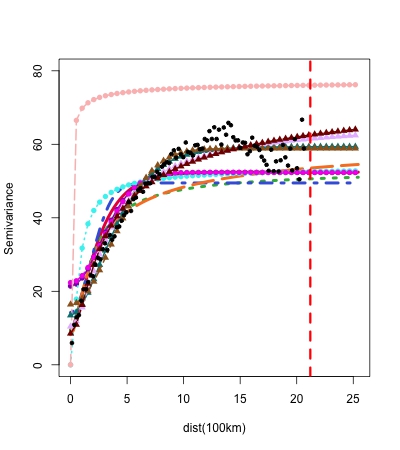}
		\label{fig: real data semivariance plot of parametric methods}
		(b) Semivariance of parametric methods based on real data
	\end{minipage}
	\caption{(a) Estimated covariance and (b) semivariance of parametric methods based on real data.}
	\label{fig:real data estimated plots of nonparametric methods and parametric methods}
\end{figure}

\section{Conclusion}\label{sec:conclusion}
In this article, we propose a new nonparametric model for isotropic covariance functions valid in any dimension.
Our model is capable of uniformly approximating any isotropic covariance functions in $\mathbb{R}^\infty$ with continuous spectral density. To the best of our knowledge, this is the first work that explores $L_\infty$ norm approximation as well as $L_q$ norm approximation under mild conditions.
We also design likelihood-based estimator of isotropic Gaussian processes with irregularly-spaced data.
One advantage is that 
in the spirit of maximizing likelihood it directly involves all the observations with their correlations taken into account, unlike some existing methods which connect observations to covariance or variogram indirectly through empirical estimators. And optimization over a simplex naturally offers us sparse estimations and reduces the worry of overfit. 
It is noted that although Huang method and Choi method are applied to both simulated and real data, their theoretical properties are unknown. In contrast, a theoretical proof of consistency results of our sieve estimator is derived. Another advantage of our method is that we investigate the range parameter estimates, the key point to take care of underlying covariance functions with too slow and too quick decay rate, and present a practical guidance of the choice of m, which makes it a well-applicable algorithm. As shown in Section \ref{sec:simulation}, our method not only outperforms other nonparametric methods substantially in terms of $L_\infty$ and scaled $L_2$ norm errors of correlation, covariance and semivariance estimates in all five settings, but also is superior to all parametric methods when the true covariance is LinearMatern in Setting 5.

Admittedly, in theory our model may not be applicable for all types of isotropic covariance functions, but through extensive numerical simulation studies we show good numerical performances of the proposed method even for those covariance functions that do not belong to domain of covariance functions that our theoretical model allows for.
Besides our proposed method is more time-consuming and computation-expensive for the reason that evaluating the likelihood requires order $n^3$ operations, yet it still seems tolerable considering the gains in estimation accuracy.
Thus, as a part of future explorations, we plan to extend our model with covariance tapering, a popular and well-studied technique enabling the use of sparse matrix algorithms to increase computation efficiency (\citep{Kaufman2008}).
In Section \ref{sec:real data}, we treat precipitation data from different years as independent observations. Hence, another interesting direction is to extend our model to space-time covariance functions to handle more flexiable cases where observations from different years are allowed to be auto-correlated.

	\section*{Disclosure statement}
	The authors report there are no competing interests to declare.

	\section*{Data availability statement}
	The data that supports the findings of this study in Section \ref{sec:real data} are available in the repository  `access' at \url{https://www.ncei.noaa.gov/data/gsoy/}, (\cite{NOAA2020}), along with its detailed documentation found at \url{https://www.ncei.noaa.gov/data/gsoy/doc/GSOYReadme.txt}.

	\section*{Acknowledgement}
	We would like to express our sincere appreciation to Dr.Chunfeng Huang for sharing the code of his paper and providing insightful suggestions. We would also like to thank the referees whose comments helped us to improve the paper substantially and the edior whose assistance is gratefully acknowledged.


\begin{thebibliography}{33}
	\newcommand{\enquote}[1]{`#1'}
	\providecommand{\natexlab}[1]{#1}
	\providecommand{\url}[1]{\normalfont{#1}}
	\providecommand{\urlprefix}{ }
	
	\bibitem[Abramowitz and Stegun(1965)]{Abramowitz1965}
	Abramowitz, M., and Stegun, I. (1965), \emph{Handbook of Mathematical
		Functions}, 9th ed., New York.
	
	\bibitem[Amemiya(1985)]{Amemiya1985}
	Amemiya, T. (1985), \emph{Advanced Econometrics}, Harvard University Press, pp.
	105--158. Asymptotic properties of extremum estimators.
	
	\bibitem[Anderson(1970)]{Anderson1970}
	Anderson, T.W. (1970), \enquote{Estimation of covariance matrices which are
		linear combinations or whose inverse are linear combinations of given
		matrices}, \emph{Essays in Probability and Statistics}, pp. 1--24.
	
	\bibitem[Bachoc(2014)]{Bachoc2014}
	Bachoc, F. (2014), \enquote{Asymptotic analysis for the role of spatial
		sampling for covariance parameter estimation of Gaussian processes},
	\emph{Journal of Multivariate Analysis}, pp. 1--25.
	
	\bibitem[Bachoc(2021)]{Bachoc2021}
	Bachoc, F. (2021), \emph{Oxford Handbook of Innovation}, Springer, Cham, chap.
	Asymptotic analysis of Maximum Likelihood Estimation of Covariance Parameters
	for Gaussian Processes: An Introduction with Proofs, pp. 283--303.
	
	\bibitem[Bachoc and Furrer(2016)]{Bachoc2016}
	Bachoc, F., and Furrer, R. (2016), \enquote{On the smallest eigenvalues of
		covariance matrices of multivariate spatial processes}, \emph{Stat}, 1,
	102--107.
	
	\bibitem[Bernstein(1912)]{Bernstein1912}
	Bernstein, S. (1912), \enquote{Demonstration of a theorem of Weierstrass based
		on the calculus of probabilities}, \emph{Communications of the Kharkov
		Mathematical Society}.
	
	\bibitem[Chen(2007)]{Chen2007}
	Chen, X. (2007), \enquote{Large Sample Sieve Estimation of Semi-Nonparametric
		Models}, in \emph{Handbook of Econometrics}, Vol.~6B, 1st ed., eds.
	J.~Heckman and E.~Leamer, chap.~76.
	
	\bibitem[Choi et~al.(2013)Choi, Li, and Wang]{Choi2013}
	Choi, I., Li, B., and Wang, X. (2013), \enquote{Nonparametric Estimation of
		Spatial and Space-Time Covariance Function}, \emph{Journal of Agricultural,
		Biological, and Environmental Statistics}, 18, 611--630.
	
	\bibitem[Choudhuri et~al.(2004)Choudhuri, Ghosal, and Roy]{Choudhuri2004}
	Choudhuri, N., Ghosal, S., and Roy, A. (2004), \enquote{Bayesian estimation of
		the spectral density of a time series}, \emph{Journal of the American
		Statistical Association}, 99, 1050--1059.
	
	\bibitem[Cressie(1993)]{Cressie1993}
	Cressie, N. (1993), \emph{Statistics for Spatial Data}, John Wiley \& Sons, New
	York.
	
	\bibitem[Cuzik(1995)]{Cuzick1995}
	Cuzik, J. (1995), \enquote{A Strong Law for Weighted Sums of i.i.d Random
		Variables}, \emph{Statistics \& Probability Letters}, 76, 1482--1487.
	
	\bibitem[Efron et~al.(2004)Efron, Hastie, Johnstone, and Tibshirani]{Efron2004}
	Efron, B., Hastie, T., Johnstone, I., and Tibshirani, R. (2004), \enquote{Least
		angle regression}, \emph{Annals of Statistics}, 32, 407--451.
	
	\bibitem[Farouki(2012)]{Farouki2012}
	Farouki, R.T. (2012), \enquote{The Bernstein polynomial basis: a centennial
		retrospective}, \emph{Computer Aided Geometric Design}, 29, 379--419.
	
	\bibitem[Grenander(1981)]{Grenander1981}
	Grenander, U. (1981), \emph{Abstract Inference}, Wiley, New York.
	
	\bibitem[Horn and Johnson(1991)]{Topics1991}
	Horn, R., and Johnson, C. (1991), \emph{Topics in Matrix Analysis}, Cambridge
	University Press.
	
	\bibitem[Huang et~al.(2011)Huang, Hsing, and Cressie]{Huang2011}
	Huang, C., Hsing, T., and Cressie, N. (2011), \enquote{Nonparametric estimation
		of the variogram and its spectrum}, \emph{Biometrika}, 98, 775--789.
	
	\bibitem[Huang and Rossini(1997)]{Huang1997}
	Huang, J., and Rossini, A.J. (1997), \enquote{Sieve Estimation for the
		Proportional-Odds Failure-Time Regression Model With Interval Censoring},
	\emph{Journal of the American Statistical Association}, 92, 960--967.
	
	\bibitem[Im et~al.(2007)Im, Stein, and Zhu]{Im2007}
	Im, H., Stein, M., and Zhu, Z. (2007), \enquote{Semiparametric Estimation of
		Spectral Density With Irregular Observations}, \emph{Journal of the American
		Statistical Association}, 102, 726--735.
	
	\bibitem[Kaufman et~al.(2008)Kaufman, Schervish, and Nychka]{Kaufman2008}
	Kaufman, C.G., Schervish, M.J., and Nychka, D.W. (2008), \enquote{Tapering for
		Likelihood-Based Estimation in Large Spatial Data Sets}, \emph{Journal of the
		American Statistical Association}, pp. 1545--1555.
	
	\bibitem[Keshavarz et~al.(2016)Keshavarz, Scott, and Nguyen]{Keshavarz2016}
	Keshavarz, H., Scott, C., and Nguyen, X. (2016), \enquote{On the consistency of
		inversion-free parameter estimation for Gaussian random fields},
	\emph{Journal of Multivariate Analysis}, 150, 245--266.
	
	\bibitem[Newey(1991)]{Newey1991}
	Newey, W.K. (1991), \enquote{Uniform Convergence in Probability and Stochastic
		Equicontinuity}, \emph{Econometrica}, 59, 1161--1167.
	
	\bibitem[Newey and McFadden(1994)]{Newey1994}
	Newey, W.K., and McFadden, D. (1994), \emph{Handbook of Econometrics}, Elsevier
	Science, chap.~4, pp. 2111--2245. Large sample estimation and hypothesis
	testing.
	
	\bibitem[NOAA(2020)]{NOAA2020}
	NOAA (2020). Global Summary of The Month (GSOY) data retrieved at
	\url{https://www.ncei.noaa.gov/data/gsoy/}.
	
	\bibitem[Piotr et~al.(2017)Piotr, Caroline, and Donald]{Piotr2017}
	Piotr, Z., Caroline, U., and Donald, R. (2017), \enquote{Maximum likelihood
		estimation for linear Gaussian covariance models}, \emph{Journal of Royal
		Statistical Society}, 79, 1269--1292.
	
	\bibitem[Reich and Fuentes(2012)]{Reich2012}
	Reich, B., and Fuentes, M. (2012), \enquote{Nonparametric Bayesian models for a
		spatial covariance}, \emph{Statistical Methodology}, 9, 265--274.
	
	\bibitem[Schoenberg(1938)]{Schoenberg1938}
	Schoenberg, I. (1938), \enquote{Metric spaces and completely monotone
		functions}, \emph{Annals of Mathematics}, 39, 811--841.
	
	\bibitem[Shaby and Ruppert(2012)]{Shaby2012}
	Shaby, B., and Ruppert, D. (2012), \enquote{Tapered Covariance: Bayesian
		Estimation and Asymptotics}, \emph{Journal of Computational and Graphical
		Statistics}, 21, 433--452.
	
	\bibitem[Shen and Wong(1994)]{Shen1994}
	Shen, X., and Wong, W. (1994), \enquote{Convergence rate of sieve estimates},
	\emph{Annals of Statistics}, 22.
	
	\bibitem[Stein(1999)]{Stein1999}
	Stein, M. (1999), \emph{Interpolation of Spatial Data: Some Theory for
		Kriging}, Springer, New York.
	
	\bibitem[van~der Vaart and Wellner(1996)]{Van1996}
	van~der Vaart, A.W., and Wellner, J.A. (1996), \emph{Weak Convergence and
		Empirical Processes}, Springer, New York.
	
	\bibitem[Xue et~al.(2010)Xue, Miao, and Wu]{Xue2010}
	Xue, H., Miao, H., and Wu, H. (2010), \enquote{Sieve estimation of constant and
		time-varying coefficients in nonlinear ordinary differential equation models
		by considering both numerical error and measurement error}, \emph{Annals of
		Statistics}, 38, 2351--2387.
	
  	\bibitem[Wang(2022)]{Wang2022}
  Wang, Y. (2022), \emph{Covariance Function Estimation and Causal Inference Method}, North Carolina State University.
	
	\bibitem[Zhang(2004)]{Zhang2004}
	Zhang, H. (2004), \enquote{Inconsistent Estimation and Asymptotically
		Equivalent Interpolations in Model-based Geostatistics}, \emph{Journal of the
		American Statistical Association}, 99, 250--261.
	
	\bibitem[Zheng et~al.(2010)Zheng, Zhu, and Roy]{Zheng2010}
	Zheng, Y., Zhu, J., and Roy, A. (2010), \enquote{Nonparametric Bayesian
		Inference for the Spectral Density Function of a Random Field},
	\emph{Biometrika}, 97, 238--245.
	
\end{thebibliography}

\bibliographystyle{gNST}

\appendix

\section{Proof of Main Results }
\label{sec:verify}

\subsection{Proof of Approximation Theorem in Theorem 1} 
\label{subsec: proof of covariance approximation}

We first give a theorem that will be used for later proof of Approximation Theorem.

\noindent \textbf{Dini Theorem.}
$A$ is a compact set. $f:A\rightarrow\mathbb{R}$ is continuous. $\{f_n\}_{n=1}^\infty$ is a montone sequence of continuous functions which converge to $f$ pointwise on $A$. Then $f_n$ converges to $f$ uniformly.

\begin{proof}
	We consider two cases.
	\begin{enumerate}
		\item Case 1: $g(s)$ has an upper bound on $(0,1)$, namely $M=sup_{(0,1)}g(s)<\infty$.
		
		Define $g(0)=\underset{s\rightarrow 0}{\lim\inf}\; g(s)$ and $g(1)=\underset{s\rightarrow 1}{\lim\sup}\;g(s)$, thus
		$g(s)$  is continuous and bounded on $[0,1]$. 
		\citep{Bernstein1912} states that given any $\epsilon>0$, there exists $g_m(s)$ of order m taking the form 
		\begin{equation*}
		g_m(s)=\sum_{k=1}^m w_k {m-1\choose k-1} s^{k-1} (1-s)^{m-k}, w_k\geq 0, k=1,\cdots, m; 
		\end{equation*}
		so that $||g_m(s)-g(s)||_\infty <\epsilon$. Then the corresponding approximated covariance function computed through $C_m(h)=\int_0^1 s^{h^2}g_m(s)ds$ has the form
		\begin{equation*}
		C_m(h)=\sum_{k=1}^m w_k \frac{1}{m}\frac{Beta(k+h^2,m-k+1)}{Beta(k,m-k+1)}, w_k\geq 0, k=1,\cdots,m.
		\end{equation*}
		Apparently, $C_m\in \mathcal{A}_m$.
		\begin{enumerate}
			\item For the first part of \textbf{Approximation Theorem}, 
			one can derive	
			\begin{equation}
			|C(h)-C_{m}(h)| 
			\leq |\int_0^1 s^{h^2} ||g(s)-g_m(s)||_\infty ds|
			= \epsilon \frac{1}{h^2+1}.\label{eq:C(h)-C_m(h)} 
			\end{equation}
			for any $h\geq 0$. Hence, for any $\epsilon>0$ there exists $C_m \in\mathcal{A}_m$ such that $||C-C_m||_\infty \leq \epsilon$ when $g(s)$ has an upper bound on $(0,1)$.
			
			\item Consider the second part. First we notice $||C||_2<\infty$ for the reason that $C(h)\leq M\frac{1}{h^2+1}$ leads to $||C||_2\leq M (\int_{0}^{\infty}\frac{1}{(h^2+1)^2}dh)^{\frac{1}{2}}<M\frac{\sqrt{\pi}}{2}$. By (\ref{eq:C(h)-C_m(h)}),
			\begin{align*}
			||C-C_m||_2 &\leq \epsilon (\int_0^\infty (\frac{1}{h^2+1})^2dh)^{\frac{1}{2}}=\epsilon(\frac{h}{2(1+h^2)}+\frac{arctan h}{2}\bigg|_{0}^{\infty})^{\frac{1}{2}}=\epsilon\frac{\sqrt{\pi}}{2}.
			\end{align*}
			Hence, for any $\epsilon^\prime=\epsilon\frac{\sqrt{\pi}}{2}>0$ there exists $C_m \in\mathcal{A}_m$ such that $||C-C_m||_2 \leq \epsilon^\prime$ when $g(s)$ has an upper bound on $(0,1)$.
		\end{enumerate}

		\item Case 2: $g(s)$ has no upper bound on $(0,1)$.
		
		Define a sequence of continuous and bounded functions $\{g_{K}\}_{K=1}^\infty$ as $g_{K}(s)=I_{\{g(s)<K\}}g(s)+I_{\{g(s)\geq K\}}K$. Also define $C_{K}(h)=\int_0^1 s^{h^2}g_{K}(s)ds$ for $K\in\mathbf{N}^+$. Since $\int_{0}^{1}g_{K}(s)ds$ is finite for $K\in\mathbf{N}^+$, $\{C_{K}\}_{K=1}^\infty$ are valid isotropic continuous covariance functions in $\mathbb{R}^\infty$. According to Monotone Convergence theorem, $\{C_K\}_{K=1}^\infty\subset \mathcal{A}$ is a monotonic increasing sequence of functions and converge to $C(h)$ pointwise on $[0,+\infty)$.
		
		Note that  $\underset{h\rightarrow\infty}{\lim} C(h)=0$ because $C(h)=\int_{0}^{1}s^{h^2}g(s)ds$ for $h\geq 0$ is upper bounded by value $C(0)=\int_{0}^{1}s^0 g(s)ds$ and 
		under Dominated Convergence Theorem one have $\underset{h\rightarrow\infty}{\lim}\int_0^1 s^{h^2}g(s)ds = \int_0^1 \underset{h\rightarrow\infty}{\lim}s^{h^2}g(s)ds =0$.  Similarly,  $\underset{h\rightarrow\infty}{\lim}C_K(h)=0$ for $K\in \mathbf{N}^+.$
		Define $C(+\infty)=0$ and $C_K(+\infty)=0,K\in\mathbf{N}^+$. 
		Hence, $\{C_{K}\}_{K=1}^\infty$ is a monotonic increasing sequence of functions and converge to continuous function $C(h)$ pointwise on the compact set $[0,+\infty]$. 
		
		\begin{enumerate}
			\item Consider $L_\infty$ approximation. Following from Dini Theorem, one obtain 
			\begin{align*}
			\underset{K\rightarrow\infty}{\lim}||C-C_{K} ||_\infty=0. 
			\end{align*}
			As a result, for any $\epsilon>0$, there exists $C_{K}$ such that $||C_{K} - C||_\infty<\frac{\epsilon}{2}$.  For such $C_{K}\in \mathcal{A}$, $g_{K}(s)$ has an upper bound $K$ on $[0,1]$. As shown above, there exists $g_{K,m}(s)$ such that
			$||g_{K}-g_{K,m}||_{\infty} <\frac{\epsilon}{2}$ and accordingly $||C_{K}-C_{K,m}||_\infty\leq \frac{\epsilon}{2}$ where $C_{K,m}(h)=\int_0^1 s^{h^2}g_{K,m}(s)ds$.
			By triangle inequality, we obtain $||C-C_{K,m}||_\infty \leq ||C_{K}-C_{K,m}||_\infty + ||C-C_{K}||_\infty\leq \epsilon$.  
			
			\item Consider the second part.  It follows from  $\underset{K\rightarrow\infty}{\lim}C_K(h)=C(h)$  that 
			$\underset{K\rightarrow\infty}{\lim} (C(h)-C_K(h))=0$ for $h\in[0,+\infty]$. 
			Given $||C||_2<\infty$, one have 
			$||C-C_K||_2<\infty$ because $||C-C_K||_2 \leq ||C||_2+||C_K||_2 \leq 2||C||_2 $. By Dominated convergence theorem it suffices to show that  $\underset{K\rightarrow\infty}{\lim}\int_{0}^{\infty} (C(h)-C_K(h))^2dh=0$ and thus $\underset{K\rightarrow\infty}{\lim}|| C-C_K ||_2=0$. 
			For any $\epsilon>0$, there exists $C_K$ so that $||C-C_K||_2<\frac{\epsilon}{2}$. As proved above, for $C_K$ with bounded $g_K(s)$, there exists $C_{K,m}\in\mathcal{A}_m$ satisfying $||C_{K,m}-C_m||_2$. Hence, 
			$||C-C_{K,m}||_2 \leq ||C_{K}-C_{K,m}||_2 + ||C_{0}-C_{K}||_2 \leq \epsilon$
		\end{enumerate}
	\end{enumerate} 
	In sum,we concludes the proof of Approximation Theorem. 
\end{proof}

\subsection{Proof of $L_q$ Approximation in Remark \ref{remark:L_q approximation} }
\label{subsec: Proof of L_q Approximation}

\begin{proof}
	First we notice that  for any $m\in \mathbf{N}^+$, $||A_{m,m}||_q^q = \int_{0}^{\infty}(1+\frac{h^2}{m})^{-q}dh<\infty$ when $q> \frac{1}{2}$. Thus $A_{k,m}\in L_q$ for $q>\frac{1}{2}$ and $1\leq k\leq m$.
	
	Similar to the proof in Appendix \ref{subsec: proof of covariance approximation},  we consider two cases to verify $L_q$ approximation.
	\begin{enumerate}
		\item Case 1: $g(s)$ has an upper bound on $(0,1)$, namely $M=sup_{(0,1)}g(s)<\infty$. Clearly, $C(h)\in L_q$ for the reason that $C(h)=\int_{0}^{1} s^{h^2}g(s)ds\leq \frac{M}{h^2+1}$.  For any $\epsilon$, one can obtain $C_m$ and $g_m$ satisfying (\ref{eq:C(h)-C_m(h)}). Further, one have
		\begin{equation*}
		\int_{0}^{\infty} (C(h)-C_m(h))^q dh \leq   \int_{0}^{\infty} \epsilon^q\frac{1}{(h^2+1)^q} dh = \epsilon^q || \frac{1}{h^2+1} ||_q^q.
		\end{equation*}
		Hence,  for any $\epsilon^\prime = \epsilon|| \frac{1}{h^2+1} ||_q>0$ there exists $C_m\in\mathcal{A}_m$ such that $||C-C_m ||_q <\epsilon^\prime.$
		
		\item Case 2: $g(s)$ has no upper bound on $(0,1).$
		
		For $\{C_K\}_{K=1}^\infty$ defined in Appendix 	\ref{subsec: proof of covariance approximation}, $\{C-C_K\}_{K=1}^\infty$ converge to $0$ pointwise on $[0,+\infty)$ and are upper bounded by $2C_0\in L_q$. Under Dominated convergence theorem it holds that $\underset{K\rightarrow\infty}
		{lim}\int_{0}^{\infty}(C(h)-C_K(h))^qdh = 0$ and thus $\underset{K\rightarrow\infty}{lim}||C-C_K||_q = 0$.
		Accordingly, given any $\epsilon>0$, there exists $C_K$ satisfying $||C-C_K||_q <\frac{\epsilon}{2}$. As mentioned above $g_K$ has an upper bound, where $g_K$ satisfies $C_K=\int_{0}^{1}s^{h^2}g_K(s)ds$. One can always find $C_{K,m}\in\mathcal{A}_m$ such that $||C_K-C_{K,m} ||_q < \frac{\epsilon}{2}$. By triangle inequality, we obtain
		$||C -C_{K,m}|| \leq||C -C_{m}|| + ||C_m -C_{K,m}||\leq \epsilon. $
	\end{enumerate}
	In conclusion, given any $C\in L_q, q > \frac{1}{2}$, for any $\epsilon>0$, there exists a function $C_m\in \mathcal{A}_m$ so that $||  C-C_m||_q<\epsilon$.
\end{proof}

\subsection{Proof of Corollary \ref{corollary: basis of A_m}}
\begin{proof}
	It is sufficient to prove that if there exists $\{w_k\}_{k=1}^{m}$ satisfying $\sum_{k=1}^m w_kA_{k,m}(h)\equiv0$ for $h\in[0,\infty)$, then $w_1\equiv w_2\cdots w_m\equiv0.$ 
	
	Consider $C\in \mathcal{A}$ in the form of 
	\begin{equation}
	C(h) = \int_0^\infty exp(-r^2h^2) f(r)dr\label{eq: Bochner isotropic Inf eq new}.
	\end{equation}
	Letting $\tilde{C}(h)=C(\sqrt{h})$ and $\tilde{f}(r)=f(\sqrt{r})\frac{1}{2}r^{-\frac{1}{2}}$, one can get the laplace transformation:
	\begin{equation*}
	\tilde{C}(h)=\int_0^\infty exp(-sh)\tilde{f}(s)ds.
	\end{equation*}
	Likewise, for $A_{k,m}(h)$ with spectral density
	$f_{k,m}(r) =\frac{2}{Beta(k,m-k+1)}r e^{-kr^2}(1-e^{-r^2})^{m-k}$ in (\ref{eq: Bochner isotropic Inf eq new}), we obtain the laplace transformation:
	\begin{equation}
	\tilde{A}_{k,m}(h)=\int_0^\infty exp(-sh)\tilde{f}_{k,m}(s)ds,\label{eq:laplace transformation basis function}
	\end{equation}
	where $\tilde{A}_{k,m}(h)=A_{k,m}(\sqrt{h})\mbox{ and }\tilde{f}_{k,m}(r)=f_{k,m}(\sqrt{r})\frac{1}{2}r^{-\frac{1}{2}}.$
	
	$\sum_{k=1}^m w_kA_{k,m}(h)\equiv0$ on $ [0,\infty)$ implies $ \sum_{k=1}^m w_k\tilde{A}_{k,m}(h)\equiv0\mbox{ on } [0,\infty)$.  Further, notice that $\sum_{k=1}^m w_k\tilde{A}_{k,m}(h)=\int_0^\infty exp(-sh) \sum_{k=1}^m w_k \tilde{f}_{k,m}(s)ds\equiv0$  by (\ref{eq:laplace transformation basis function}).  Due to uniqueness of laplace transformation, it holds that $\sum_{k=1}^m w_{k} \tilde{f}_{k,m}(r)\equiv 0$ for $r\in[0,\infty)$.  Further, we derive
	\begin{eqnarray}
	\sum_{k=1}^m w_k\tilde{f}_{k,m}(r)&=&\sum_{k=1}^m w_k f_{k,m}(\sqrt{r})\frac{1}{2}r^{-\frac{1}{2}} \\
	&=&\sum_{k=1}^m w_k \frac{1}{Beta(k,m-k+1)}\cdot exp(-kr)(1-exp(-r))^{m-k} \label{eq: eq 1}
	\end{eqnarray}
	Replacing $s$ with  $exp(-r)$ in (\ref{eq: eq 1}) gives
	\begin{align*}
	\sum_{k=1}^m w_k \frac{2}{Beta(k,m-k+1)}\cdot s^{k}(1-s)^{m-k}\equiv 0.
	\end{align*}
	Since Bersterin polynomials are linearly independent, $w_1\equiv w_2\equiv \cdots\equiv 0$. 
	We finish the proof.
\end{proof}

\subsection{ Proof of Corollary \ref{corollary: increasing A_m}}
\begin{proof}		
	\begin{enumerate}
		\item First, we show $\mathcal{A}_m\subset\mathcal{A}$, for $m\in \mathbf{N}^+$. It is equivalent to show that for any $w_k\geq 0,k=1,\cdots,m$, we have $C_m(h)=\sum_{k=1}^m w_k\frac{1}{m} A_{k,m}\in \mathcal{A}$. Recall that such $C_m(h)$ can be written as 
		\begin{align}
		C_m(h)=\int_0^\infty exp(-r^2h^2)f_m(r)dr,\label{eq:C_m to f_m}
		\end{align}
		where $f_m(r)=2rexp(-r^2)\sum_{k=1}^m w_k {m-1\choose k-1} exp(-r^{2}(k-1)) (1-exp(-r^2))^{m-k} $.
		Since $\int_0^\infty f_m(r)dr = \int_0^1 g_m(s)ds= \sum_{k=1}^m w_k {m-1\choose k-1}Beta(k,m-k+1)<\infty$,
		by definition of $\mathcal{A}$, $C_m\in\mathcal{A}$.
		\item Next we show $\mathcal{A}_m \subset \mathcal{A}_{m+1}$ for $m\in \mathbf{N}^+$. 
		For any $C_m(h)=\sum_{k=1}^m w_k\frac{1}{m} A_{k,m}\in \mathcal{A}$ with $g_m(s)$ in the form
		$g_m(s)=\sum_{k=1}^m w_k {m-1\choose k-1} s^{k-1} (1-s)^{m-k}$,
		one can write
		\begingroup
		\allowdisplaybreaks
		\begin{align*}
		g_m(s) &=  \sum_{k=1}^m w_k {m-1 \choose k-1} s^{k-1} (1-s)^{m-k}\\
		&= \sum_{k=1}^m w_k {m-1 \choose k-1} s^{k-1} (1-s)^{m-k}(1-s+s)\\
		&=  \sum_{k=1}^m w_k {m-1 \choose k-1} s^{k-1} (1-s)^{m-k+1}  + \sum_{k=1}^m w_k {m-1 \choose k-1} s^{k} (1-s)^{m-k}\\
		&= \sum_{k=1}^{m+1} c_k {m \choose k-1} s^{k-1} (1-s)^{m+1-k},\\
		\mbox{ where } c_k &= 
		\begin{cases} 
		w_1, &\mbox{if }k=1;\\
		\frac{w_k {m-1 \choose k-1} + w_{k-1} {m-1 \choose k-2}}{ {m \choose k-1} }, & \mbox{if } 2\leq k\leq m; \\ 
		w_m, & \mbox{if } k=m+1 .
		\end{cases}	
		\end{align*}
		\endgroup
		Further, we obtain
		$C_m(h) = \int_0^1 s^{h^2}   \sum_{k=1}^{m+1} c_k {m \choose k-1} s^{k-1} (1-s)^{m+1-k} ds
		=\sum_{k=1}^{m+1} c_k\frac{1}{m+1}A_{k,m+1}$
		with $c_k \geq 0, k=1,\cdots,m+1$, indicating $C_m(h)\in \mathcal{A}_{m+1}$. Therefore, $\mathcal{A}_m\subset\mathcal{A}_{m+1}$.
	\end{enumerate}
\end{proof}

\subsection{Proof of Corollary \ref{corollary: compact A_m}}
\begin{proof}
	Since Corollary 1 implies that there is a continuous bijection between $\mathcal{A}_m$ and the set $\{\bm{w}\in \mathcal{R}^m: w_k\geq 0\mbox{ for }1\leq k\leq m\}$, $\mathcal{A}_m$ is closed under $L_\infty$ norm and $L_q$ norm ($q>1$).
\end{proof}

\subsection{Proof of Consistency Theorem \ref{theorem:estimation consistency}}

First, some useful lemmas are provided for later verification.
\begin{lemma}	\label{lemma:derivative of logdet} 
	$\partial (logdet(\mathbf{X})) = tr(\mathbf{X}^{-1}\partial \mathbf{X})$ 
\end{lemma}
\begin{lemma}\label{lemma:derivative of trace}
	$\partial(tr(\mathbf{X}))) = tr(\partial\mathbf{X})$ 
\end{lemma}
\begin{lemma}\label{lemma:derivative of inverse}
	$\partial(\mathbf{X}^{-1}) = -\mathbf{X}^{-1}(\partial\mathbf{X})\mathbf{X}^{-1}$ 
\end{lemma}
\begin{lemma}\label{lemma:derivative of product}
	$\partial(\mathbf{X}\mathbf{Y}) = (\partial\mathbf{X}) \mathbf{Y} + \mathbf{X}(\partial\mathbf{Y})$ 
\end{lemma}
\begin{lemma}\label{lemma: eigen and trace}
	$\mathbf{A}$ is a semipositive definite matrix and $\mathbf{B}$ is a symmetric matrix. Then $\lambda_{min}(\mathbf{B})tr(\mathbf{A}) \leq tr(\mathbf{A}\mathbf{B}) \leq \lambda_{max}(\mathbf{B})tr(\mathbf{A})$
\end{lemma}

\begin{lemma}[Theorem 5.3.4 in \citep{Topics1991}]\label{lemma: eigen range of A1 to Am}
	For n by n non-negative definite matrices $\mathbf{A}$ and $\mathbf{B}$,  Schur product denoted as $\circ$ is element-wise matrix product. Then 
	\begin{align*}
	\lambda_{min}(\mathbf{A})\lambda_{min}(\mathbf{B}) & \leq \lambda_{min}(\mathbf{A}) \underset{k}{min}\mathbf{B}_{[k,k]} \leq \lambda_{min}(\mathbf{A}\circ \mathbf{B})\\
	\lambda_{max}(\mathbf{A} \circ \mathbf{B})& \leq \lambda_{max}(\mathbf{A}) \underset{k}{max}\mathbf{B}_{[k,k]} \leq \lambda_{max}(\mathbf{A})\lambda_{max}(\mathbf{B})
	\end{align*}	
\end{lemma}

Before we proceed, we need to demonstrate the following claims.

\begin{claim}\label{claim: order of A_k^m}The two inequality below hold. For $n\in\mathbb{N}^+$,
	\begin{align*}
	\lambda_{max}(\mathbf{A}_{1,m,n})&\leq \lambda_{max}(\mathbf{A}_{2,m,n})\cdots\leq \lambda_{max}(\mathbf{A}_{m,m,n})\\
	\lambda_{min}(\mathbf{A}_{m,m,n})&\leq \lambda_{min}(\mathbf{A}_{2,m,n})\cdots\leq \lambda_{min}(\mathbf{A}_{1,m,n})
	\end{align*}
\end{claim}
\begin{proof}
	For $1 \leq k\leq m$, define matrix $\{\mathbf{B}_{k,m,n}\}_{k=1}^m$ with $(i,j)$ element
	$\mathbf{B}_{k,m,n[ij]} = \frac{k}{k+ ||s_i-s_j||_2^2}$ and thus $\mathbf{A}_{k,m,n} = \mathbf{B}_{k,m,n}\circ \mathbf{B}_{k+1,m,n} \circ \cdots \circ \mathbf{B}_{m,m,n}$, where $\circ$ is Schur product.
	Note that $\{\mathbf{B}_{k,m,n}\}^m_{k=1}$ are nonnegative definite and diag$(\mathbf{B}_{k,m,n})=(1,\cdots,1)^T.$ By Lemma \ref{lemma: eigen range of A1 to Am}, $	\lambda_{max}(\mathbf{A}_{k-1,m,n})=\lambda_{max}(\mathbf{A}_{k,m,n}\circ\mathbf{B}_{k-1,m,n}) \leq \lambda_{max}(\mathbf{A}_{k,m,n})$ and  $	\lambda_{min}(\mathbf{A}_{k-1,m,n})=\lambda_{min}(\mathbf{A}_{k,m,n}\circ\mathbf{B}_{k-1,m,n}) \geq \lambda_{min}(\mathbf{A}_{k,m,n})$ for $2\leq k\leq m$. Therefore, Claim \ref{claim: order of A_k^m} holds.
\end{proof}

\begin{claim}\label{claim: max and min eigen of C and derivative} For $C\in\bm{\Omega}_m$ with $C=\sum_{k=1}^{m}w_k A_{k,m}$ and $n\in\mathbb{N}^+$,
	\begin{align*}
	&L_C\cdot\lambda_{min}(\mathbf{A}_{m,m,n})\leq \lambda(\mathbf{\Sigma}_{C,n})\leq U_{C}\cdot \lambda_{max}(\mathbf{A}_{m,m,n}),\\
	&\lambda_{min}(\mathbf{A}_{m,m,n})\leq \lambda(\frac{\partial \mathbf{\Sigma}_{C,n}}{\partial w_k})\leq  \lambda_{max}(\mathbf{A}_{m,m,n}),
	\end{align*}
	where $\lambda(\mathbf{\Sigma}_{C,n})$ is the eigenvalue of $\mathbf{\Sigma}_{C,n}$.
\end{claim}
\begin{proof}
	One can easily prove using Assumption (\ref{assumption: compact and convex}) 
	and Claim \ref{claim: order of A_k^m}.
\end{proof}

We use $l_n(\bm{\theta})$ and $E_0 l_n(\bm{\theta})$ to denote log-likelihood function and its expectation under $P_{\bm{\theta}_0}$,
where explicit dependence on the data is suppressed for notional convenience.
For $\bm{\theta}=(\bm{\beta},C)\in\bm{\Omega}$, 
\begin{align}
l_{n}(\bm{\theta})=&-\frac{n}{2}log(2\pi)-\frac{1}{2}logdet(\mathbf{\Sigma}_{C,n})-\frac{1}{2}(\bm{Y}_n-\mathbf{X}_n\bm{\beta})^T[\mathbf{\Sigma}_{C,n}]^{-1}(\bm{Y}_n-\mathbf{X}_n\bm{\beta}),
\label{eq:log likelihood}
\end{align}
Now factorize the true covariance matrix as  $\mathbf{\Sigma}_{0,n}=\mathbf{\Gamma}_n^T\mathbf{\Gamma}_n$ so that  $\bm{Z}_n\sim N(\bm{0},\mathbf{I}_n)$, where
\begin{align}
\bm{Z}_n=[\mathbf{\Gamma}_n^T]^{-1}(\bm{Y}_n-\mathbf{X}_n\bm{\beta}_0). \label{eq:Z_n equation}
\end{align}
Further (\ref{eq:log likelihood}) can be rewritten as
\begin{align*}
l_{n}(\bm{\theta}) 
=&-\frac{n}{2}log(2\pi)-\frac{1}{2}logdet(\mathbf{\Sigma}_{C,n})-\frac{1}{2}(\bm{\beta}_0-\bm{\beta})^T\mathbf{X}_n^{T} \mathbf{\Sigma}_{C,n}^{-1} \mathbf{X}_n(\bm{\beta}_0-\bm{\beta})
\\ 
&-\frac{1}{2}\bm{Z}^T_n \mathbf{\Gamma}_n \mathbf{\Sigma}_{C,n}^{-1}\mathbf{\Gamma}_n^{T}\bm{Z}_n- (\bm{\beta}_0-\bm{\beta})^T\mathbf{X}_n^{T}\mathbf{\Sigma}_{C,n}^{-1} \bm{\Gamma}_n^T \bm{Z}_n
\\
=&-\frac{n}{2}log(2\pi)-\frac{1}{2}logdet(\mathbf{\Sigma}_{C,n}) -\frac{1}{2}(\bm{\beta}_0-\bm{\beta})^T \mathbf{X}_n^{T} \mathbf{\Sigma}_{C,n}^{-1} \mathbf{X}_n(\bm{\beta}_0-\bm{\beta})
\\
&-\frac{1}{2}\sum_{i=1}^{n}\lambda_{i,n}\chi_i^2- (\bm{\beta}_0-\bm{\beta})^T\mathbf{X}_n^{T}[\mathbf{\Sigma}_{C,n}]^{-1} \bm{\Gamma}_n^T \bm{Z}_n ,
\end{align*}
where $\lambda_{1,n},\cdots, \lambda_{n,n}$ are the eigenvalues of $\mathbf{\Gamma}_n [\mathbf{\Sigma}_{C,n}]^{-1}\mathbf{\Gamma}_n^T$ and $\chi_i^2$ are $ iid$ random variable in $ \chi^2(1)$ distribution. 
Now that $\sum_{i=1}^{n}\lambda_{i,n}=tr(\mathbf{\Gamma}_n \mathbf{\Sigma}_{C,n}^{-1}\mathbf{\Gamma}_n^T)=tr(\mathbf{\Sigma}_{C,n}^{-1}\mathbf{\Sigma}_{0,n})$, 
\begin{align}
l_n(\bm{\theta})=&-\frac{n}{2}log(2\pi)-\frac{1}{2}logdet( \mathbf{\Sigma}_{C,n})-\frac{1}{2}(\bm{\beta}_0-
\bm{\beta})^T \mathbf{X}_n^{T} \mathbf{\Sigma}_{C,n}^{-1} \mathbf{X}_n(\bm{\beta}_0-
\bm{\beta})   \nonumber
\\
&-\frac{1}{2}tr(\mathbf{\Sigma}_{C,n}^{-1}\mathbf{\Sigma}_{0,n})-\frac{1}{2}\sum_{i=1}^{n}\lambda_{i,n}(\chi_i^2-1) -(\bm{\beta}_0-\bm{\beta})^T\mathbf{X}_n^{T}[\mathbf{\Sigma}_{C,n}]^{-1} \mathbf{\Gamma}_n^T \bm{Z}_n.
\label{eq:form of log likelihood}
\end{align}
And it is apparent that 
\begin{align}
E_0l_n(\bm{\theta})
=&-\frac{n}{2}log(2\pi)-\frac{1}{2}logdet(\mathbf{\Sigma}_{C,n})-\frac{1}{2}(\bm{\beta}_0-
\bm{\beta})^T \mathbf{X}_n^{T} \mathbf{\Sigma}_{C,n}^{-1} \mathbf{X}_n(\bm{\beta}_0-
\bm{\beta})  \nonumber
\\
&-\frac{1}{2}tr(\mathbf{\Sigma}_{C,n}^{-1}\mathbf{\Sigma}_{0,n}).
\label{eq: expectation of log likelihood}
\end{align}
For $m\in\mathbb{N}^+$, consider $\bm{\theta}=(\bm{\beta},\sum_{k=1}^{m}w_kA_{k,m})\in\bm{\Theta}_m$, the log-likelihood function can be regarded as a continuous function of $(\bm{\beta},\bm{w})$. 
The first derivative of $l_n$ with respect to $(\bm{\beta},\bm{w})$ is denoted as $(l_{n\bm{\beta}}^T,l_{n\bm{w}}^T)^T$. It is easy to get
\begin{align}
l_{n\bm{\beta}}&=\mathbf{X}_n^T \mathbf{\Sigma}_{C,n}^{-1} (\bm{Y}_n - \mathbf{X}_n\bm{\beta}) \label{eq:derivative of beta}\\
l_{nw_k} &=\frac{1}{2}(\bm{Y}_n-\mathbf{X}_n\bm{\beta})^{T} \mathbf{\Sigma}_{C,n}^{-1}\frac{\partial \mathbf{\Sigma}_{C,n}}{\partial w_k} \mathbf{\Sigma}_{C,n}^{-1} (\bm{Y}_n-\mathbf{X}_n\bm{\beta}) -\frac{1}{2}tr(\mathbf{\Sigma}_{C,n}^{-1}\frac{\partial \mathbf{\Sigma}_{C,n}}{\partial w_k}), \label{eq:derivative of w}
\end{align}
where $\frac{\partial \mathbf{\Sigma}_{C,n}}{\partial w_k}=\mathbf{A}_{k,m,n}$ for $k=1,\cdots,m.$
Plugging (\ref{eq:Z_n equation}) in these two equations above gives us
\begin{align}
l_{n\bm{\beta}}=&\mathbf{X}_n^T [\mathbf{\Sigma}_{C,n}]^{-1} \mathbf{\Gamma}_n^T\bm{Z}_n+\mathbf{X}_n^T \mathbf{\Sigma}_{C,n}^{-1}\mathbf{X}_n(\bm{\beta}_0-\bm{\beta}) \label{eq: first derivative of beta}
\\
l_{nw_k} =&
\frac{1}{2}\bm{Z}_n^T\cdot \mathbf{\Gamma}_n \mathbf{\Sigma}_{C,n}^{-1} \frac{\partial \mathbf{\Sigma}_{C,n}}{\partial w_k}  \mathbf{\Sigma}_{C,n}^{-1} \mathbf{\Gamma}_n^T\cdot \bm{Z}_n+
(\bm{\beta}_0-\bm{\beta})^T\mathbf{X}_n^T  \mathbf{\Sigma}_{C,n}^{-1} \frac{\partial \mathbf{\Sigma}_{C,n}}{\partial w_k} \mathbf{\Sigma}_{C,n}^{-1} \mathbf{\Gamma}_n^T\bm{Z}_n+ \nonumber\\
&
\frac{1}{2}(\bm{\beta}_0-\bm{\beta})^T \mathbf{X}_n^T \mathbf{\Sigma}_{C,n}^{-1} \frac{\partial \mathbf{\Sigma}_{C,n}}{\partial w_k}\mathbf{\Sigma}_{C,n}^{-1} \mathbf{X}_n (\bm{\beta}_0-\bm{\beta})-\frac{1}{2} tr(\mathbf{\Sigma}_{C,n}^{-1}\frac{\partial \mathbf{\Sigma}_{C,n}}{\partial w_k}).
\nonumber
\end{align}
Similar to (\ref{eq:form of log likelihood}), one have
\begin{align}
l_{nw_k} =&\frac{1}{2}\sum_{j=1}^{n}\lambda_{j,k,n} (\chi_j^2-1)+\frac{1}{2}tr(\mathbf{\Sigma}_{C,n}^{-1} \frac{\partial \mathbf{\Sigma}_{C,n}}{\partial w_k} \mathbf{\Sigma}_{C,n}^{-1} \mathbf{\Sigma}_{0,n})
+ (\bm{\beta}_0-\bm{\beta})^T\mathbf{X}_n^T  \mathbf{\Sigma}_{C,n}^{-1} \frac{\partial \mathbf{\Sigma}_{C,n}}{\partial w_k} \mathbf{\Sigma}_{C,n}^{-1} \mathbf{\Gamma}_n^T \bm{Z}_n+ \nonumber
\\
&\frac{1}{2}(\bm{\beta}_0-\bm{\beta})^T \mathbf{X}_n^T \mathbf{\Sigma}_{C,n}^{-1} \frac{\partial \mathbf{\Sigma}_{C,n}}{\partial w_k} \mathbf{\Sigma}_{C,n}^{-1}  \mathbf{X}_n (\bm{\beta}_0-\bm{\beta})-\frac{1}{2} tr(\mathbf{\Sigma}_{C,n}^{-1}\frac{\partial \mathbf{\Sigma}_{C,n}}{\partial w_k}),
\label{eq: first derivative of w}
\end{align}
where $\{\lambda_{j,k,n},j=1,\cdots,n\}$ are the eigenvalues of $\mathbf{\Gamma}_n\mathbf{\Sigma}_{C,n}^{-1}\frac{\partial \mathbf{\Sigma}_{C,n}}{\partial w_k} \mathbf{\Sigma}_{C,n}^{-1} \mathbf{\Gamma}_n^T$ satisfying $\sum_{j=1}^{n} \lambda_{j,k,n}=tr(\mathbf{\Gamma}_n\mathbf{\Sigma}_{C,n}^{-1} \frac{\partial \mathbf{\Sigma}_{C,n}}{\partial w_k} \mathbf{\Sigma}_{C,n}^{-1} \mathbf{\Gamma}_n^T )=tr(\mathbf{\Sigma}_{C,n}^{-1} \frac{\partial \mathbf{\Sigma}_{C,n}}{\partial w_k} \mathbf{\Sigma}_{C,n}^{-1} \mathbf{\Sigma}_{0,n})$
and $\chi^2_j$ are $i.i.d$ random variables in $\chi^2(1)$ distribution.


We start by showing Claim \ref{claim: uniform convergence} that $\frac{1}{n}l_n(\bm{\theta})$ converges uniformly to its expectation among $\bm{\theta}_m\in\bm{\Theta_m}$ for any $m\in\mathbb{N}^+$ and Claim \ref{claim: unique global optimizer} that the expectation has a unique maximizer at $\bm{\theta}_0$. Then we can get $E_0\frac{1}{n}l_n(\hat{\bm{\theta}}_{m,n})-E_0\frac{1}{n}l_n(\bm{\theta}_0)\overset{P_{\bm{\theta}_0}}\longrightarrow 0$ and further prove $\bar{d}(\hat{\bm{\theta}}_{m,n},\bm{\theta}_0)\overset{P_{\bm{\theta}_0}}\longrightarrow 0$.
\begin{claim}[Uniform convergence] \label{claim: uniform convergence}
	Under Assumption (\ref{assumption: compact and convex})(\ref{assumption: X^TX})(\ref{assumption:max and min eigen for true})(\ref{assumption:max and min eigen for A_m,m})
	it holds that
	\begin{align}
	\underset{n\rightarrow\infty}{lim} \underset{\bm{\theta}\in\bm{\Theta}_m}{sup} \frac{1}{n}l_n(\bm{\theta})-E_0 \frac{1}{n}l_n(\bm{\theta})=0, \mbox{ for any }m\in\mathbb{N}^+.
	\label{eq:uniform convergence}
	\end{align}
\end{claim}

\begin{proof}
	For fixed $m\in\mathbb{N}^+$, there is one-to-one correspondence between $\bm{\Omega}_m$ and the vector space $\bm{w}_m\in \{(w_1,\cdots,w_m): L_C \leq \sum_{k=1}^{m}w_k\leq U_C,w_k\geq 0,k=1,\cdots,m\}$. For $\bm{\theta}\in\bm{\Theta}_m$ both $\frac{1}{n}l_n(\bm{\theta})$ and 
	$E_0\frac{1}{n}l_n(\bm{\theta})$ can be regarded as continuous functions of $(\bm{\beta},\bm{w}_m)$.
	Therefore it suffices to verify the three conditions below and apply Corollary 2.2 in \citep{Newey1991} to obtain the uniform convergence in (\ref{eq:uniform convergence}). 
	\begin{enumerate}[label=B\arabic*]
		\item \label{condition: compact} (Compactness) The space $\{(\bm{\beta},\bm{w}_m):(\bm{\beta},\sum_{k=1}^m w_kA_{k,m})\in \bm{\Theta}_m)\}$ is compact.
		\item \label{condition: pointwise} (Pointwise convergence) For each $\bm{\theta}\in \{(\bm{\beta},\bm{w}_m):(\bm{\beta},\sum_{k=1}^m w_kA_{k,m})\in \bm{\Theta}_m)\}$, $\frac{1}{n}l_n(\bm{\theta})-E_0\frac{1}{n}l_n(\bm{\theta})=o_{p}(1)$.
		\item \label{condition: stochastic equicontinuity} (Stochastic equicontinuity) The space $\{(\bm{\beta},\bm{w}_m):(\bm{\beta},\sum_{k=1}^m w_kA_{k,m})\in \bm{\Theta}_m)\}$ is convex. $\frac{1}{n}(l_{n\bm{\beta}}, l_{n\bm{w}})$ is bounded by some random variable $B_n$ satisfying $B_n=O_p(1)$.
	\end{enumerate}
	Under Assumption (\ref{assumption: compact and convex})
	, it is natural that $\{(\bm{\beta},\bm{w}_m):(\bm{\beta},\sum_{k=1}^m w_kA_{k,m})\in \bm{\Theta}_m)\}$ is convex and compact.
	By the equation (\ref{eq:form of log likelihood}) and (\ref{eq: expectation of log likelihood}), one can get
	\begin{align}
	\frac{1}{n}l_n(\bm{\theta})-E_0\frac{1}{n}l_n(\bm{\theta})
	=&-\frac{1}{n} (\bm{\beta}_0-\bm{\beta})^T\mathbf{X}_n^{T}\mathbf{\Sigma}_{C,n}^{-1} \mathbf{\Gamma}_n^T \bm{Z}_n -\frac{1}{2n}\sum_{i=1}^{n}\lambda_{i,n}(\chi_i^2-1). \label{eq: 1/n log likelihood-expectation}
	\end{align}
	We wish to show the right hand side of (\ref{eq: 1/n log likelihood-expectation}) converges to zero to get (\ref{condition: pointwise}). 
	If both $\underset{n}{\sup}\lambda_{max}(\mathbf{\Gamma}_n \mathbf{\Sigma}_{C,n}^{-1}\mathbf{\Gamma}_n^T)$ and 
	$\underset{n}{\sup}\frac{1}{n}||\mathbf{\Gamma}_n \mathbf{\Sigma}_{C,n}^{-1} \bm{X}_n  (\bm{\beta}_0-\bm{\beta})||_2^2$ are bounded, then we apply Theorem 1.1 of \citep{Cuzick1995} to get the results of almost sure convergence to zero. By Claim \ref{claim: max and min eigen of C and derivative} together with Assumption (\ref{assumption: X^TX}) and Assumption (\ref{assumption:max and min eigen for A_m,m})
	, it is true that
	\begin{align*}
	\lambda_{max}(\mathbf{\Gamma}_n\mathbf{\Sigma}_{C,n}^{-1}\mathbf{\Gamma}_n^T) \leq
	\lambda_{max}(\mathbf{\Gamma}_n\cdot\lambda_{max}(\mathbf{\Sigma}_{C,n}^{-1})\mathbf{I}_n\cdot\mathbf{\Gamma}_n) \leq  1/\lambda_{min}(\mathbf{\Sigma}_{C,n})\cdot\lambda_{max}(\mathbf{\Sigma}_{0,n}) <\infty
	\end{align*}
	and 
	\begin{align*}
	&\frac{1}{n}(\bm{\beta_0}-\bm{\beta})^T\mathbf{X}_n^{T}\mathbf{\Sigma}_{C,n}^{-1} \mathbf{\Gamma}_n^T \cdot\mathbf{\Gamma}_n\mathbf{\Sigma}_{C,n}^{-1} \bm{X}_n  (\bm{\beta}_0-\bm{\beta})  \nonumber\\
	\leq &\frac{1}{n}\lambda_{max}(\mathbf{X}_n^{T}\mathbf{\Sigma}_{C,n}^{-1}\mathbf{\Sigma}_{0,n} \mathbf{\Sigma}_{C,n}^{-1}\bm{X}_n)\cdot || \bm{\beta}_0-\bm{\beta} ||_2^2  \nonumber
	\\
	\leq &\frac{1}{n}\lambda_{max}(\mathbf{X}_n^T\mathbf{X}_n) \lambda_{max}(\mathbf{\Sigma}_{0,n})\cdot 1/\lambda^2_{min}(\mathbf{\Sigma}_{C,n})
	\cdot || \bm{\beta}_0-\bm{\beta} ||_2^2  \nonumber\\
	<&\infty . 
	\end{align*}
	Thus we get that
	\begin{align*}
	\frac{1}{n}l_n(\bm{\theta})-E_0\frac{1}{n}l_n(\bm{\theta})\overset{P_{\bm{\theta}_0}}{\longrightarrow} 0   \mbox{   pointwise.}
	\end{align*}
	
	Consider (\ref{condition: stochastic equicontinuity}). According to the equation (\ref{eq: first derivative of beta}),
	\begin{align}
	\frac{1}{n}l_{n\bm{\beta}}&=\frac{1}{n}\mathbf{X}_n^T \mathbf{\Sigma}_{C,n}^{-1} \mathbf{\Gamma}_n^T\bm{Z}_n+
	\frac{1}{n}\mathbf{X}_n^T \mathbf{\Sigma}_{C,n}^{-1}\mathbf{X}_n(\bm{\beta}_0-\bm{\beta}).
	\label{eq:1/n first derivative w.r.t beta}
	\end{align}
	We wish to verify in the right hand side of (\ref{eq:1/n first derivative w.r.t beta}) the first term converges to zero almost surely and the second term is bounded.
	Denote $\mathbf{X}_n^T \mathbf{\Sigma}_{C,n}^{-1}\mathbf{\Gamma}_n^T$ as $\mathbf{M}_n$ and the $r_{th}$ row as $\mathbf{M}_{nr\cdot}$ for $r=1,\cdots,p$ so that 
	\begin{align*}
	\mathbf{M}_n&=
	\begin{pmatrix}
	\mathbf{M}_{n1\cdot}^T\\
	\vdots \\
	\mathbf{M}_{np\cdot}^T
	\end{pmatrix}.
	\end{align*}
	By Theorem 1.1 in \citep{Cuzick1995}, if $\underset{n}{\sup}\frac{1}{n}||\mathbf{M}_{nr\cdot}||_2^2<\infty $, then
	\begin{align*}
	\begin{pmatrix}
	\frac{1}{n}\mathbf{M}_{n1\cdot}^T\mathbf{Z}_n\\
	\vdots \\
	\frac{1}{n}\mathbf{M}_{np\cdot}^T\mathbf{Z}_n
	\end{pmatrix}
	\overset{a.s}{\longrightarrow}
	\begin{pmatrix}
	0\\
	\vdots\\
	0
	\end{pmatrix}.
	\end{align*}

	Note that 
	$
	\underset{n}{\sup}\frac{1}{n}||\mathbf{M}_{nr\cdot}||_2^2=\underset{n}{\sup}\frac{1}{n} (\mathbf{M}_n\mathbf{M}_n^T)_{[r,r]}\leq \lambda_{max}(\frac{1}{n} \mathbf{X}_n^T\mathbf{\Sigma}_{C,n}^{-1}\mathbf{\Sigma}_{0,n}\mathbf{\Sigma}_{C,n}^{-1}\mathbf{X}_n)
	$, where $(\mathbf{M}_n\mathbf{M}_n^T)_{[r,r]}$ is the $(r,r)$ element of matrix $\mathbf{M}_n\mathbf{M}_n^T$. Since
	\begin{align*}
	&\frac{1}{n}\lambda_{max}( \mathbf{X}_n^T\mathbf{\Sigma}_{C,n}^{-1}\mathbf{\Sigma}_{0,n}\mathbf{\Sigma}_{C,n}^{-1}\mathbf{X}_n)
	\leq  \lambda_{max}(\mathbf{\Sigma}_{0,n})\lambda_{max}(\frac{1}{n}\mathbf{X}_n^T \mathbf{X}_n) \cdot1/\lambda^2_{min}(\mathbf{\Sigma}_{C,n})\\
	&||\frac{1}{n}\mathbf{X}_n^T \mathbf{\Sigma}_{C,n}^{-1}\mathbf{X}_n(\bm{\beta}_0-\bm{\beta})||_2^2 \leq \lambda^2_{max}(\frac{1}{n} \mathbf{X}_n^T \mathbf{X}_n)\cdot 1/\lambda^2_{min}(\mathbf{\Sigma}_{C,n}) \cdot|| \bm{\beta}_0-\bm{\beta} ||_2^2, 
	\end{align*}
	then by Claim \ref{claim: max and min eigen of C and derivative} and Assumption (\ref{assumption: X^TX})(\ref{assumption:max and min eigen for true})(\ref{assumption:max and min eigen for A_m,m}), both two terms are bounded. 
	
	Next from the equation (\ref{eq: first derivative of w}), one obtain
	\begin{align}
	\frac{1}{n}l_{nw_k}&=\frac{1}{2n}tr(\mathbf{\Sigma}_{C,n}^{-1} \frac{\partial \mathbf{\Sigma}_{C,n}}{\partial w_k} \mathbf{\Sigma}_{C,n}^{-1} \mathbf{\Sigma}_{0,n})+
	\frac{1}{2n}(\bm{\beta}_0-\bm{\beta})^T \mathbf{X}_n^T \mathbf{\Sigma}_{C,n}^{-1} \frac{\partial \mathbf{\Sigma}_{C,n}}{\partial w_k} \mathbf{\Sigma}_{C,n}^{-1}  \mathbf{X}_n (\bm{\beta}_0-\bm{\beta}) \nonumber\\
	&-\frac{1}{2n} tr(\mathbf{\Sigma}_{C,n}^{-1}\frac{\partial \mathbf{\Sigma}_{C,n}}{\partial w_k})
	+\frac{1}{n} (\bm{\beta}_0-\bm{\beta})^T\mathbf{X}_n^T  \mathbf{\Sigma}_{C,n}^{-1} \frac{\partial \mathbf{\Sigma}_{C,n}}{\partial w_k} \mathbf{\Sigma}_{C,n}^{-1} \mathbf{\Gamma}_n^T \bm{Z}_n + \frac{1}{2n}\sum_{j=1}^{n}\lambda_{j,k,n} (\chi_j^2-1).
	\label{eq:first derivative w.r.t theta}
	\end{align}
	If in the right hand side of (\ref{eq:first derivative w.r.t theta}) the last two terms converge to zero almost surely and the first three terms are bounded, $\frac{1}{n}l_{n\bm{w}_k}$ is bounded.

	For the last two terms, again by Theorem 1.1 in \citep{Cuzick1995}, we wish $\lambda_{max}(\mathbf{\Gamma}_n\mathbf{\Sigma}_{C,n}^{-1}\frac{\partial \mathbf{\Sigma}_{C,n}}{\partial w_k} \mathbf{\Sigma}_{C,n}^{-1} \mathbf{\Gamma}_n^T)$ and $\underset{n}{\sup}\frac{1}{n}||\mathbf{\Gamma}_n      \mathbf{\Sigma}_{C,n}^{-1} \frac{\partial \mathbf{\Sigma}_{C,n}}{\partial w_k}\mathbf{\Sigma}_{C,n}^{-1}\mathbf{X}_n(\bm{\beta}-\bm{\beta}_0)||_2^2$ are bounded, which is true given the reason that by Claim \ref{claim: max and min eigen of C and derivative} and Assumption (\ref{assumption: compact and convex})(\ref{assumption:max and min eigen for true})(\ref{assumption:max and min eigen for A_m,m}) 
	\begin{align*}
	&	\lambda_{max}(\mathbf{\Gamma}_n[\mathbf{\Sigma}_{C,n}]^{-1}\frac{\partial \mathbf{\Sigma}_{C,n}}{\partial w_k} [\mathbf{\Sigma}_{C,n}]^{-1} \mathbf{\Gamma}_n^T) \leq \lambda_{max}(\frac{\partial \mathbf{\Sigma}_{C,n}}{\partial w_k}) \lambda_{max}(\mathbf{\Sigma}_{0,n}) \lambda^{-2}_{min}(\mathbf{\Sigma}_{C,n})<\infty\\
	&	\frac{1}{n} (\bm{\beta}-\bm{\beta}_0)^T\mathbf{X}_n^T  \mathbf{\Sigma}_{C,n}^{-1} \frac{\partial \mathbf{\Sigma}_{C,n}}{\partial w_k} \mathbf{\Sigma}_{C,n}^{-1} \mathbf{\Gamma}_n^T\cdot  \mathbf{\Gamma}_n
	\mathbf{\Sigma}_{C,n}^{-1} \frac{\partial \mathbf{\Sigma}_{C,n}}{\partial w_k} \mathbf{\Sigma}_{C,n}^{-1} \mathbf{X}_n
	(\bm{\beta}-\bm{\beta}_0) \leq \lambda_{max}(\mathbf{\Sigma}_{0,n})\times\\ &\hspace{5cm}\lambda^{-4}_{min}(\mathbf{\Sigma}_{C,n}) 
	\lambda^2_{max}(\frac{\partial \mathbf{\Sigma}_{C,n}}{\partial w_k}) \lambda_{max}(\frac{1}{n}\mathbf{X}_n^T\mathbf{X}_n)||  \bm{\beta}_0-\bm{\beta}||_2^2 <\infty.
	\end{align*}

	For the first three terms, one can derive
	\begin{align*}
	&\frac{1}{n}tr(\mathbf{\Sigma}_{C,n}^{-1} \frac{\partial \mathbf{\Sigma}_{C,n}}{\partial w_k} \mathbf{\Sigma}_{C,n}^{-1} \mathbf{\Sigma}_{0,n}) 
	\leq \frac{1}{n}tr(\frac{\partial \mathbf{\Sigma}_{C,n}}{\partial w_k} ) \lambda_{max}( \mathbf{\Sigma}_{C,n}^{-1} \mathbf{\Sigma}_{0,n} \mathbf{\Sigma}_{C,n}^{-1}) 
	\leq \frac{1}{n}tr(\frac{\partial \mathbf{\Sigma}_{C,n}}{\partial w_k} ) \lambda^{-2}_{min}(\mathbf{\Sigma}_{C,n})\lambda_{max}(\mathbf{\Sigma}_{0,n}) 
	\\
	&\frac{1}{n}(\bm{\beta}-\bm{\beta}_0)^T \mathbf{X}_n^T \mathbf{\Sigma}_{C,n}^{-1} \frac{\partial \mathbf{\Sigma}_{C,n}}{\partial w_k} \mathbf{\Sigma}_{C,n}^{-1}  \mathbf{X}_n (\bm{\beta}-\bm{\beta}_0) 
	\leq \lambda_{max}(\frac{\partial \mathbf{\Sigma}_{C,n}}{\partial w_k})\lambda^{-2}_{min}(\mathbf{\Sigma}_{C,n})\lambda_{max}(\frac{1}{n}\mathbf{X}_n^T\mathbf{X}_n)||\bm{\beta}_0 -\bm{\beta}  ||_2^2 
	\\
	&\frac{1}{n}tr(\mathbf{\Sigma}_{C,n}^{-1}\frac{\partial \mathbf{\Sigma}_{C,n}}{\partial w_k}) 
	\leq \lambda^{-1}_{min}(\mathbf{\Sigma}_{C,n})\frac{1}{n}tr(\frac{\partial \mathbf{\Sigma}_{C,n}}{\partial w_k}),
	\end{align*}
	which are all bounded given Claim \ref{claim: max and min eigen of C and derivative} and Assumption (\ref{assumption: compact and convex})(\ref{assumption: X^TX})(\ref{assumption:max and min eigen for true})(\ref{assumption:max and min eigen for A_m,m}) 
	together with $\frac{1}{n}tr(\frac{\partial \mathbf{\Sigma}_{C,n}}{\partial w_k} )=1$. 
	Thus $\frac{1}{n}l_{n\bm{\beta}}$ and $\frac{1}{n}l_{n\bm{w}}$ are dominated by some random variable $B_n=O_p(1)$, $E[B_n]<\infty$.  Henceforth, we can apply Corollary 2.2 in (\citep{Newey1991}) to get that $\frac{1}{n}l_n$ is stochastically equicontinous and 
	\begin{align*}
	\underset{\bm{\theta}\in\bm{\Theta}_m}{sup}\frac{1}{n}l_n(\bm{\theta})-E_0\frac{1}{n}l_n(\bm{\theta})=o_p(1). 
	\end{align*}
\end{proof}

\begin{claim}\label{claim: unique global optimizer}Under Assumption (\ref{assumption:identification of C and beta})
	, $E_0\frac{1}{n}l_n(\bm{\theta})$ has a unique global maximum among $\bm{\Theta}$ at $\bm{\theta_0}$.
\end{claim}
\begin{proof}
	For any $\bm{\theta}=(\bm{\beta},C)\in \bm{\Theta}$,
	\begin{align}
	&E_0\frac{1}{n}l_n(\bm{\theta}_0)-E_0\frac{1}{n}l_n(\bm{\theta})  \nonumber\\
	=&\frac{1}{2n}logdet(\mathbf{\Sigma}_{C,n})-\frac{1}{2n} logdet(\mathbf{\Sigma}_{0,n})+\frac{1}{2n}tr(\mathbf{\Sigma}_{C,n}^{-1}\mathbf{\Sigma}_{0,n})-\frac{1}{2n}tr(\mathbf{\Sigma}_{0,n}^{-1}\mathbf{\Sigma}_{0,n})  \nonumber \\
	&+\frac{1}{2n}(\bm{\beta}_0-\bm{\beta})^T \mathbf{X}_n^{T} \mathbf{\Sigma}_{C,n}^{-1}\mathbf{X}_n(\bm{\beta}_0-\bm{\beta})
	\nonumber \\
	=&\frac{1}{2n} logdet(\mathbf{\Sigma}_{0,n}^{-\frac{1}{2}}\mathbf{\Sigma}_{C,n} \mathbf{\Sigma}_{0,n}^{-\frac{1}{2}}) + \frac{1}{2n}tr([\mathbf{\Sigma}_{0,n}^{-\frac{1}{2}}\mathbf{\Sigma}_{C,n} \mathbf{\Sigma}_{0,n}^{-\frac{1}{2}}]^{-1}) - \frac{1}{2n}\cdot n  \nonumber \\
	&+\frac{1}{2n}(\bm{\beta}_0-\bm{\beta})^T \mathbf{X}_n^{T} \mathbf{\Sigma}_{C,n}^{-1}\mathbf{X}_n(\bm{\beta}_0-\bm{\beta})
	\nonumber \\
	=&\frac{1}{2n}\sum_{i=1}^{n}(log\lambda_{i}+\frac{1}{\lambda_{i}}-1)+\frac{1}{2n}(\bm{\beta}_0-\bm{\beta})^T \mathbf{X}_n^{T} \mathbf{\Sigma}_{C,n}^{-1}\mathbf{X}_n(\bm{\beta}_0-\bm{\beta}), 
	\label{eq:Expectation minus}
	\end{align}
	where $\{\lambda_{i},i=1,\cdots,n\}$ are the eigenvalues of $\mathbf{\Sigma}_{0,n}^{-\frac{1}{2}}\mathbf{\Sigma}_{C,n} \mathbf{\Sigma}_{0,n}^{-\frac{1}{2}}$ of which the positive definiteness is ensured by Claim \ref{claim: max and min eigen of C and derivative} and
	Assumption (\ref{assumption:max and min eigen for true})(\ref{assumption:max and min eigen for A_m,m}).
	$f(x)=logx+\frac{1}{x}-1\geq 0$ for any $x>0$ and the equality holds if and only if $x=1$.
	Accordingly, $E_0\frac{1}{n}l_n(\bm{\theta}_0)-E_0\frac{1}{n}l_n(\bm{\theta})\geq 0$ and the equality holds if and only if $\mathbf{\Sigma}_{0,n}=\mathbf{\Sigma}_{C,n}$ and $\mathbf{ X}_n\bm{\beta}=\mathbf{X}_n\bm{\beta}_0$ and hence by Assumption (\ref{assumption:identification of C and beta}), 
	$\bm{\theta}_0$ is the unique global optimizer.
\end{proof}

Noticing that the Taylor expansion of $f(x)$ at $x=1$ gives
\begin{align}
f(x)=\frac{1}{2}(x-1)^2+o((x-1)^3), \label{eq:Taylor expansion of f}
\end{align}
one can always find a scalar $M_1>0$ such that for $x$ near $1$
\begin{align*}
f(x)\leq M_1(x-1)^2.
\end{align*}
Thus following from the equation (\ref{eq:Expectation minus}) for any $C\in\Omega$ and all $n\in \mathbb{N}$ one obtain
\begin{align*}
&E_0\frac{1}{n}L_n(\bm{\beta}_0,C_0) -E_0\frac{1}{n}L_n(\bm{\beta}_0,C)\\
\leq & \frac{M_1}{2n}\sum_{i=1}^n (\lambda_{i}(\mathbf{\Sigma}_{0,n}^{-\frac{1}{2}}\mathbf{\Sigma}_{C,n}\mathbf{ \Sigma}_0^{-\frac{1}{2}})-1)^2
\\
= &\frac{M_1}{2n}\sum_{i=1}^{n}tr((\mathbf{\Sigma}_{0,n}^{-\frac{1}{2}}\mathbf{\Sigma}_{C,n}\mathbf{ \Sigma}_0^{-\frac{1}{2}}-\mathbf{ I}_n)^2) 
\\
= &
\frac{M_1}{2n}tr(\mathbf{\Sigma}_{0,n}^{-\frac{1}{2}}(\mathbf{\Sigma}_{0,n}-\mathbf{\Sigma}_{C,n})\mathbf{\Sigma}_{0,n}^{-1}(\mathbf{\Sigma}_{0,n}-\mathbf{\Sigma}_{C,n})\mathbf{\Sigma}_{0,n}^{-\frac{1}{2}})
\\
\leq& \frac{M_1}{2}\lambda^{-2}_{min} (\mathbf{\Sigma}_{0,n})\cdot\frac{1}{n}
tr((\mathbf{\Sigma}_{0,n}-\mathbf{\Sigma}_{C,n})^2)\\
\leq&  \frac{M_1}{2}\lambda^{-2}_{min} (\mathbf{\Sigma}_{0,n})\cdot
\underset{1\leq i\leq n}{sup}\sum_{j=1}^n (C_0(|| \bm{s}_i-\bm{s}_j ||_2)-C(|| \bm{s}_i-\bm{s}_j ||_2)^2,
\end{align*}
where $\{\lambda_{i}(\mathbf{\Sigma}_{0,n}^{-\frac{1}{2}}\mathbf{\Sigma}_{C,n}\mathbf{ \Sigma}_0^{-\frac{1}{2}}),i=1\cdots,n\}$ are the eigenvalues of $\mathbf{\Sigma}_{0,n}^{-\frac{1}{2}}\mathbf{\Sigma}_{C,n}\mathbf{ \Sigma}_0^{-\frac{1}{2}}$.
By Approximation Theorem, one can find a sequence of functions $\{\pi_{m}(C_0):\pi_{m}(C_0)\in\Omega_m\}$ satisfying $||C_0-\pi_{m}(C_{0})||_\infty\rightarrow 0$ which by Assumption (\ref{assumption: parameter to covariance}) 
leads to,
\begin{align}\label{eq:sieve convergence}
\underset{n,m\rightarrow\infty}{\lim}E_0\frac{1}{n}L_n(\bm{\beta}_0,C_0) -E_0\frac{1}{n}L_n(\bm{\beta}_0,\pi_m({C_0}))= 0.
\end{align}

We can obtain the inequality regarding the sieve MLE $(\hat{\bm{\beta}}_{m,n},\hat{C}_{m,n})$ as
\begin{align*}
&\frac{1}{n}l_n(\bm{\beta}_0,\pi_m(C_0)) -E_0\frac{1}{n}l_n(\bm{\beta}_0,\pi_m(C_0)) + 
E_0\frac{1}{n}l_n(\bm{\beta}_0,\pi_m(C_0))-E_0\frac{1}{n}l_n(\bm{\beta}_0,C_0) 
\\
\leq & \frac{1}{n}l_n(\hat{\bm{\beta}}_{m,n},\hat{C}_{m,n}) - E_0\frac{1}{n}l_n(\bm{\beta}_0,C_0)
\\
\leq & \frac{1}{n}l_n(\hat{\bm{\beta}}_{m,n},\hat{C}_{m,n}) - E_0\frac{1}{n}l_n(\hat{\bm{\beta}}_{m,n},\hat{C}_{m,n}).
\end{align*}
The first and second inequality hold due to the fact that $(\hat{\bm{\beta}}_{m,n},\hat{C}_{m,n})$ is the maximizer of $\frac{1}{n}l_n(\bm{\theta})$ among $\bm{\Theta}_m$ and $(\bm{\beta}_0,C_0)$ is the maximizer of $E_0\frac{1}{n}l_n(\bm{\theta})$ among $\bm{\Theta}$. For any $m$,
$\frac{1}{n}l_n(\bm{\beta}_0,\pi_m(C_0)) -E_0\frac{1}{n}l_n(\bm{\beta}_0,\pi_m(C_0))\overset{P_{\bm{\theta}_0}}{\longrightarrow} 0$ and $\frac{1}{n}l_n(\hat{\bm{\beta}}_{m,n},\hat{C}_{m,n}) - E_0\frac{1}{n}l_n(\hat{\bm{\beta}}_{m,n},\hat{C}_{m,n})\overset{P_{\bm{\theta}_0}}{\longrightarrow}0$ are results of (\ref{eq:uniform convergence}) by Claim \ref{claim: uniform convergence}. By (\ref{eq:sieve convergence}) the sieve space can be cleverly chosen so that $E_0\frac{1}{n}l_n(\bm{\beta}_0,\pi_{m}(C_0))-E_0\frac{1}{n}l_n(\bm{\beta}_0,C_0) \rightarrow 0$. Hence
\begin{align*}
\frac{1}{n}l_n(\hat{\bm{\beta}}_{m,n},\hat{C}_{m,n}) - E_0\frac{1}{n}l_n(\bm{\beta}_0,C_0)\overset{P_{\bm{\theta}_0}}\longrightarrow 0.
\end{align*}
Moreover, because
\begin{align*}
&   |E_0\frac{1}{n}l_n(\hat{\bm{\beta}}_{m,n},\hat{C}_{m,n}) - 
E_0\frac{1}{n}l_n(\bm{\beta}_{0},C_0)| \\
\leq& 	|E_0\frac{1}{n}l_n(\hat{\bm{\beta}}_{m,n},\hat{C}_{m,n}) - \frac{1}{n}l_n(\hat{\bm{\beta}}_{m,n},\hat{C}_{m,n})| + |\frac{1}{n}l_n(\hat{\bm{\beta}}_{m,n},\hat{C}_{m,n})-
E_0\frac{1}{n}l_n(\bm{\beta}_{0},C_0)|,
\end{align*}
one can conclude
\begin{align}
E_0\frac{1}{n}l_n(\hat{\bm{\beta}}_{m,n},\hat{C}_{m,n}) - 
E_0\frac{1}{n}l_n(\bm{\beta}_{0},C_0)\overset{P_{\bm{\theta}_0}}\longrightarrow 0.\label{eq: converge in P}
\end{align}


The Taylor expansion in (\ref{eq:Taylor expansion of f}) indicates there exists a scalar $M_2>0$ such that
\begin{align*}
f(x)\geq M_2(x-1)^2.
\end{align*}
for $x$ near 1. Consider a sequence $\{\bm{\theta}_n=(\bm{\beta}_n,C_n),C_n\in\bm{\Omega}_{m(n)}\}_{n=1}^\infty \subset \bm{\Theta}$ satisfying $\bar{d}(\bm{\theta}_n,\bm{\theta}_0)>\epsilon$ for some $\epsilon>0$. From the equation (\ref{eq:Expectation minus}) we have
\begin{align*}
&E_0\frac{1}{n}l_n(\bm{\theta}_0) -E_0\frac{1}{n}l_n(\bm{\theta}_n)
\\
\geq &\frac{1}{2n}(\bm{\beta}_0-\bm{\beta})^T \mathbf{X}_n^{T}\mathbf{\Sigma}_{C_n,n}^{-1} \mathbf{X}_n(\bm{\beta}_0-\bm{\beta})+
\frac{M_2}{2n} tr((\mathbf{I}_n-\mathbf{\Sigma}_{0,n}^{-\frac{1}{2}}\mathbf{\Sigma}_{C_n,n}\mathbf{\Sigma}_{0,n}^{-\frac{1}{2}})^2)
\\
\geq & \frac{1}{2} ||\bm{\beta}_0-\bm{\beta} ||_2^2 \lambda_{min}(\frac{1}{n}\mathbf{X}_n\mathbf{X}_n)\lambda^{-1}_{max}(\mathbf{\Sigma}_{C_n,n})
+\frac{M_2}{2}\lambda_{min}(\mathbf{\Sigma}_{0,n}^{-2})\cdot \frac{1}{n}tr((\mathbf{\Sigma}_{0,n}-\mathbf{\Sigma}_{C_n,n})^2)
\\
\geq & \frac{1}{2} ||\bm{\beta}_0-\bm{\beta} ||_2^2 \lambda_{min}(\frac{1}{n}\mathbf{X}_n\mathbf{X}_n)|| C_n||_\infty^{-1}\lambda^{-1}_{max}(\mathbf{A}_{m(n),m(n),n})
+\frac{M_2}{2}\lambda_{min}(\mathbf{\Sigma}_{0,n}^{-2})\cdot \frac{1}{n}tr((\mathbf{\Sigma}_{0,n}-\mathbf{\Sigma}_{C_n,n})^2),
\end{align*}
where the last inequality holds due to Claim \ref{claim: max and min eigen of C and derivative}.
Given $\bar{d}(\bm{\theta},\bm{\theta}_0)>\epsilon$,  it is true that $||\bm{\beta}_0-\bm{\beta} ||_2\geq \frac{\epsilon}{2}$ or $|| C_0-C_{n} ||_\infty \geq \frac{\epsilon}{2}$ and furthermore,
\begin{align*}
E_0\frac{1}{n}l_n(\bm{\theta}_0) -E_0\frac{1}{n}l_n(\bm{\theta}_n) \geq min( \frac{\epsilon^2}{4} M_4, \frac{M_2}{2} \cdot M_3)>0,
\end{align*}
where by Assumption (\ref{assumption: compact and convex})(\ref{assumption: X^TX})(\ref{assumption: limsup max eigen final})(\ref{assumption: covariance to parameter}), 
the scalars $M_3$ and $M_4$ take the form
\begin{align*}
M_3&=\lambda_{min}(\mathbf{\Sigma}_{0,n}^{-2})\cdot 	\underset{n\rightarrow\infty}{\lim\inf} \underset{
	\begin{subarray}{c}
	C\in\bm{\Omega}\\
	|| C-C_0 ||_\infty\geq \epsilon/2
	\end{subarray}}
{\inf}\frac{1}{n}\sum_{i,j=1}^n(C(|| \bm{s}_i-\bm{s}_j ||_2)-C_0(||\bm{s}_i-\bm{s}_j ||_2))^2>0\\
M_4&=\frac{1}{2}\underset{n}{\inf}\;\lambda_{min}(\frac{1}{n}\mathbf{X}_n\mathbf{X}_n)\cdot U_C^{-1}\cdot[\underset{n\rightarrow\infty}{\lim\sup}\;\lambda_{max}(\mathbf{A}_{m(n),m(n),n})]^{-1}>0. 
\end{align*}
Accordingly, $\sup_{\{\bar{d}(\bm{\theta},\bm{\theta}_0)\geq \epsilon,\bm{\theta}\in \Omega_m,m\in\mathbb{N}^+\}}E_0\frac{1}{n}l_n(\bm{\theta}_0)-E_0\frac{1}{n}l_n(\bm{\theta})>0$. Finally from (\ref{eq: converge in P}), we see that $\bar{d}((\hat{\beta}_{m,n},\hat{C}_{m,n}),(\bm{\beta}_0,C_0))\leq \epsilon$ with probability approaching 1.

\newpage
\section{Tables}
\subsection{Tables in the main text}

\begin{table}[H]
	\centering
	\caption{Simulation results of nonparametric methods.}
	\label{table:simulation of average for three NP}
	
	\resizebox{\textwidth}{!}{%
		\begin{tabular}{ccccccc}
			\toprule
			& & \multicolumn{5}{c}{Truth}\\
			\cline{3-7}
			& & Setting 1 & Setting 2 & Setting 3 & Setting 4 & Setting 5 \\
			& & Matern & Cauchy & Gaussian & GenCauchy & LinearMatern \\
			\midrule
			& Bias($C_0(0)$)  & 0.33 (1.98)&-0.39 (2.18)&1.56 (3.98)&-1.44 (3.04)&0.29 (2.17) \\
			Our NP & $||\hat{R}_0-R_0||_2$  & 1.3 (0.96)&1.74 (0.8)&0.97 (1.58)&3.22 (3.15)&1.23 (1.13) \\
			method & $||\hat{R}_0-R_0||_{\infty}$ & 2.32 (1.23)&2.56 (0.89)&1.5 (2.29)&6.16 (3.37)&2.02 (1.54) \\
			& $||\hat{C}_0-C_0||_2$ & 1.57 (1.01)&1.96 (0.86)&1.83 (3.16)&3.66 (3.25)&1.6 (1.17) \\
			& $||\hat{C}_0-C_0||_{\infty}$ & 3.13 (1.47)&3.16 (1.22)&2.66 (3.87)&6.66 (3.68)&3 (1.57) \\
			& $||\hat{\gamma}_0-\gamma_0||_{2}$ & 1.64 (0.92)&1.68 (1.01)&0.67 (0.57)&2.3 (2.41)&1.71 (1.03) \\
			& $||\hat{\gamma}_0-\gamma_0||_{\infty}$  & 2.57 (1.16)&2.54 (1.22)&1.12 (1.18)&5.37 (2.84)&2.52 (1.35) \\
			\rule{0pt}{5ex}
			& Bias($C_0(0)$)  & 2.91 (2.95)&-2.45 (2.51)&0.47 (3.48)&-9.94 (2.96)&2.82 (2.93) \\
			Huang & $||\hat{R}_0-R_0||_2$   & 4.28 (1.32)&2.98 (0.78)&2.39 (1.47)&8.62 (2.24)&3.98 (1.3) \\
			& $||\hat{R}_0-R_0||_{\infty}$  & 6.74 (1.63)&7.67 (2.26)&3.75 (2.07)&12.09 (2.05)&6.25 (1.55) \\
			& $||\hat{C}_0-C_0||_2$  & 4.56 (1.58)&3.44 (1.07)&3.11 (1.78)&10.29 (2.46)&4.29 (1.61) \\
			& $||\hat{C}_0-C_0||_{\infty}$ & 7.05 (1.96)&8.84 (2.96)&4.55 (2.17)&15.65 (2.94)&6.62 (1.99) \\
			& $||\hat{\gamma}_0-\gamma_0||_{2}$ & 3.28 (0.82)&2.29 (0.73)&1.95 (1.05)&2.13 (0.61)&3.12 (0.83) \\
			& $||\hat{\gamma}_0-\gamma_0||_{\infty}$ & 7.27 (1.96)&6.28 (1.5)&3.19 (1.6)&6.49 (1.34)&6.42 (1.44) \\
			\rule{0pt}{5ex}
			& Bias($C_0(0)$)   & 0 (2) & -0.02 (2.22) & 4.45 (2.84) & -0.28 (1.98) & 0.07 (2.21) \\
			modified- & $||\hat{R}_0-R_0||_2$   & 1.45 (0.36) & 2.08 (0.44) & 6.49 (0.61) & 8.61 (1.3) & 1.95 (0.33) \\
			Choi & $||\hat{R}_0-R_0||_{\infty}$  & 2.95 (1) & 3.48 (1.32) & 12.4 (1.82) & 25.68 (0.63) & 3.84 (1.1) \\
			& $||\hat{C}_0-C_0||_2$  & 1.69 (0.5) & 2.26 (0.62) & 5.78 (0.65) & 8.56 (1.84) & 2.2 (0.48) \\
			& $||\hat{C}_0-C_0||_{\infty}$ & 3.59 (1.35) & 3.99 (1.68) & 10.58 (2.15) & 25.44 (2.16) & 4.67 (1.59) \\
			& $||\hat{\gamma}_0-\gamma_0||_{2}$  & 1.88 (0.73) & 2.26 (0.97) & 7.37 (1.26) & 8.68 (1) & 2.37 (0.76) \\
			& $||\hat{\gamma}_0-\gamma_0||_{\infty}$ & 3.23 (0.97) & 3.79 (1.37) & 14.65 (2.2) & 25.72 (0.49) & 4.12 (1.11) \\
			\rule{0pt}{5ex}
			& Bias($C_0(0)$) & 0.07 (2.02) & 0.02 (2.21)&1.04 (2.58) & -0.08 (2.73)&0.34 (2.3) \\
			original Choi & $||\hat{R}_0-R_0||_2$  & 4.23 (5.11)&3.68 (2.84)&21.42 (8.19)&7.49 (1.13)&8.63 (8.43) \\
			method & $||\hat{R}_0-R_0||_{\infty}$ & 9.09 (7.87)&8.48 (7.25)&29.54 (10.79)&12 (3.41)&14.1 (11.92) \\
			& $||\hat{C}_0-C_0||_2$  & 4.35 (5.16)&3.83 (2.98)&21.95 (8.53)&7.64 (1.35)&8.87 (8.73) \\
			& $||\hat{C}_0-C_0||_{\infty}$ & 9.32 (8.03)&8.7 (7.62)&30.08 (11.19)&12.46 (3.95)&14.46 (12.44) \\
			& $||\hat{\gamma}_0-\gamma_0||_{2}$  & 4.42 (5.06)&3.85 (2.82)&20.95 (7.87)&7.43 (1.67)&8.62 (8.03) \\
			& $||\hat{\gamma}_0-\gamma_0||_{\infty}$ & 9.13 (7.78)&8.55 (7.06)&29.04 (10.41)&11.75 (3.37)&14.03 (11.62) \\
			\bottomrule
		\end{tabular}%
	}
	
	\vspace{0.5em}
	\begin{minipage}{0.95\textwidth}
		\footnotesize
		Note: The entries in the table show the means and standard deviations, in parentheses, of scaled 2-norm and inf-norm errors of correlation, covariance and semivariance estimates as well as bias of $\hat{C}_0(0)$ based on 100 MC simulations. Columns correspond to true models and rows correspond to estimation methods. The entries are in units of $10^{-2}$.
	\end{minipage}
\end{table}

\begin{table}[H]
	\centering
	\caption{Relative loss of efficiency of nonparametric methods vs best parametric fits.}
	\label{table:simulation of relative loss efficiency of three nonparametric methods vs best parametric fit}
	
	\resizebox{\textwidth}{!}{%
		\begin{tabular}{lcccccc}
			\toprule
			& & \multicolumn{5}{c}{Truth}\\
			\cline{3-7}
			& & Setting 1 & Setting 2 & Setting 3 & Setting 4 & Setting 5 \\
			& & Matern & Cauchy & Gaussian & GenCauchy & LinearMatern \\
			\midrule
			Our NP & Bias($C_0(0)$) & 13.98 & 19.49 & 22.24 & 37.33 & 2.95 \\
			method & $||\hat{R}_0-R_0||_{2}$ & 2.53 & 4.82 & 4.49 & 5.55 & -0.14 \\
			& $||\hat{R}_0-R_0||_{\infty}$ & 2.3 & 2.65 & 3.87 & 5.24 & -0.35 \\
			& $||\hat{C}_0-C_0||_{2}$ & 1.03 & 2.23 & 1.13 & 3.23 & 0.02 \\
			& $||\hat{C}_0-C_0||_{\infty}$ & 0.81 & 0.86 & 0.68 & 2.24 & -0.19 \\
			& $||\hat{\gamma}_0-\gamma_0||_{2}$ & 0.32 & 0.49 & 0.19 & 1.35 & -0.11 \\
			& $||\hat{\gamma}_0-\gamma_0||_{\infty}$ & 0.59 & 0.8 & 0.4 & 2.98 & -0.24 \\
			\rule{0pt}{4.5ex}
			Huang & Bias($C_0(0)$) & 131.06 & 127.71 & 5.99 & 264 & 37.45 \\
			& $||\hat{R}_0-R_0||_{2}$ & 10.62 & 8.97 & 12.53 & 16.52 & 1.79 \\
			& $||\hat{R}_0-R_0||_{\infty}$ & 8.57 & 9.94 & 11.18 & 11.25 & 1.02 \\
			& $||\hat{C}_0-C_0||_{2}$ & 4.9 & 4.68 & 2.62 & 10.87 & 1.73 \\
			& $||\hat{C}_0-C_0||_{\infty}$ & 3.08 & 4.21 & 1.87 & 6.62 & 0.79 \\
			& $||\hat{\gamma}_0-\gamma_0||_{2}$ & 1.64 & 1.03 & 2.5 & 1.17 & 0.63 \\
			& $||\hat{\gamma}_0-\gamma_0||_{\infty}$ & 3.48 & 3.45 & 3.01 & 3.81 & 0.92 \\
			\rule{0pt}{4.5ex}
			modified- & Bias($C_0(0)$) & -0.85 & -0.19 & 65.13 & 6.36 & -0.11 \\
			Choi & $||\hat{R}_0-R_0||_{2}$ & 2.95 & 5.96 & 35.67 & 16.5 & 0.37 \\
			& $||\hat{R}_0-R_0||_{\infty}$ & 3.19 & 3.98 & 39.27 & 25.01 & 0.24 \\
			& $||\hat{C}_0-C_0||_{2}$ & 1.18 & 2.73 & 5.73 & 8.88 & 0.4 \\
			& $||\hat{C}_0-C_0||_{\infty}$ & 1.08 & 1.36 & 5.67 & 11.39 & 0.26 \\
			& $||\hat{\gamma}_0-\gamma_0||_{2}$ & 0.51 & 1 & 12.21 & 7.86 & 0.23 \\
			& $||\hat{\gamma}_0-\gamma_0||_{\infty}$ & 0.99 & 1.69 & 17.42 & 18.06 & 0.24 \\
			\bottomrule
		\end{tabular}%
	}
	
	\vspace{0.5em}
	\begin{minipage}{0.95\textwidth}
		\footnotesize
		Note: The entries in the table are relative loss of efficiency of scaled 2-norm and inf-norm errors of correlation, covariance, and semivariance estimates, as well as absolute bias. For example, regarding the inf-norm error of correlation estimation, when the true covariance is Setting 1 (Matern), the best among the 12 parametric fits is the likelihood-based estimation of the Matern model with the minimum mean value of inf-norm errors. Thus, the relative loss of efficiency of our proposed nonparametric method is computed through $\frac{||\hat{R}_0-R_0||_{\infty,\text{our proposed method}}}{||\hat{R}_0-R_0||_{\infty,\text{best parametric fit}}}-1$. The entries are in units of $1$.
	\end{minipage}
\end{table}

\begin{table}[H]
	\centering
	\caption{Case study results of nonparametric methods.}
	\label{table:real data results of nonparametric models}
	
	\small
	\begin{tabular}{lccc}
		\toprule
		& Our method & Huang & modified Choi \\
		\midrule
		time(mins) & 19.26 & 0.21 & 0.45 \\
		log-likelihood/(rn) & -0.703 & -0.725 & -1.690 \\
		nugget($\mathrm{inch}^2$) & 7.38 & 4.09 & 14.93 \\
		$C_0(0)$ ($\mathrm{inch}^2$) & 48.3 & 56.74 & 37.42 \\
		\bottomrule
	\end{tabular}
\end{table}

\begin{table}[H]
	\centering
	\caption{Results for candidate $m$ of our proposed method.}
	\label{table:all m our method Year 40 Station 189}
	
	\small
	\begin{tabular}{lccccccc}
		\toprule
		$m$ & 2 & 3 & 4 & 6 & 10 & 15 & 23 \\
		\midrule
		time(mins) & 0.55 & 0.92 & 1.14 & 1.45 & 2.89 & 3.86 & 8.45 \\
		log-likelihood/(rn) & -0.717 & -0.712 & -0.708 & -0.709 & -0.705 & -0.703 & -0.703 \\
		nugget($\mathrm{inch}^2$) & 9.07 & 11.08 & 8.68 & 10.36 & 9.17 & 7.62 & 7.38 \\
		$C_0(0)$ ($\mathrm{inch}^2$) & 36.48 & 42.48 & 41.27 & 50.53 & 44.88 & 46.38 & 48.3 \\
		\bottomrule
	\end{tabular}
\end{table}

\subsection{Additional Results of Section 4}
\label{sec:additional results of section 4}
\begin{table}[H]
	\centering
	\scriptsize
	\caption{Simulation results of variogram-based parametric methods.}
	\begin{tabular}{clccccc}
		\hline\hline
		&                       & \multicolumn{5}{c}{Truth}\\
		\cline{3-7}
		&                       & Setting 1   & Setting 2    &  Setting 3  &    Setting 4    & Setting 5  \\    
		&                       & Matern   &  Cauchy      & Gaussian     & GenCauchy     & LinearMatern \\
		\hline         
		Model 1 &  Bias($C_0(0)$)    & 0.17 (2.22)&-7.61 (2.04)&0.31 (2.64)&-17.23 (1.7)&0.31 (2.41) \\ 
		Matern  &     $||\hat{R}_0-R_0||_2$   & 0.87 (0.68)&7.71 (0.29)&1.27 (0.59)&17.28 (0.4)&1.43 (0.31) \\ 
		&     $||\hat{R}_0-R_0||_{\infty}$   & 1.69 (1.37)&11.13 (0.71)&2.31 (1.09)&21.87 (0.77)&3.1 (0.82) \\ 
		&      $||\hat{C}_0-C_0||_2$  & 1.25 (0.87)&7.94 (0.5)&1.84 (0.86)&17.78 (0.59)&1.72 (0.68) \\ 
		&     $||\hat{C}_0-C_0||_{\infty}$  & 2.53 (1.63)&11.42 (0.86)&3.27 (1.49)&22.48 (0.85)&3.73 (1.55) \\ 
		&       $||\hat{\gamma}_0-\gamma_0||_{2}$    & 1.5 (0.92)&3.55 (0.6)&1.86 (0.88)&5.51 (0.59)&1.92 (0.82) \\ 
		&    $||\hat{\gamma}_0-\gamma_0||_{\infty}$  & 2.26 (1.44)&8.93 (1.56)&3.04 (1.22)&16.89 (1.38)&3.34 (0.95) \\ 
		\rule{0pt}{4.5ex}    
		Model 2 &  Bias($C_0(0)$)  & 8.57 (2.67)&0.08 (2.48)&13.97 (3.24)&-11.24 (2.03)&10.23 (2.88) \\ 
		Cauchy  &     $||\hat{R}_0-R_0||_2$   & 8.55 (0.51)&0.85 (0.66)&13.97 (0.43)&10.09 (0.65)&9.75 (0.5) \\ 
		&     $||\hat{R}_0-R_0||_{\infty}$    & 11.28 (0.77)&1.99 (1.54)&19.3 (0.97)&11.58 (0.72)&12.95 (0.78) \\ 
		&     $||\hat{C}_0-C_0||_2$   & 9.74 (1.19)&1.28 (0.96)&16.05 (1.45)&11.34 (0.97)&11.33 (1.28) \\ 
		&    $||\hat{C}_0-C_0||_{\infty}$  & 12.61 (1.29)&3.22 (2.24)&22.48 (1.68)&13 (1.22)&14.75 (1.34) \\ 
		&      $||\hat{\gamma}_0-\gamma_0||_{2}$    & 4.16 (0.8)&1.47 (0.97)&9.45 (0.69)&2.86 (0.62)&4.73 (0.74) \\ 
		&     $||\hat{\gamma}_0-\gamma_0||_{\infty}$  & 9.32 (1.6)&2.2 (1.26)&17.9 (1.71)&10.32 (1.59)&10.36 (1.4) \\ 
		\rule{0pt}{4.5ex}    
		Model 3 &  Bias($C_0(0)$)   & -0.87 (2.11)&-8.93 (1.91)&0.09 (2.61)&-19.4 (1.53)&-0.75 (2.3) \\ 
		Gaussian  &   $||\hat{R}_0-R_0||_2$   & 5.1 (0.21)&10.53 (0.14)&0.77 (0.63)&30.32 (3.84)&4.36 (0.21) \\ 
		&     $||\hat{R}_0-R_0||_{\infty}$    & 12.3 (1.29)&18.22 (0.74)&1.33 (1.1)&91.39 (23.39)&9.99 (1.13) \\ 
		&     $||\hat{C}_0-C_0||_2$   & 5 (0.62)&10.3 (0.26)&1.54 (0.95)&30.31 (3.87)&4.31 (0.67) \\ 
		&     $||\hat{C}_0-C_0||_{\infty}$  & 11.77 (2.48)&18.35 (0.94)&2.66 (1.55)&91.44 (23.26)&9.67 (2.33) \\ 
		&       $||\hat{\gamma}_0-\gamma_0||_{2}$   & 5.11 (0.61)&6.08 (0.44)&1.62 (0.97)&16.77 (3.44)&4.47 (0.66) \\ 
		&    $||\hat{\gamma}_0-\gamma_0||_{\infty}$  & 12.49 (1.03)&17.16 (1.42)&2.45 (1.39)&73.65 (18.59)&10.16 (0.86) \\ 
		\rule{0pt}{4.5ex}    
		Model 4  &  Bias($C_0(0)$)   & 23.67 (3.56)&14.32 (3.37)&35.04 (4.23)&0.16 (2.73)&27.47 (3.77) \\ 
		GenCauchy  &     $||\hat{R}_0-R_0||_2$   & 18.56 (0.86)&10.77 (0.99)&23.64 (0.76)&1.1 (0.88)&20.33 (0.8) \\ 
		&     $||\hat{R}_0-R_0||_{\infty}$   & 23.07 (1.12)&12.35 (1.1)&32.83 (1.14)&2.21 (1.75)&25.63 (1.01) \\ 
		&    $||\hat{C}_0-C_0||_2$   & 24.53 (2.33)&14.17 (2.11)&35.55 (2.99)&1.69 (1.34)&28.2 (2.42) \\ 
		&    $||\hat{C}_0-C_0||_{\infty}$   & 29.49 (2.34)&16.24 (2.55)&45.38 (2.82)&3.67 (2.75)&33.8 (2.36) \\ 
		&      $||\hat{\gamma}_0-\gamma_0||_{2}$   & 7.2 (0.81)&2.99 (0.65)&12.32 (0.71)&1.26 (0.74)&7.53 (0.73) \\ 
		&     $||\hat{\gamma}_0-\gamma_0||_{\infty}$  & 18.41 (1.91)&10.08 (2.18)&23.92 (1.83)&2.15 (1.27)&18.2 (1.61) \\ 
		\hline\hline
	\end{tabular}  
	\\
	\vspace{1ex}
	{\raggedright \scriptsize{Note:The entries in the table show the means and standard deviations, in parenthesis, of scaled 2-norm and inf-norm errors of correlation, covariance and semivariance estimates as well as bias of $\hat{C}_0(0)$ based on 100 mc simulations. Columns correspond to true models and rows correspond to estimation methods. The entries are in units of $10^{-2}$.}\par}
	\label{table:simulation of average for variogram-based parametric methods}
\end{table}

\begin{table}[H]
	\centering
	\scriptsize
	\caption{Simulation results of likelihood-based parametric methods.}
	\begin{tabular}{clccccc}
		\hline\hline
		&                       & \multicolumn{5}{c}{Truth}\\
		\cline{3-7}
		&                       & Setting 1   & Setting 2    &  Setting 3  &    Setting 4    & Setting 5  \\    
		&                       & Matern   &  Cauchy      & Gaussian     & GenCaucy     & LinearMatern \\
		\hline    
		Model 1  & Bias($C_0(0)$)  & 0.02 (1.93)&-0.12 (2.09)&5.37 (2.19)&0.82 (2.66)&0.34 (2.15) \\ 
		Matern                  & $||\hat{R}_0-R_0||_2$ & 0.37 (0.27)&8.44 (0.17)&3.25 (0.26)&17.98 (0.19)&2.03 (0.44) \\ 
		&        $||\hat{R}_0-R_0||_{\infty}$                & 0.7 (0.52)&12.76 (0.35)&5.28 (0.41)&23.21 (0.38)&3.57 (0.84) \\ 
		&      $||\hat{C}_0-C_0||_2$  & 0.77 (0.54)&8.47 (0.22)&5.19 (1.1)&18 (0.2)&2.27 (0.64) \\ 
		&      $||\hat{C}_0-C_0||_{\infty}$ & 1.73 (1.11)&12.77 (0.42)&7.13 (1.27)&23.19 (0.44)&3.96 (1.05) \\ 
		&    $||\hat{\gamma}_0-\gamma_0||_{2}$   & 1.24 (0.85)&8.36 (1.69)&1.26 (0.38)&18.71 (2.22)&2.15 (0.84) \\ 
		&    $||\hat{\gamma}_0-\gamma_0||_{\infty}$  & 1.62 (1.07)&12.65 (1.77)&2.56 (0.84)&24.01 (2.31)&3.52 (1.02) \\ 
		\rule{0pt}{4.5ex} 
		Model 2 & Bias($C_0(0)$)  & 11.69 (2.23)&0.02 (1.83)&121.24 (6.24)&-8.98 (1.65)&16.15 (2.52) \\ 
		Cauchy                 & $||\hat{R}_0-R_0||_2$  & 11.63 (0.38)&0.3 (0.24)&36.35 (0.34)&11.24 (0.28)&14.13 (0.37) \\ 
		&       $||\hat{R}_0-R_0||_{\infty}$ & 15.14 (0.43)&0.7 (0.56)&47.77 (0.42)&12.91 (0.37)&18.41 (0.43) \\ 
		&      $||\hat{C}_0-C_0||_2$   & 14.55 (1.02)&0.61 (0.48)&113.78 (4.82)&12.29 (0.46)&18.88 (1.24) \\ 
		&      $||\hat{C}_0-C_0||_{\infty}$ & 17.73 (0.96)&1.7 (1.22)&124.77 (5.35)&14.23 (0.6)&22.69 (1.17) \\ 
		&    $||\hat{\gamma}_0-\gamma_0||_{2}$  & 3.66 (1)&1.13 (0.83)&7.68 (0.8)&3.96 (1.09)&3.86 (0.92) \\ 
		&     $||\hat{\gamma}_0-\gamma_0||_{\infty}$ & 6.1 (1.29)&1.41 (0.96)&19.04 (1.63)&6.79 (0.62)&6.69 (1.24) \\ 
		\rule{0pt}{4.5ex} 
		Model 3 & Bias($C_0(0)$)  & -0.78 (1.83)&-0.23 (2.05)&0.07 (1.98)&1.05 (2.73)&-2.67 (1.95) \\ 
		Gaussian                & $||\hat{R}_0-R_0||_2$  & 8.73 (0.22)&12.66 (0.11)&0.18 (0.13)&21.4 (0.1)&9.3 (0.22) \\ 
		&        $||\hat{R}_0-R_0||_{\infty}$  & 19.96 (0.55)&26.13 (0.38)&0.31 (0.23)&33.26 (0.33)&20.51 (0.52) \\ 
		&      $||\hat{C}_0-C_0||_2$  & 8.79 (0.29)&12.68 (0.14)&0.86 (0.65)&21.42 (0.09)&9.57 (0.39) \\ 
		&       $||\hat{C}_0-C_0||_{\infty}$ & 20.04 (0.67)&26.14 (0.43)&1.59 (1.19)&33.24 (0.38)&20.86 (0.71) \\ 
		&    $||\hat{\gamma}_0-\gamma_0||_{2}$  & 8.4 (0.86)&12.52 (1.61)&0.56 (0.42)&22.36 (2.41)&8.12 (0.66) \\ 
		&     $||\hat{\gamma}_0-\gamma_0||_{\infty}$  & 19.26 (1.41)&25.91 (1.79)&0.8 (0.59)&34.29 (2.46)&18.18 (1.42) \\ 
		\rule{0pt}{4.5ex} 
		Model 4 & Bias($C_0(0)$)  & 45.56 (3.3)&20.88 (2.57)&264.21 (11.11)&0.04 (1.93)&57.48 (3.82) \\ 
		GenCauchy                 & $||\hat{R}_0-R_0||_2$  & 28.04 (0.47)&14.38 (0.47)&50.04 (0.26)&0.49 (0.4)&31.29 (0.43) \\ 
		&       $||\hat{R}_0-R_0||_{\infty}$  & 33.86 (0.52)&16.6 (0.57)&64.63 (0.32)&0.99 (0.8)&38.03 (0.48) \\ 
		&     $||\hat{C}_0-C_0||_2$   & 47.48 (2.14)&20.82 (1.36)&252.29 (9.52)&0.87 (0.72)&58.7 (2.56) \\ 
		&     $||\hat{C}_0-C_0||_{\infty}$ & 52.96 (2.21)&24.57 (1.83)&266.51 (10.3)&2.05 (1.51)&64.88 (2.65) \\ 
		&    $||\hat{\gamma}_0-\gamma_0||_{2}$ & 4.56 (0.5)&2.75 (0.53)&10.52 (0.94)&0.98 (0.67)&5.02 (0.36) \\ 
		&     $||\hat{\gamma}_0-\gamma_0||_{\infty}$  & 8.09 (1.02)&4.93 (1.01)&25.58 (1.83)&1.35 (0.79)&9.59 (1.39) \\ 
		\hline\hline
	\end{tabular}  
	\\
	\vspace{1ex}
	{\raggedright \scriptsize{Note: The entries in the table show the means and standard deviations, in parenthesis, of scaled 2-norm and inf-norm errors of correlation, covariance and semivariance estimates as well as bias of $\hat{C}_0(0)$ based on 100 mc simulations. Columns correspond to true models and rows correspond to estimation methods. The entries are in units of $10^{-2}$.} \par}
	\label{table:simulation of average for likelihood-based parametric methods}
\end{table}

\begin{table}[H]
	\centering
	\scriptsize
	\caption{Simulation results of cov-based parametric methods using original weights of Choi method.}
	\begin{tabular}{clccccc}
		\hline\hline
		&                       & \multicolumn{5}{c}{Truth}\\
		\cline{3-7}
		&                       & Setting 1   & Setting 2    &  Setting 3  &    Setting 4    & Setting 5  \\    
		&                       & Matern   &  Cauchy      & Gaussian     & GenCauchy    & LinearMatern \\
		\hline              
		Model 1 &Bias($C_0(0)$)  & 0.06 (2.01)&0 (2.19)&0.77 (2.36)&-0.17 (2.62)&0.29 (2.25) \\ 
		Matern              & $||\hat{R}_0-R_0||_2$  & 4.66 (8.33)&7.9 (1.77)&19.89 (12.44)&15.85 (0.75)&7.5 (7.21) \\ 
		&       $||\hat{R}_0-R_0||_{\infty}$  & 9.93 (11.22)&16.07 (5.69)&29.22 (16.38)&32.49 (2.53)&15.49 (11.77) \\ 
		&      $||\hat{C}_0-C_0||_2$ & 4.77 (8.35)&7.97 (1.89)&20.3 (12.57)&15.92 (1.37)&7.71 (7.46) \\ 
		&     $||\hat{C}_0-C_0||_{\infty}$   & 10.15 (11.27)&16.17 (6.3)&29.69 (16.57)&32.36 (4.37)&15.82 (12.25) \\ 
		&     $||\hat{\gamma}_0-\gamma_0||_{2}$  & 4.92 (8.28)&7.88 (2.27)&19.61 (12.33)&15.68 (1.33)&7.68 (6.89) \\ 
		&       $||\hat{\gamma}_0-\gamma_0||_{\infty}$    & 9.95 (11.17)&16.07 (5.52)&28.91 (16.23)&32.53 (2.22)&15.43 (11.54) \\ 
		\rule{0pt}{4.5ex} 
		Model 2 & Bias($C_0(0)$)   & 0.08 (2.03)&0.03 (2.22)&0.95 (2.54)&-0.11 (2.69)&0.39 (2.37) \\ 
		Cauchy              & $||\hat{R}_0-R_0||_2$  & 11.86 (8.45)&3 (4.11)&31.49 (10.1)&8.65 (0.83)&17.49 (8.56) \\ 
		&      $||\hat{R}_0-R_0||_{\infty}$ & 15.34 (9.98)&6.31 (7.4)&41.95 (12.43)&23.64 (2.64)&22.62 (10.68) \\ 
		&      $||\hat{C}_0-C_0||_2$  & 11.95 (8.47)&3.18 (4.28)&32.01 (10.3)&8.78 (1.37)&17.75 (8.91) \\ 
		&    $||\hat{C}_0-C_0||_{\infty}$  & 15.41 (10.01)&6.94 (7.65)&42.46 (12.67)&23.57 (4.55)&22.89 (11.11) \\ 
		&    $||\hat{\gamma}_0-\gamma_0||_{2}$    & 11.81 (8.47)&3.36 (3.8)&31.12 (10.02)&8.55 (1.17)&17.31 (8.2) \\ 
		&    $||\hat{\gamma}_0-\gamma_0||_{\infty}$   & 15.31 (9.98)&6.48 (7.1)&41.52 (12.32)&23.67 (2.32)&22.49 (10.32) \\ 
		\rule{0pt}{4.5ex} 
		Model 3 & Bias($C_0(0)$)  & 0.05 (2.01)&-0.02 (2.18)&0.76 (2.35)&-0.19 (2.6)&0.28 (2.24) \\ 
		Gaussian              & $||\hat{R}_0-R_0||_2$  & 6.6 (7.77)&10.63 (1.15)&19.13 (12.53)&19.69 (0.68)&8.15 (6.47) \\ 
		&       $||\hat{R}_0-R_0||_{\infty}$ & 13.85 (9.73)&20.56 (4.52)&28.61 (16.53)&37.46 (2.93)&16.83 (10.76) \\ 
		&       $||\hat{C}_0-C_0||_2$  & 6.68 (7.77)&10.68 (1.23)&19.53 (12.65)&19.75 (1.23)&8.35 (6.7) \\ 
		&      $||\hat{C}_0-C_0||_{\infty}$   & 14.11 (9.81)&20.87 (5.02)&29.08 (16.72)&37.31 (4.86)&17.18 (11.23) \\ 
		&   $||\hat{\gamma}_0-\gamma_0||_{2}$   & 6.78 (7.77)&10.63 (1.9)&18.87 (12.41)&19.48 (1.85)&8.2 (6.28) \\ 
		&       $||\hat{\gamma}_0-\gamma_0||_{\infty}$   & 13.88 (9.7)&20.51 (4.46)&28.31 (16.4)&37.5 (2.7)&16.74 (10.59) \\ 
		\rule{0pt}{4.5ex} 
		Model 4 & Bias($C_0(0)$)  & 0.07 (2.03)&0.03 (2.23)&0.95 (2.6)&-0.08 (2.73)&0.4 (2.4) \\ 
		GenCauchy               & $||\hat{R}_0-R_0||_2$  & 19.98 (8.38)&9.91 (4.84)&39.25 (8.8)&2.12 (1.89)&27.12 (8.97) \\ 
		&       $||\hat{R}_0-R_0||_{\infty}$ & 25.28 (9.34)&16.57 (5.67)&51.86 (10.55)&4.08 (3.42)&33.35 (10.23) \\ 
		&     $||\hat{C}_0-C_0||_2$  & 20.07 (8.42)&10.05 (5.02)&39.82 (9.07)&2.56 (2.26)&27.42 (9.42) \\ 
		&      $||\hat{C}_0-C_0||_{\infty}$   & 25.58 (9.35)&17 (6.13)&52.44 (10.86)&5.28 (4.1)&33.67 (10.7) \\ 
		&     $||\hat{\gamma}_0-\gamma_0||_{2}$    & 19.92 (8.35)&9.86 (4.66)&38.94 (8.71)&2.15 (1.5)&26.96 (8.6) \\ 
		&       $||\hat{\gamma}_0-\gamma_0||_{\infty}$   & 24.97 (9.38)&16.22 (5.56)&51.49 (10.43)&4.12 (2.92)&33.19 (9.84) \\ 
		\hline\hline
	\end{tabular}  
	\\
	\vspace{1ex}
	{\raggedright \scriptsize{Note: The entries in the table show the means and standard deviations, in parenthesis, of scaled 2-norm and inf-norm errors of correlation, covariance and semivariance estimates as well as bias of $\hat{C}_0(0)$ based on 100 mc simulations. Columns correspond to true models and rows correspond to estimation methods. The entries are in units of $10^{-2}$.} \par}
	\label{table:simulation of average for cov-based parametric methods using original weights of Choi method}
\end{table}

\begin{table}[H]
	\centering
	\scriptsize
	\caption{Simulation results of cov-based parametric methods using modified weights.}
	\begin{tabular}{clccccc}
		\hline\hline
		&                       & \multicolumn{5}{c}{Truth}\\
		\cline{3-7}
		&                       & Setting 1   & Setting 2    &  Setting 3  &    Setting 4    & Setting 5  \\    
		&                       & Matern   &  Cauchy      & Gaussian     & GenCauchy    & LinearMatern \\
		\hline             
		Model 1  & Bias($C_0(0)$)    & -0.07 (2.12) & -11.11 (2.47) & 0.91 (2.79) & -15.87 (5.77) & -0.15 (2.41) \\ 
		Matern               & $||\hat{R}_0-R_0||_2$  & 0.76 (0.67) & 9.04 (1.29) & 1.31 (0.54) & 26.8 (6.16) & 1.44 (0.34) \\ 
		&     $||\hat{R}_0-R_0||_{\infty}$   & 1.52 (1.44) & 21.69 (3.59) & 2.47 (1.07) & 46.42 (4.63) & 3.36 (0.94) \\ 
		&     $||\hat{C}_0-C_0||_2$  & 1.1 (0.71) & 7.08 (0.58) & 1.8 (0.85) & 18.49 (2.01) & 1.57 (0.61) \\ 
		&      $||\hat{C}_0-C_0||_{\infty}$    & 2.39 (1.4) & 14.63 (2.52) & 3.45 (1.56) & 33.1 (1.64) & 3.7 (1.51) \\ 
		&       $||\hat{\gamma}_0-\gamma_0||_{2}$   & 1.6 (0.91) & 11.25 (2.99) & 2.13 (1.09) & 31.63 (8.17) & 2.04 (0.82) \\ 
		&       $||\hat{\gamma}_0-\gamma_0||_{\infty}$   & 2.36 (1.45) & 25.39 (4) & 3.38 (1.43) & 48.74 (5.25) & 3.7 (0.91) \\ 
		\rule{0pt}{4.5ex} 
		Model 2 & Bias($C_0(0)$)    & 1.25 (2.09) & -0.05 (2.1) & 3.64 (2.76) & -2.09 (1.59) & 1.88 (2.29) \\ 
		Cauchy               & $||\hat{R}_0-R_0||_2$  & 9.72 (0.5) & 1.06 (0.92) & 17.16 (0.66) & 15.22 (2.26) & 10.8 (0.51) \\ 
		&     $||\hat{R}_0-R_0||_{\infty}$  & 24.38 (1.75) & 2.48 (2.14) & 34.86 (1.63) & 35.09 (2.72) & 26.02 (1.71) \\ 
		&      $||\hat{C}_0-C_0||_2$   & 9.56 (0.68) & 1.31 (1.08) & 16.43 (0.99) & 14.4 (2.49) & 10.54 (0.71) \\ 
		&       $||\hat{C}_0-C_0||_{\infty}$     & 23.82 (2.35) & 3.4 (2.43) & 33.22 (2.45) & 33.31 (3.42) & 25.16 (2.39) \\ 
		&      $||\hat{\gamma}_0-\gamma_0||_{2}$   & 9.83 (0.54) & 1.64 (0.94) & 17.83 (0.79) & 15.82 (2.23) & 10.96 (0.56) \\ 
		&        $||\hat{\gamma}_0-\gamma_0||_{\infty}$   & 25.07 (1.59) & 2.76 (1.72) & 36.86 (1.75) & 35.41 (2.63) & 27.04 (1.62) \\ 
		\rule{0pt}{4.5ex} 
		Model 3 & Bias($C_0(0)$)     & -4.99 (1.86) & -16.38 (2.85) & 0.11 (2.77) & -6.22 (7.84) & -4.73 (2.13) \\ 
		Gaussian               & $||\hat{R}_0-R_0||_2$   & 5.61 (0.46) & 12.87 (1.59) & 0.8 (0.68) & 22.13 (6.51) & 4.85 (0.47) \\ 
		&     $||\hat{R}_0-R_0||_{\infty}$ & 14.44 (1.45) & 32.66 (3.94) & 1.39 (1.18) & 43.1 (5.09) & 11.89 (1.4) \\ 
		&    $||\hat{C}_0-C_0||_2$   & 4.47 (0.56) & 9.61 (0.47) & 1.55 (0.96) & 19.63 (0.88) & 3.8 (0.58) \\ 
		&      $||\hat{C}_0-C_0||_{\infty}$    & 10.84 (2.1) & 20.12 (2.42) & 2.78 (1.65) & 38.39 (3.91) & 8.9 (1.9) \\ 
		&      $||\hat{\gamma}_0-\gamma_0||_{2}$   & 6.62 (0.93) & 15.39 (3.55) & 1.82 (1.13) & 20.91 (7.42) & 6.01 (1.06) \\ 
		&      $||\hat{\gamma}_0-\gamma_0||_{\infty}$   & 15.67 (1.44) & 36.28 (4.32) & 2.71 (1.54) & 43.79 (5.6) & 13.15 (1.42) \\ 
		\rule{0pt}{4.5ex} 
		Model 4 & Bias($C_0(0)$)    & 0.05 (2.01) & 0.09 (2.23) & 0.15 (2.54) & -0.07 (2) & 0.07 (2.18) \\ 
		GenCauchy               & $||\hat{R}_0-R_0||_2$   & 21.65 (0.71) & 9.42 (0.88) & 30.18 (1) & 1.92 (1.61) & 22.51 (0.74) \\ 
		&     $||\hat{R}_0-R_0||_{\infty}$   & 60.73 (2.14) & 35.02 (3.85) & 61.69 (2.03) & 3.88 (3.25) & 59.24 (2.12) \\ 
		&     $||\hat{C}_0-C_0||_2$   & 21.65 (0.79) & 9.44 (1.02) & 30.16 (1.18) & 2.24 (1.92) & 22.5 (0.85) \\ 
		&      $||\hat{C}_0-C_0||_{\infty}$   & 60.71 (2.39) & 34.94 (4.58) & 61.63 (2.39) & 4.76 (3.81) & 59.21 (2.41) \\ 
		&      $||\hat{\gamma}_0-\gamma_0||_{2}$   & 21.71 (0.73) & 9.47 (0.87) & 30.27 (1.11) & 1.77 (1.1) & 22.58 (0.77) \\ 
		&       $||\hat{\gamma}_0-\gamma_0||_{\infty}$   & 60.76 (2.15) & 35.03 (3.38) & 61.78 (2.24) & 3.78 (2.82) & 59.28 (2.2) \\ 
		\hline\hline
	\end{tabular}  
	\\
	\vspace{1ex}
	{\raggedright \scriptsize{Note: The entries in the table show the means and standard deviations, in parenthesis, of scaled 2-norm and inf-norm errors of correlation, covariance and semivariance estimates as well as bias of $\hat{C}_0(0)$ based on 100 mc simulations. Columns correspond to true models and rows correspond to estimation methods. The entries are in units of $10^{-2}$.}\par}
	\label{table:simulation of average for cov-based parametric methods using modified weights}
\end{table}

\newpage
\subsection{Additional Results of Section 5}
\label{sec:additional results of section 5}
\begin{table}[h]
	\centering
	\scriptsize
	\caption{Case study results of parametric methods.}
	\begin{tabular}{clcccc}
		\hline\hline
		&	                               &  Matern &    Cauchy &  Gaussian & GenCauchy  \\ \hline
		likelihood-
		& nugget($inch^2$) & 9.45 & 9.38 & 14.23 & 8.2  \\ 
		based 	& $C_0(0)$($inch^2$) & 42.98 & 43.82 & 35.25 & 56.71 \\
		& log-likelihood/(rn) &-0.709 & -0.710 &-0.731 & -0.709 \\
		& range(100km) &  1.69 & 1.26 & 2.55 & 0.86  \\ 
		& other parameters & 1 & - & - & (0.5,2)\\ 
		& time(mins)& 15.764 & 7.166 & 7.192 & 12.935  \\
		\rule{0pt}{5ex}    
		variogram-
		&  	nugget($inch^2$) & 13.47 & 10.43 & 16.42 & 8.52  \\ 
		based 		& $C_0(0)$($inch^2$) & 45.89 & 57.07 & 42.45 & 72.25  \\ 
		& log-likelihood/(rn) & -0.740 & -0.725 & -0.775 & -0.720\\
		& range(100km)  &  1.89 & 2.24 & 5.16 & 1.35 \\     
		&  other parameters   &  2   &   -   &  -    & (0.5,2) \\
		& time(mins) & 0.04 & 0.05 & 0.05 & 0.05 \\
		\rule{0pt}{5ex}    
		cov-	& nugget($inch^2$)  & 21.42 & 0 & 22.37 & 0  \\ 
		based	&$C_0(0)$ ($inch^2$) &  30.92 & 53.74 & 29.98 & 77.83 \\ 
		&  log-likelihood/(rn) & -0.749 & -1.085 & -0.756 & -1.016\\
		&	range(100km) & 0.63 & 0.46 & 4.02 & 0.01 \\ 
		&  other parameters & 10 & - & - & (0.5,2) \\ 
		& time(mins) & 4.44 & 2.88 & 1.62 & 8.18 \\ 
		\hline\hline
	\end{tabular}\\
	\vspace{1ex}
	{\raggedright \scriptsize{Note: Columns correspond to parametric model and rows specifies the estimation method. The entries ’other parameters‘ refer to the smoothness parameter $\nu$ for Matern and $\kappa_1,\kappa_2$ for GenCauchy. The entries 'range' refer to $\rho$ of four parametric models. }\par}
	\label{table:real data results of parametric models}
\end{table}

\newpage
\section{Figures}

\subsection{Additional Figures in Section 4}
\begin{figure}[H]
	\centering
	
	\begin{subfigure}{0.19\textwidth}
		\centering
		\includegraphics[width=\linewidth,height=4.8cm]{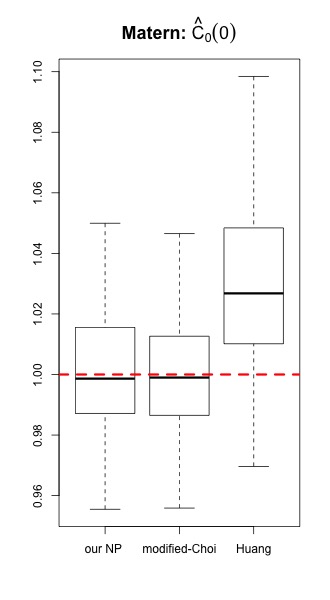}
		\caption{}
		\label{fig: simulation boxplot of Matern C_zero}
	\end{subfigure}
	\begin{subfigure}{0.19\textwidth}
		\centering
		\includegraphics[width=\linewidth,height=4.8cm]{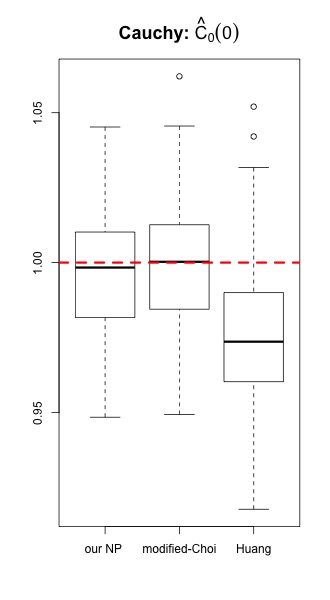}
		\caption{}
		\label{fig: simulation boxplot of Cauchy C_zero}
	\end{subfigure}
	\begin{subfigure}{0.19\textwidth}
		\centering
		\includegraphics[width=\linewidth,height=4.8cm]{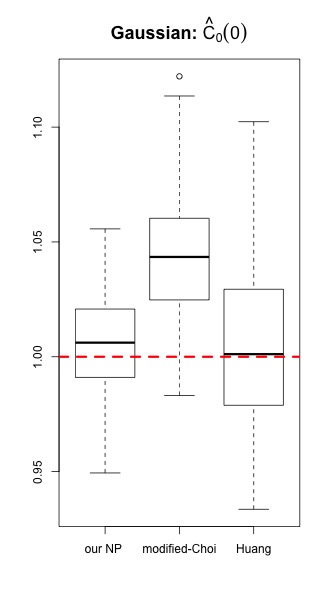}
		\caption{}
		\label{fig: simulation boxplot of Gaussian C_zero}
	\end{subfigure}
	\begin{subfigure}{0.19\textwidth}
		\centering
		\includegraphics[width=\linewidth,height=4.8cm]{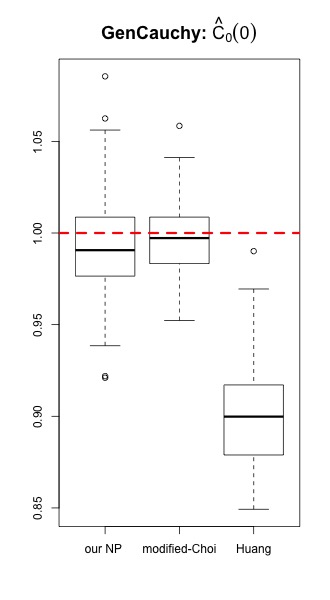}
		\caption{}
		\label{fig: simulation boxplot of GenCauchy C_zero}
	\end{subfigure}
	\begin{subfigure}{0.19\textwidth}
		\centering
		\includegraphics[width=\linewidth,height=4.8cm]{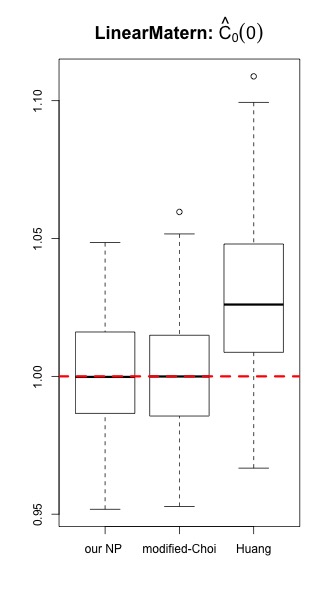}
		\caption{}
		\label{fig: simulation boxplot of LinearMatern C_zero}
	\end{subfigure}
	
	\vspace{4pt}
	
	\begin{subfigure}{0.19\textwidth}
		\centering
		\includegraphics[width=\linewidth,height=4.8cm]{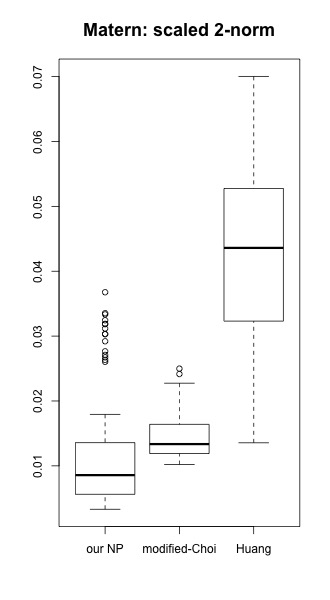}
		\caption{}
		\label{fig: simulation boxplot of Matern norm2_cor}
	\end{subfigure}
	\begin{subfigure}{0.19\textwidth}
		\centering
		\includegraphics[width=\linewidth,height=4.8cm]{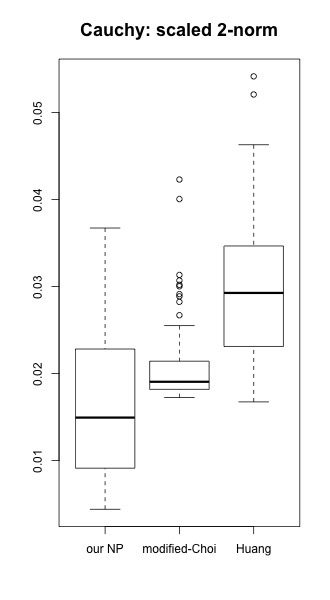}
		\caption{}
		\label{fig: simulation boxplot of Cauchy norm2_cor}
	\end{subfigure}
	\begin{subfigure}{0.19\textwidth}
		\centering
		\includegraphics[width=\linewidth,height=4.8cm]{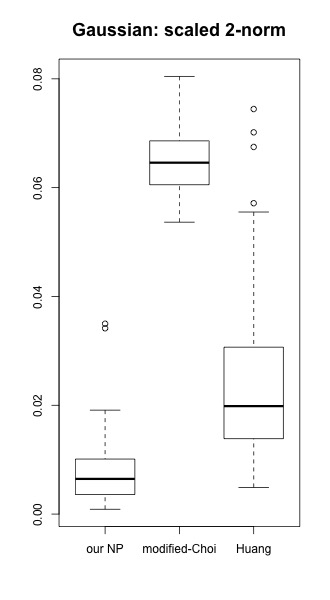}
		\caption{}
		\label{fig: simulation boxplot of Gaussian norm2_cor}
	\end{subfigure}
	\begin{subfigure}{0.19\textwidth}
		\centering
		\includegraphics[width=\linewidth,height=4.8cm]{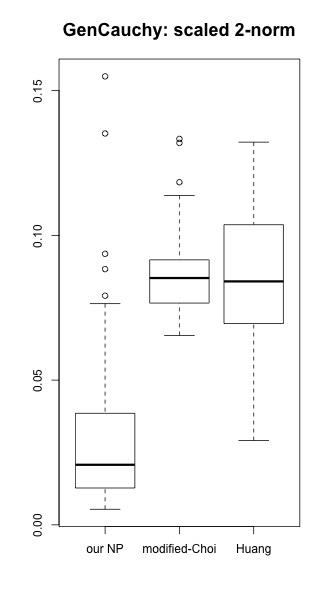}
		\caption{}
		\label{fig: simulation boxplot of GenCauchy norm2_cor}
	\end{subfigure}
	\begin{subfigure}{0.19\textwidth}
		\centering
		\includegraphics[width=\linewidth,height=4.8cm]{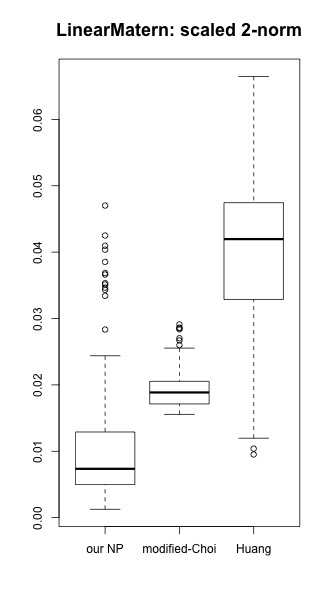}
		\caption{}
		\label{fig: simulation boxplot of LinearMatern norm2_cor}
	\end{subfigure}
	
	\vspace{4pt}
	
	\begin{subfigure}{0.19\textwidth}
		\centering
		\includegraphics[width=\linewidth,height=4.8cm]{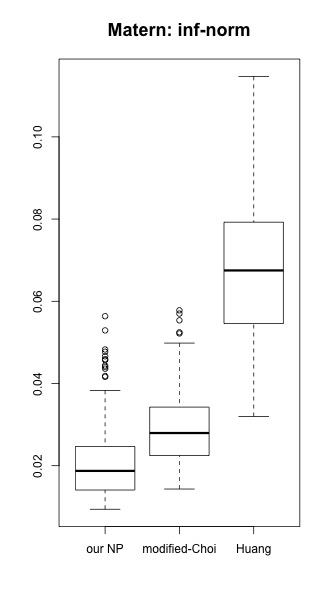}
		\caption{}
		\label{fig: simulation boxplot of Matern supnorm_cor}
	\end{subfigure}
	\begin{subfigure}{0.19\textwidth}
		\centering
		\includegraphics[width=\linewidth,height=4.8cm]{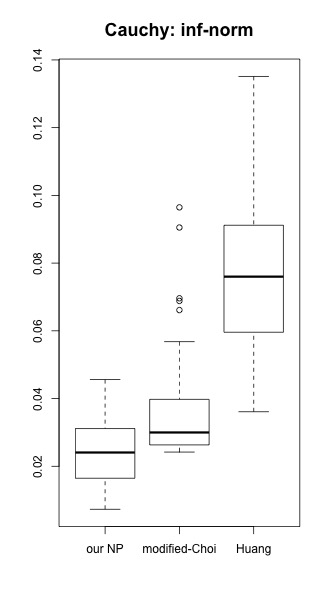}
		\caption{}
		\label{fig: simulation boxplot of Cauchy supnorm_cor}
	\end{subfigure}
	\begin{subfigure}{0.19\textwidth}
		\centering
		\includegraphics[width=\linewidth,height=4.8cm]{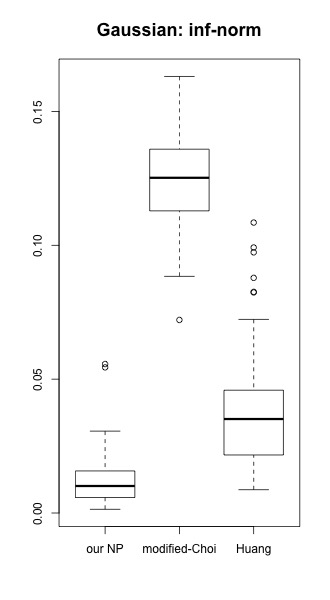}
		\caption{}
		\label{fig: simulation boxplot of Gaussian supnorm_cor}
	\end{subfigure}
	\begin{subfigure}{0.19\textwidth}
		\centering
		\includegraphics[width=\linewidth,height=4.8cm]{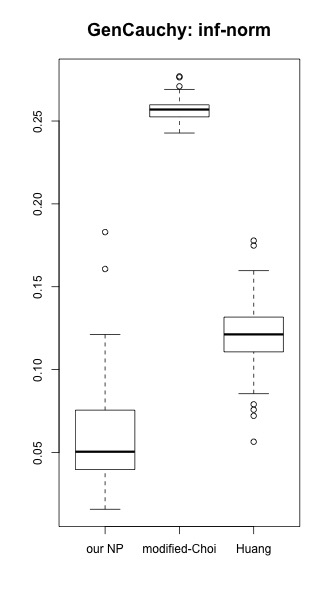}
		\caption{}
		\label{fig: simulation boxplot of GenCauchy supnorm_cor}
	\end{subfigure}
	\begin{subfigure}{0.19\textwidth}
		\centering
		\includegraphics[width=\linewidth,height=4.8cm]{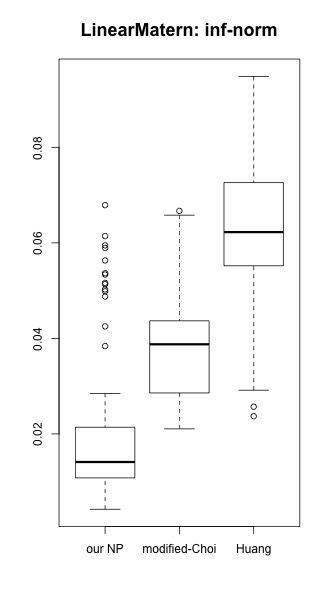}
		\caption{}
		\label{fig: simulation boxplot of LinearMatern supnorm_cor}
	\end{subfigure}
	
	\caption{Boxplots of estimation results based on 100 MCs. Each row includes boxplots for Settings 1--5. The first through third rows show results for $\hat{C}_0(0)$ in (a)--(e), the scaled 2-norm in (f)--(j), and the infinity-norm error in (k)--(o) of correlation estimation. In each subfigure, boxplots from left to right correspond to our nonparametric method (our NP), the modified Choi method, and the Huang method.}
	\label{fig:simulation of boxplot}
\end{figure}


\end{document}